\documentclass[10pt,a4paper,reqno]{amsart}
\usepackage{amsthm,amsfonts,amssymb,amsmath,euscript,comment, enumerate,slashed, marvosym, pifont}
\usepackage{graphicx}
\usepackage{color}
\usepackage{bbm}
\usepackage{geometry}
\usepackage{amsfonts}
\usepackage{xcolor}
\usepackage{hyperref}
\usepackage{seqsplit}
\usepackage{mathtools}

\newcommand{\pfstep}[1]{\smallskip \noindent {\it #1.}}

\newtheorem{thm}{Theorem}[section]

\newtheorem{corollary}[thm]{Corollary}

\newtheorem{df}[thm]{Definition}
\newtheorem{rk}[thm]{Remark}
\newtheorem{definition}[thm]{Definition}
\newtheorem{conjecture}[thm]{Conjecture}
\newtheorem{proposition}[thm]{Proposition}
\newtheorem{lemma}[thm]{Lemma}
\newtheorem{theorem}[thm]{Theorem}
\newtheorem{assumptions}[thm]{Assumptions}

\newcommand{\m}[1]{\mathbb{#1}}

\newcommand{\ep}{\varepsilon}

\newcommand{\f}{\frac}
\newcommand{\rd}{\partial}

\newcommand{\alp}{\alpha}
\newcommand{\bt}{\beta}

\newcommand{\ls}{\lesssim}
\newcommand{\de}{\delta}

\newcommand{\la}{\langle}
\newcommand{\ra}{\rangle}

\newcommand{\ud}{\mathrm{d}}

\newcommand{\RR}{\mathbb R}
\newcommand{\mfh}{\mathfrak h}
\newcommand{\mfk}{\mathfrak k}
\def \i {\infty}
\def \th {\theta}

\newcommand{\Lxi}{L^{(\xi)}}
\newcommand{\Lbxi}{\underline{L}^{(\xi)}}

\newcommand{\uxi}{\underline{\xi}}
\newcommand{\ux}{\underline{x}}
\newcommand{\ueta}{\underline{\eta}}
\newcommand{\uth}{\underline{\th}}

\newcommand{\calF}{\mathcal F}
\newcommand{\calG}{\mathcal G}
\newcommand{\calH}{\mathcal H}
\newcommand{\calR}{\mathcal R}
\newcommand{\ucalF}{\underline{\calF}}

\newcommand{\srd}{\slashed{\rd}}

\newcommand{\brk}{\@ifstar{\brkb}{\brki}}
\newcommand{\brki}[1]{\langle{#1}\rangle}
\newcommand{\brkb}[1]{\left\langle{#1}\right\rangle}

\numberwithin{equation}{section}

\begin{document}

\title{Burnett's conjecture in generalized wave coordinates}

\begin{abstract} 
We prove Burnett's conjecture in general relativity when the metrics satisfy a generalized wave coordinate condition, i.e., suppose $\{g_n\}_{n=1}^\infty$ is a sequence of Lorentzian metrics (in arbitrary dimensions $d \geq 3$) satisfying a generalized wave coordinate condition and such that $g_n\to g$ in a suitably weak and ``high-frequency'' manner, then the limit metric $g$ satisfies the Einstein--massless Vlasov system. Moreover, we show that the Vlasov field for the limiting metric can be taken to be a suitable microlocal defect measure corresponding to the convergence. The proof uses a compensation phenomenon based on the linear and nonlinear structure of the Einstein equations.
\end{abstract}

\author{C\'ecile Huneau}
\address{CMLS, Ecole Polytechnique, 91120 Palaiseau, France}
\email{cecile.huneau@polytechnique.edu}
\author{Jonathan Luk}
\address{Department of Mathematics, Stanford University, CA 94304, USA}
\email{jluk@stanford.edu}

\maketitle

\section{Introduction}

It is well-known that the ``high-frequency'' limit $g$ of a sequence of metrics $\{g_n\}_{n=1}^\infty$ to the Einstein vacuum equation $\mathrm{Ric}(g_n) = 0$ may not be itself a solution to the Einstein vacuum equations. Nonetheless, Burnett conjectured that such limits must still be of a very specific type, where the ``effective matter'' that arises in the limiting process takes the form of a massless Vlasov field:
\begin{conjecture}[Burnett \cite{Burnett}]\label{conj:Burnett}
If $g$ is a suitable ``high-frequency limit'' of a sequence of vacuum metrics $\{g_n\}_{n=1}^\infty$, then $g$ solves the Einstein--massless Vlasov system (after defining a suitable Vlasov field).
\end{conjecture}
As described in \cite{Burnett}, the conjecture, if true, shows in particular that in the limit, high-frequency gravitational waves do not interact directly, but they only interact via the Einstein equation through their influence on the geometry.

\textbf{In this paper, we prove Conjecture~\ref{conj:Burnett} under the additional assumption that
$\{g_n\}_{n=1}^\infty$ and $g_0$ satisfy suitable (generalized) wave coordinate conditions.} 

The Burnett conjecture was previously proven in an elliptic gauge when the metric admits a $\m U(1)$ symmetry \cite{HL.Burnett}, and in the double null coordinates gauge when the metric is angularly regular \cite{LR.HF}. This is the first proof of Conjecture~\ref{conj:Burnett} without symmetry assumptions or angular regularity assumptions. 

To discuss our result, we introduce the main assumptions of our theorem.
\begin{assumptions}\label{ass:main}
Let $d \geq 3$. Consider a sequence of $C^\infty$ Lorentzian metrics $\{g_n\}_{n=1}^\infty$ on a bounded open set $U\subset \mathbb R^{d+1}$ and suppose that the following holds:
\begin{enumerate}
\item \textbf{The Einstein vacuum equation holds for $\{g_n\}_{n=1}^\infty$:} $\mathrm{Ric}(g_n) = 0$, $\forall n \in \mathbb N$.
\item \textbf{$g_0$ is a uniform limit:} There exists a $C^\infty$ Lorentzian metric $g_0$ and a decreasing sequence $\{\lambda_n \}_{n=1}^\infty\subset (0,1)$ with $\lim_{n\to \infty} \lambda_n = 0$ such that 
\begin{equation}\label{eq:ass.1}
|g_n - g_0|\leq \lambda_n.
\end{equation}
\item \textbf{$g_0$ is a high-frequency limit:} the following bounds hold: 
\begin{equation}
|\partial g_n| \ls 1,\quad  |\rd^2 g_n| \ls \lambda_n^{-1},
\end{equation}\label{eq:ass.2}
where the implicit constants are independent of $n$. 
\item \textbf{The generalized wave coordinate condition holds for $\{g_n\}_{n=1}^\infty$ and $g_0$:} Defining 
\begin{equation}\label{eq:generalized.wave.def}
H^\alp_n \doteq (g_{n}^{-1})^{\mu\nu} \Gamma_{\mu\nu}^\alp(g_n),\quad H^\alp_0 \doteq (g_{0}^{-1})^{\mu\nu} \Gamma_{\mu\nu}^\alp(g_0),
\end{equation}
where $\Gamma_{\mu\nu}^\alp(g_n)$ and $\Gamma_{\mu\nu}^\alp(g_0)$ denote the Christoffel symbols of $g_n$ and $g_0$, respectively, there exists $\eta \in (0,1]$ such that
\begin{equation}\label{eq:wave.coord.assumption}
|H^\alp_n - H^\alp_0| \ls \lambda_n^{\eta}, \quad |\rd H^\alp_n| \ls 1,\quad \forall n \in \mathbb N,\,\forall \alp = 0,1,\cdots d,
\end{equation}
where the implicit constants are independent of $n$. 
\end{enumerate}
\end{assumptions}

\begin{rk}\label{rmk:improved.assumption}
We remark that instead of (3), the original conjecture in \cite{Burnett} only requires the bound $\sup_n |\partial g_n| < \infty$. Our stronger assumption (3) which concerns also the second derivative is imposed in the spirit of \cite{Isaacson1}, which fixes the frequency scale. 

\begin{enumerate}
\item It is not difficult to see, because all the arguments in the proof of the main theorem have room, that there exists an explicitly computable $\de>0$ such that the assumption on the second derivatives of $g_n$ can be replaced by $|\rd^2 g_n| \ls \lambda_n^{-1-\de}$. (At the same time, the condition $|\rd H^\alp_n| \ls 1$ can also be relaxed to $|\rd H^\alp_n| \ls \lambda_n^{-\de}$, with $\de = \de(\eta) >0$.) The advantage for the assumptions as stated in Assumption~\ref{ass:main} is they imply that $\Box_{g_0}(g_n-g_0)$ is bounded (see equation \eqref{eq:h.eqn}), which is not necessary but makes the exposition cleaner.
\item As one will see in the proof, in the main theorem (Theorem~\ref{thm:main}) below, the assumption on the second derivatives of $g_n$ is only needed in part (4), i.e., the Vlasov equation for $\mu$. If we were only interested in the statements (1)--(3) of Theorem~\ref{thm:main}, then the assumption $|\rd^2 g_n| \ls \lambda_n^{-1}$ (as well as the assumption on the uniform boundedness of $\rd H_n^\alp$) can be removed.
\end{enumerate}

However, it remains an open problem to prove (or disprove) Burnett's conjecture in generalized wave coordinates without any assumptions on the second derivatives of $g_n$.
\end{rk}



We need a few definitions to introduce our main theorem.

\begin{definition}\label{def:measure}
\begin{enumerate}
\item Let $T^*U \equiv U\times \mathbb R^{d+1}$ denote the cotangent bundle of $U$ and let $S^*U = U \times ((\mathbb R^{d+1}\setminus\{0\})/\sim)$ denote the cosphere bundle, where $\xi \sim \eta$ if and only if $\xi = \alp \eta$ for some $\alp >0$. 
\item For $k \in \mathbb N \cup \{0\}$, we say that $f: U \times (\mathbb R^{d+1}\setminus\{0\}) \to \mathbb R$ is positively $k$-homogeneous if $f(x,\alp \xi) = \alp^k f(x,\xi)$ for all $\alp \geq 1$ and all $(x,\xi) \in U \times (\mathbb R^{d+1}\setminus\{0\})$ with $|\xi|\geq 1$. For any fixed $k \in \mathbb N\cup \{0\}$, we can define a Radon measure on $S^* U$ by its action on positively $k$-homogeneous continuous functions $f: U \times (\mathbb R^{d+1}\setminus\{0\}) \to \mathbb R$.
\end{enumerate}
\end{definition}

The main result of the paper is the following theorem:
\begin{theorem}\label{thm:main}
Burnett's conjecture (Conjecture~\ref{conj:Burnett}) is true under Assumption~\ref{ass:main}.

More precisely, under (1)--(4) in Assumption~\ref{ass:main}, define $\mu$ by
\begin{equation}\label{eq:mu.def}
\mu= g_0^{\alpha\rho}g_0^{\beta \sigma}(\frac{1}{4}\mu_{\rho \beta \alpha \sigma} -\frac{1}{8}\mu_{\rho \alpha \beta \sigma}),
\end{equation}
where $\mu_{\alp\bt\sigma\rho}$ is a Radon measure\footnote{Here, the measure acts on continuous functions which are positively $2$-homogeneous in $\xi$; see Definition~\ref{def:measure}.} on $S^* U$ defined\footnote{The well-definedness of $\mu_{\alp\bt\sigma\rho}$ follows from the works of Tartar \cite{Tartar} and Gerald \cite{Gerard}; see Remark~\ref{rmk:existence.of.mu} in Section~\ref{sec:mdm}.} so that for $h = g_n - g_0$ and after passing to a subsequence (still labelled by $n$), the following holds for all $0$-th order pseudo-differential operator $A$ with principal symbol $a(x,\xi)$ which is real and $0$-positively homogeneous in $\xi$:
$$\lim_{n \to \infty} \la \rd_\gamma h_{\alp\bt}, A \rd_\de h_{\rho\sigma} \ra_{L^2(\mathrm{dVol}_{g_0})} = \int_{S^* U} a \xi_\gamma \xi_\de \, \ud \mu_{\alp\bt\rho\sigma}.$$

Then $(U,g_0,\mu)$ satisfies the Einstein--massless Vlasov system, namely, the following all hold: 
\begin{enumerate}
\item ($\mu$ is a massless field, i.e., $\mu$ is supported on the light cone) $$g_0^{\mu\nu} \xi_\mu\xi_\nu \, \ud \mu \equiv 0.$$
\item (The Einstein equation holds) $$\int_{U} \psi \mathrm{Ric}_{\mu \nu}(g_0) \, \mathrm{dVol}_{g_0} = \int_{S^* U} \psi \xi_\mu \xi_\nu \ud \mu,\quad \forall \psi \in C^\infty_c(U).$$
\item ($\mu$ is real-valued and non-negative) $$\int_{S^* U} f \, \ud \mu \geq 0\quad \hbox{$\forall$ continuous positively $2$-homogeneous $f: S^* U\to [0,\infty)$}. $$
\item ($\mu$ satisfies the Vlasov equation) For any smooth $\widetilde{a}:S^*U \to \mathbb R$ which is spatially compactly supported in $U$ and positively $1$-homogeneous in $\xi$:
\begin{equation}\label{eq:transport.in.thm}
\int_{S^*U} \{g_0^{\mu\nu} \xi_\mu\xi_\nu, \widetilde{a}(x,\xi) \}  \, \ud \mu = 0,
\end{equation}
where $\{ \cdot, \cdot \}$ denotes the Poisson bracket $\{f,h\} \doteq \rd_{\xi_\mu} f \rd_{x^\mu} h - \rd_{x^\mu} f \rd_{\xi_\mu} h.$
\end{enumerate}
\end{theorem}

The proof of the four parts of Theorem~\ref{thm:main} can be found in Proposition~\ref{prop:supp.null.cone}, Proposition~\ref{prop:Ricci.limit}, Proposition~\ref{prop:positive} and Theorem~\ref{thm:main.transport}, respectively.

\begin{rk}
Theorem~\ref{thm:main} includes as a special case when the classical wave coordinate condition is satisfied, i.e., when $H_n, \, H_0 \equiv 0$. We remark that one motivation to consider generalized wave coordinates instead of just wave coordinates is that the former seems to come up naturally in the construction of high-frequency limits; see \cite{Touati2, touati2024reverse}. 
\end{rk}

\begin{rk}
Notice that Theorem~\ref{thm:main} is more specific than Conjecture~\ref{conj:Burnett} in the sense that we give an explicit construction of the massless Vlasov field as the microlocal defect measure.
\end{rk}

\subsection{Ideas of the proof}

The broad strategy follows our previous work \cite{HL.Burnett}: we first show that the failure of strong convergence can be captured by a suitably defined microlocal defect measure and then show that this microlocal defect measure has the necessary properties, including satisfying the transport equation in \eqref{eq:transport.in.thm}. The proof of the transport equation is the hardest part and can be viewed as a question about compensated compactness. It is well known that the microlocal defect measure associated to a sequence of solutions to any \emph{linear} wave equation (with regularity consistent with our main theorem) must satisfy \eqref{eq:transport.in.thm}; see for instance \cite{Francfort, FrancfortMurat, Tartar}. Thus the main difficulty in our case is to show that there is a compensation phenomenon that allows for the microlocal defect measure to still satisfy \eqref{eq:transport.in.thm} even in the presence of additional linear and nonlinear terms in the Einstein equations.

\subsubsection{Ricci curvature of the limit}
Our work relies heavily on the structure of the Einstein vacuum equations in generalized wave coordinates. As is known from the work of Lindblad--Rodnianski \cite{LinRod.WN}, the nonlinearity of the Einstein vacuum equations in generalized wave coordinates contains both terms that obey the classical null condition and terms that violate the classical null condition. Denoting $h = g_n - g_0$, we have (see \eqref{eq:h.eqn})
\begin{equation}\label{eq:h.eqn.intro}
\begin{split}
-2 \mathrm{Ric}(g_0)_{\mu\nu} = &\: \Box_{g_0} h_{\mu\nu} + P_{\mu\nu}(g_0)(\rd h,\rd h)+ g_0^{\alp\alp'}g_0^{\bt\bt'} h_{\alp'\bt'} \rd^2_{\alp\bt} h_{\mu\nu} \\
&\: + \hbox{linear terms} + \hbox{null condition terms} + \hbox{better terms},
\end{split}
\end{equation}
where $P_{\mu\nu}$ is a quadratic nonlinearity violating the classical null condition; see \eqref{eq:P.def}. 

To understand whether $\mathrm{Ric}(g_0)_{\mu\nu} = 0$, we need to understand the weak limit of the right-hand side of \eqref{eq:h.eqn.intro} as $n\to \infty$. The linear terms obviously converge weakly to $0$; the null condition terms also converge weakly to $0$ due to compensation compactness. It turns out that the main quasilinear term $g_0^{\alp\alp'}g_0^{\bt\bt'} h_{\alp'\bt'} \rd^2_{\alp\bt} h_{\mu\nu}$ has a hidden null structure which makes it also converge weakly to $0$ (see \eqref{eq:quasilinear.hidden.null.structure}). Thus the only term that does not converge weakly to $0$ is $P_{\mu\nu}(g_0)(\rd h,\rd h)$. The microlocal defect measure $\mu$ is defined so that we exactly have
$$\int_U \psi \mathrm{Ric}_{\alp\bt}(g_0) \, \mathrm{dVol}_{g_0} = - \f 12 \int_U \psi P_{\alp\bt}(g_0)(\rd h, \rd h) \, \mathrm{dVol}_{g_0} = \int_{S^*U} \psi \xi_\alp \xi_\bt \, \ud \mu.$$

\subsubsection{Non-negativity of the measure $\mu$}
It is not difficult to show, essentially using boundedness of $\Box_{g_0} h$ that $\mu$ is supported on the null cone. However, $\mu$ defined in Theorem~\ref{thm:main} does not seem to be obviously non-negative. In order to show this, we introduce a $\xi$-dependent null frame and only reveal the non-negativity after expanding the metric with respect to this frame. The computations are inspired in part by \cite[(3.17)]{Touati2}. See Proposition~\ref{prop:positive} for details.

\subsubsection{Propagation equation for $\mu$}
We now turn to ideas for proving that the transport equation \eqref{eq:transport.in.thm} holds for $\mu$. As pointed out above, since we know that microlocal defect measures associated to linear wave equations obey \eqref{eq:transport.in.thm}, we need to be able to deal with the extra terms in \eqref{eq:h.eqn.intro}, now viewed as an equation for $\Box_{g_0} h$. We have three main challenges: the linear terms, the semilinear terms and the quasilinear terms in \eqref{eq:h.eqn.intro}. 

For the linear terms, we need an \emph{exact} cancellation, which indeed holds true and can be verified with explicit computations. 

For the nonlinear (semilinear and quasilinear) terms, we need to control trilinear terms of the type.
\begin{equation}\label{eq:intro.trilinear}
\la \rd h, A (\rd h \rd h) \ra,
\end{equation}
where $A$ is a $0$-th order pseudo-differential operator. In general, terms of the form \eqref{eq:intro.trilinear} need not converge to $0$. However, it turns out that if there is a null form, e.g., if $\mathfrak Q$ is one of the classical null forms, then using also that $h$ satisfies wave equations, it can be shown that $\la \rd h, A (\mathfrak Q(\rd h, \rd h)) \ra \to 0$. After reducing to a constant coefficient problem by freezing coefficients, we prove this with the help of a trilinear estimate by Ionescu--Pasauder \cite{IP}; see Proposition~\ref{prop:main.trilinear}. In particular, this allows us to treat the null condition terms in \eqref{eq:h.eqn.intro}.

However, the $P$ term in \eqref{eq:h.eqn.intro} violates the classical null condition. Nonetheless, there is a much more subtle trilinear structure, meaning that after suitable algebraic manipulations using the generalized wave coordinate condition, one can reveal a hidden null structure and can deal with this term using ideas in the paragraph above. We notice that this structure is present only when we consider the propagation equation for $\mu$ in \eqref{eq:mu.def},  and does not hold for the individual $\mu_{\alp\bt\rho\sigma}$!

An additional interesting challenge comes from the quasilinear terms $g_0^{\mu\mu'} g_0^{\nu\nu'} h_{\mu\nu} \rd_{\mu'\nu'}^2 h_{\alp'\bt'}$. Here, we need to decompose $h_{\mu\nu}$ into $h_{\mu\nu} = \sum_{i=1}^3 \mfh^{(i)}_{\mu\nu}$ according to their frequencies (see precise decomposition in Definition~\ref{def:freq.decomposition}). The term $\mfh^{(1)}_{\mu\nu}$ has low frequency, so $\rd \mfh^{(1)}_{\mu\nu}$ is $o(1)$. This can then be exploited using the Calder\'on commutator theorem in a similar manner as \cite{HL.Burnett}. The term $\mfh^{(2)}_{\mu\nu}$ has high frequency, but the frequency lives away from the light cone of $g_0$. In this frequency regime, the wave operator $\Box_{g_0}$ is elliptic, and thus the control of $\Box_{g_0} h$ that we get from the equation gives very strong compactness. Finally, the term $\mfh^{(3)}_{\mu\nu}$ has high frequency with the frequency possibly close to the light cone. Here is the most delicate case, but it can be dealt with after noting that there is a secret trilinear null structure if we use the generalized wave coordinate condition. The structure is a bit similar to the $P$ terms discussed above.

\subsection{Related works}

We discuss a short sample of related works. We refer the reader to our recent survey \cite{HL.survey} for further references. 

\subsubsection{Physics literature}\label{sec:physics}

The study of high-frequency spacetimes with the type of scaling we consider here was initiated in the works of Isaacson \cite{Isaacson1, Isaacson2}. We refer the reader also to \cite{MTW} for related discussions, There were various construction of high-frequency approximate solutions; see for instance \cite{AliHunt, CB.HF, MacCallumTaub}. There were also constructions of explicit solutions in symmetry classes; see for instance \cite{GW2, pHtF93, SGWK, SW}. 
In the work \cite{Burnett}, Burnett reformulated this type of considerations in terms of weak limits and proposed Conjecture~\ref{conj:Burnett}.  Finally, we mention \cite{GW1}, which discusses the implications of this type of high-frequency limits in the context of inhomogeneities in cosmology. 

\subsubsection{Related results on Burnett's conjecture} The first result in the spirit of Conjecture~\ref{conj:Burnett} is the work of Green--Wald \cite{GW1} which shows that under the assumptions of Conjecture~\ref{conj:Burnett}, the limiting stress-energy-momentum tensor is traceless and satisfies the weak energy condition. 

Under a codimensional $2$ symmetry assumption, compactness questions as in Conjecture~\ref{conj:Burnett} was studied in \cite{LeFLeF2019, LeFLeF2020}. More generally, a complete characterization of high-frequency limits was obtained by Luk--Rodnianski in \cite{LR.HF} under angular regularity assumptions, where they consider a setting where the manifold is given by $U^2 \times S^{2}$, but instead of exact symmetries along $S^{2}$ they only required high regularity along directions of $S^{2}$.

Beyond the angularly regular settings, we proved Conjecture~\ref{conj:Burnett} in \cite{HL.Burnett} for $\mathbb U(1)$ symmetric solutions under an elliptic gauge condition. In particular, we gave a formulation of Conjecture~\ref{conj:Burnett} in terms of microlocal defect measures of \cite{Gerard, Tartar}; it is also within the same framework that we discuss Conjecture~\ref{conj:Burnett} here. See also the subsequent work \cite{GuerraTeixeira} by Guerra--Teixeira da Costa in the same setting with a slightly different treatment of the time-dominated frequency regime.  

\subsubsection{Construction of high-frequency limits} 
There is a dual problem to Conjecture~\ref{conj:Burnett}, namely to construct sequences of high-frequency vacuum spacetimes whose limit solves the Einstein--massless Vlasov system with a non-trivial Vlasov field. Beyond the explicit solutions mentioned in Section~\ref{sec:physics}, the first constructions were given in our \cite{HL.HF} where we imposed a polarized $\mathbb U(1)$ symmetry (see also our forthcoming \cite{HL.Vlasov}). Other constructions include those in symmetry classes \cite{Lott1, Lott3, Lott2, LeFLeF2019, LeFLeF2020}, as well as much more generally, constructions which require regularity along angular directions but without exact symmetries \cite{LR.HF}.

We particularly draw attention to the works of Touati \cite{Touati2, touati2024reverse} (see also \cite{Touaticontraintes}), where the construction is carried out in generalized wave coordinates, consistent with the general framework of the present paper. His works can be viewed as a justification and generalization of the considerations in the pioneering work \cite{CB.HF}.

\subsection{Outline of the paper} The remainder of the paper will be organized as follows. 

We first establish some important analytic and algebraic facts which are central to our proof. In \textbf{Section~\ref{sec:trilinear}}, we prove some trilinear estimates for null forms. In \textbf{Section~\ref{sec:EE}}, we analyze the algebraic structure of the Einstein equations.

We then turn to the proof of Theorem~\ref{thm:main}. In \textbf{Section~\ref{sec:mdm}}, we prove the first three parts of the theorem. Part (4) of the theorem, which concerns the propagation equation of the Vlasov measure is the most difficult part and is proven in \textbf{Section~\ref{sec:propagation}}.

Some proofs of the results from Section~\ref{sec:EE} are given in \textbf{Appendix~\ref{sec:appendix}}.

\subsection*{Acknowledgements} We thank John Anderson and Arthur Touati for helpful discussions. J.~Luk gratefully acknowledges the support by a Terman fellowship and the NSF grant DMS-2304445.

	\section{Trilinear compensated compactness for null forms}\label{sec:trilinear}
	
	The goal of this section is to establish estimates corresponding to the trilinear compensated compactness for null forms. As pointed out already in the introduction, we rely on the estimates proven by Ionescu--Pasauder \cite{IP}. After introducing some notations in \textbf{Section~\ref{sec:trilinear.notation}}, we will first deal with the simpler $\mathfrak Q_0$ case in \textbf{Section~\ref{sec:Q0}} and then turn to the harder case for  $\mathfrak Q_{\mu\nu}$ in \textbf{Section~\ref{sec:Qmunu}}.
	
	\subsection{Notations}\label{sec:trilinear.notation}
	
	\begin{definition}\label{def:X.norm}
	For $p \in [1,\infty]$ and $\lambda \geq 2$, define the $(\lambda,g_0)$-dependent norm $X^p_\lambda(g_0)$ by
	$$\| f\|_{X^p_\lambda(g_0)} \doteq \sum_{k=0}^2 \lambda^{-1+k} \| \rd f \|_{L^p} + \| \Box_{g_0} f \|_{L^p}.$$
	\end{definition}
	
	\begin{rk}
	The reason that we use the $X^p_\lambda(g_0)$ norm is that by the assumptions of $g_n - g_0$ in Assumption~\ref{ass:main} and the equation that it satisfies (see Proposition~\ref{prop:h.eqn} below), $\| g_n - g_0 \|_{X^p_\lambda(g_0)} \ls 1$ for every $p \in [1,\infty]$, uniformly as $\lambda \to 0$.
	\end{rk}

	\subsection{The $\mathfrak Q_0^{(g)}$ null forms}\label{sec:Q0}
	\begin{df}\label{def:Q0}
	Given a Lorentzian metric $g$, define the null form $\mathfrak Q_0^{(g)}$ by
	\begin{equation}
	\mathfrak Q^{(g)}_0(\phi, \psi) \doteq (g^{-1})^{\alp\bt} \rd_\alp \phi \rd_\bt \psi.
	\end{equation}
	\end{df}
	
	The following estimate can be proven easily with integration by parts (cf.~\cite[Proposition~12.2]{HL.Burnett}, \cite[Lemma~3.6]{GuerraTeixeira}).
	\begin{proposition}\label{prop:stupid.trilinear}
	Suppose $g$ is a smooth Lorentzian metric on an open $U\subset R^{d+1}$ and $\{\phi^{(i)}\}_{i=1,2,3} \subset C_c^\infty(U;\mathbb C)$ are supported in a fixed compact set $K \subset U$.
	
	Then for every $f \in C^\infty_c$
	$$\Big| \la f \rd_\alp \phi^{(3)}, \mathfrak Q_{0}^{(g)}(\phi^{(1)}, \phi^{(2)}) \ra_{L^2(\mathrm{dVol}_{g})} \Big| \ls \lambda \| \phi^{(1)}  \|_{X^\infty_\lambda(g)} \| \phi^{(2)}  \|_{X^2_\lambda(g)} \| \phi^{(3)}  \|_{X^2_\lambda(g)}, $$
	where the implicit constant may depend on $f$, $g$ and $K$, but is independent of $\lambda$ and $(\phi^{(1)}, \phi^{(2)}, \phi^{(3)})$.
	\end{proposition}
	\begin{proof}
	The key is to write 
	$$f \rd_\alp \phi^{(3)} \mathfrak Q_{0}^{(g)}(\phi^{(1)}, \phi^{(2)})  =  \f 12 f \rd_\alp \phi^{(3)} \Big( \Box_{g_0} (\phi^{(1)} \phi^{(2)}) -  \phi^{(1)} \Box_{g_0} \phi^{(2)} - \phi^{(2)} \Box_{g_0} \phi^{(1)} \Big).$$
	The integral of the last two terms is bounded by 
	$$\|\rd_\alp \phi^{(3)}\|_{L^2} \|\phi^{(1)} \|_{L^\i} \|\Box_{g_0} \phi^{(2)} \|_{L^2}\quad \hbox{and} \quad \|\rd_\alp \phi^{(3)}\|_{L^2} \|\phi^{(2)} \|_{L^2} \|\Box_{g_0} \phi^{(1)} \|_{L^\i},$$ respectively, and are both acceptable by Definition~\ref{def:X.norm}. For the remaining term, we integrate by parts. The main contribution can be bounded by $\|\Box_{g_0} \phi^{(3)}\|_{L^2} \|\rd_\alp \phi^{(1)} \|_{L^\i} \|\phi^{(2)} \|_{L^2}$ or $\|\Box_{g_0} \phi^{(3)}\|_{L^2} \| \phi^{(1)} \|_{L^\i} \|\rd_\alp \phi^{(2)} \|_{L^2}$ and are therefore acceptable as before. The error terms arising from differentiating $f$ in the process of the integration by parts are better. \qedhere
	\end{proof}
	
	\subsection{The $\mathfrak Q_{\mu\nu}$ null forms and the Ionescu--Pasauder estimate}\label{sec:Qmunu}
	\begin{df}\label{def:Qmunu}
	\begin{equation}
	\mathfrak Q_{\mu\nu} (\phi, \psi) \doteq \rd_\mu \phi \rd_\nu \psi - \rd_\nu \phi \rd_\mu \psi.
	\end{equation}
	\end{df}
	
	We have a result similar to Proposition~\ref{prop:stupid.trilinear}. Notice that we need a global smallness condition \eqref{eq:metric.smallness} and that the power of $\lambda$ that we gain is weaker.
	\begin{proposition}\label{prop:main.trilinear}
	Suppose $g$ is a smooth Lorentzian metric on an open set $U\subset \RR^{d+1}$ and $\{\phi^{(i)}\}_{i=1,2,3} \subset C_c^\infty(U;\mathbb C)$ are supported in a fixed compact set $K \subset U$. Assume, in addition, that $g$ is $C^0$ close to the Minkowski metric in the following sense:
	$$g = - \mathfrak N (\ud t -\mathfrak a_i \ud x^i)^2 + \mathfrak c_i (\ud x^i + \mathfrak b^i \ud t + \mathfrak B_{j}^{i} \ud x^j)^2,$$
	where $\mathfrak N$, $\mathfrak a_i$, $\mathfrak b^i$, $\mathfrak c_i$, $\mathfrak B_{j}^{i}$ are smooth functions satisfying	
	\begin{equation}\label{eq:metric.smallness}
	|\mathfrak N-1|,\, |\mathfrak c_i-1|,\, |\mathfrak a_i|,\, |\mathfrak b^i|,\, |\mathfrak B_{j}^{i}| \leq 10^{-5}.
	\end{equation}
	
	Then for every $f \in C^\infty(U;\mathbb C)$,
	$$\Big| \la f \rd_\alp \phi^{(3)}, \mathfrak Q_{\mu\nu}(\phi^{(1)}, \phi^{(2)}) \ra_{L^2(\mathrm{dVol}_{g})} \Big| \ls \lambda^{\f{1}{15}}  \| \phi^{(1)}  \|_{X^\infty_\lambda(g)} \| \phi^{(2)}  \|_{X^2_\lambda(g)} \| \phi^{(3)}  \|_{X^2_\lambda(g)}, $$
	where the implicit constant may depend on $f$, $g$ and $K$, but is independent of $\lambda$ and $(\phi^{(1)}, \phi^{(2)}, \phi^{(3)})$.
	\end{proposition}
	
	In the constant coefficient case, this proposition is a direct result of the work of Ionescu--Pasauder (which in fact proves a slightly stronger result). We will first recall the estimate of Ionescu--Pasauder, and then adapt it to our case after freezing coefficients to reduce to the constant coefficient case.

	We introduce some conventions for Lemma~\ref{lem:IP.orig} and Propoosition~\ref{prop:normal.form}, which concern some estimates on the Minkowski spacetime. In Lemma~\ref{lem:IP.orig} and Propoosition~\ref{prop:normal.form}, we take the metric and the volume form to be that of Minkowski. Denote a point in Minkowski by $(t,\ux)$. We will use the \underline{spatial} Fourier transform denoted by 
	$$\ucalF(f)(\uxi) = \int_{\mathbb R^{d}} f(t,\ux) e^{-i\ux\cdot \uxi} \, \ud \ux,\quad \ucalF^{-1}(h)(\ux) = \f 1{(2\pi)^{d}} \int_{\mathbb R^{d}} e^{i\ux\cdot \uxi} h(\uxi) \, \ud \uxi.$$
	Let $P_k$ denote the standard Littlewood--Paley projection, but only in the spatial variables $\ux$. Define also as in \cite{IP} the following convention for angles
	\begin{equation}\label{eq:def.Xi}
	\Xi_{\iota_1 \iota_2}(\underline{\th},\ueta) = \sqrt{\f{2(|\underline{\th}||\ueta| - \iota_1 \iota_2 \underline{\theta}\cdot \ueta)}{|\underline{\theta}||\ueta|}}.
	\end{equation}
	The following lemma is from \cite{IP} up to extremely minor modifications.
	\begin{lemma}[Ionescu--Pasauder, Lemma~2.9 in \cite{IP}]\label{lem:IP.orig}
	Let $\iota, \iota_1, \iota_2 \in \{ +, -\}$, $b\leq 2$, $f,\,f_1,\, f_2 \in L^2(\mathbb R^3)$, and $k, \, k_1,\, k_2 \in \mathbb Z$. Let $m$ be one of the following symbols:
	\begin{equation}\label{eq:symbol.IP}
	m(\uth,\ueta)= \f{\uth_i \ueta_j - \uth_j \ueta_i}{|\uth| |\ueta|} \quad \hbox{or} \quad m(\uth,\ueta) = \f{\ueta_j}{|\ueta|} - \f{\uth_j}{|\uth|}.
	\end{equation}
	\begin{enumerate}
	\item Assume $\chi_1:\mathbb R^3 \to [0,1]$ is a smooth function supported in the ball $B_2$ and define
	$$L^b_{k,k_1,k_2} \doteq \int_{\mathbb R^3\times \mathbb R^3} m(\uxi-\ueta,\ueta)\chi_1(2^{-b}\Xi_{\iota_1 \iota_2}(\uxi-\ueta,\ueta)) \cdot \widehat{P_{k_1} f_1}(\uxi-\ueta) \widehat{P_{k_2} f_2}(\ueta) \widehat{P_{k} f}(\uxi) \, \ud \uxi \, \ud \ueta.$$
	Then 	
	\begin{equation}\label{eq:IP.L.est}
	|L^b_{k,k_1,k_2}| \ls 2^b (2^{k_1-k}+1) \| P_{k_1} f_1\|_{L^{\infty}} \| P_{k_2} f_2\|_{L^{2}}  \|P_{k} f\|_{L^{2}}.
	\end{equation}
	\item Assume that $\chi_{2}:\mathbb R^3\to [0,1]$ is a smooth function supported in $B_2 \setminus B_{1/2}$ and define
	$$M^b_{k,k_1,k_2} \doteq \int_{\mathbb R^3\times \mathbb R^3} m(\uxi-\ueta,\ueta) \f{\chi_2(2^{-b'} \Xi_{\iota_1 \iota_2}(\uxi-\ueta,\ueta) )}{\Phi_{\iota \iota_1 \iota_2}(\uxi,\ueta)} \widehat{P_{k_1} f_1}(\uxi-\ueta) \widehat{P_{k_2} f_2}(\ueta) \widehat{P_{k} f}(\uxi) \, \ud \uxi \, \ud \ueta, $$
	where $\Phi_{\iota \iota_1 \iota_2}(\uxi,\ueta)$ is the phase given by 
	$$\Phi_{\iota \iota_1 \iota_2}(\uxi,\ueta) \doteq \iota |\uxi| - \iota_1|\uxi - \ueta| - \iota_2|\ueta|.$$
	Then
	\begin{equation}\label{eq:IP.M.est}
	|M^b_{k,k_1,k_2}| \ls 2^{-b'} 2^{-\min\{k_1,k_2\}} (2^{k_1-k}+1) \| P_{k_1} f_1\|_{L^{\infty}} \| P_{k_2} f_2\|_{L^{2}}  \|P_{k} f\|_{L^{2}}.
	\end{equation}
	\end{enumerate}
	\end{lemma}
	\begin{proof}
	This is almost explicitly as in \cite{IP} (and we have essentially follow their notation, except for rewriting $\Phi_{\sigma\mu\nu}$ as $\Phi_{\iota \iota_1 \iota_2}$). The only difference is that we used the specific symbols in \eqref{eq:symbol.IP}, instead of a symbol satisfying $\|\mathcal F^{-1} m\|_{L^1(\mathbb R^3\times \mathbb R^3)} \leq 1$ in \cite{IP}. The estimates are therefore also modified accordingly.
	
	The point is that we can repeat the proof of \cite[Lemma~2.9]{IP} and noting that the symbol $m(\uth,\ueta)$ satisfies the pointwise bound $|m(\uth,\ueta)| \ls 2^b$ on the intersection of the supports of various relevant cutoffs. Moreover, for the operators 
	$$2^{2b} \Delta_{\uth},\quad 2^{2b} \Delta_{\ueta},\quad L_{\uth} = \sum_j \uth_j \rd_{\uth_j}, \quad L_{\ueta} = \sum_j \ueta_j \rd_{\ueta_j},\quad S_{ij} = \uth_i \rd_{\uth_j} + \ueta_i \rd_{\ueta_j}$$
	that are used in the proof of \cite[Lemma~2.9]{IP}, the symbol $m$ satisfies
	$$|(2^{2b} \Delta_{\uth})^{I_1} (2^{2b} \Delta_{\ueta})^{I_2} L_{\uth}^{I_3} L_{\ueta}^{I_4} m(\uth,\ueta)| \ls 2^b$$
	on the region where the integrand is non-vanishing. One can thus repeat the same proof to obtain the desired estimate.
	\end{proof}

	\begin{proposition}\label{prop:normal.form}
	The following estimate holds for all $\de \leq 1$ and all $\bt, \mu, \nu \in \{0,1,\cdot,d\}$:
	\begin{equation}
	\begin{split}
	&\: \Big| \la \rd_\bt P_{k_3} \phi_3, \mathfrak Q_{\mu\nu}(P_{k_1} \phi_1, P_{k_2} \phi_2) \ra \Big| \\
	\ls &\: \de (2^{k_1-k_3}+1) \| \rd \phi_{1} \|_{L^\i} \| \rd \phi_{2} \|_{L^2} \|  \rd \phi_{3}\|_{L^2} \\
	&\: + \de^{-1} 2^{-\min\{k_1,k_2\}} (2^{k_1-k_3}+1) \Big(\| \Box_M \phi_{1} \|_{L^\i} \| \rd \phi_{2} \|_{L^2} \| \rd \phi_{3}\|_{L^2} \\
	&\: \qquad\qquad+ \| \rd \phi_{1} \|_{L^\i} \| \Box_M \phi_{1} \|_{L^2}  \| \rd \phi_{3}\|_{L^2}+ \| \rd \phi_{1} \|_{L^\i} \| \rd \phi_{2} \|_{L^2} \| \Box_M \phi_{3}\|_{L^2}\Big),
	\end{split}
	\end{equation}
	where $\Box_M$ denotes the Minkowskian wave operator.
	\end{proposition}
	\begin{proof}
	Fix $\de\leq 1$ and choose $b \in \mathbb Z$ such that $\de \leq 2^b \leq 10\de$. We introduce some notations for the remainder of the proof. First, following \cite{IP}, we denote 
	\begin{align}
	\calG_m[f_1,f_2,f_3] = &\: \int_{\mathbb R^3\times \mathbb R^3} m(\uxi-\ueta,\ueta) \widehat{f_1}(\uxi - \ueta) \widehat{f_2}(\ueta) \overline{\widehat{f_3}(\uxi)}\, \ud \uxi \ud \ueta, \\
	\calH_m^{\iota_1,\iota_2,\iota_3}[f_1,f_2,f_3] = &\: \int_{\mathbb R^3\times \mathbb R^3} \f{m(\uxi-\ueta,\ueta)}{\Phi_{\iota_3 \iota_1 \iota_2}(\uxi,\ueta)} \widehat{f_1}(\uxi - \ueta) \widehat{f_2}(\ueta) \overline{\widehat{f_3}(\uxi)}\, \ud \uxi \ud \ueta.
	\end{align}
	Note that $\calG$, $\calH$ are linear in the first two slots, but conjugate linear in the third. We also define $R_j$ and $|\underline{\nabla}|$ to be the following operators:
	$$ \ucalF \calR_j  f(\uxi) = -i \f{\uxi_j }{|\uxi|} \ucalF f (\uxi),\quad \ucalF |\underline{\nabla}| f(\uxi) = |\uxi| \ucalF f (\uxi).$$ 
	
	We perform normal form as in \cite{IP}. Define $\psi^{(j,\pm)}$ by 
	$$\ucalF \psi_{j,\pm} = (\rd_t \pm i |\uxi |) \ucalF \phi_{j}.$$
	In particular, the following two identities hold:
	\begin{align}
	(\rd_t \mp i|\uxi |)\ucalF \psi_{j,\pm} = - \ucalF (\Box_M \phi_j), \label{eq:psi.and.Box} \\
	\rd_\ell \phi_{j} = \f i 2 R_\ell (\psi_{j,-} - \psi_{j,+}),\quad \rd_t \phi_{j} = \f 12 (\psi_{j,+}+\psi_{j,-}).
	\end{align}

	Now taking spatial Fourier transform $\ucalF$, the term $\la \rd_\bt P_{k_3} \phi_3, \mathfrak Q_{\mu\nu}(P_{k_1} \phi_1, P_{k_2} \phi_2) \ra$ can be written as linear combinations of terms of the form 
	\begin{equation}\label{eq:G.terms.to.estimate}
	\begin{split}
	\int_{\mathbb R} \calG_m[P_{k_1} \psi_{1,\iota}, P_{k_2} \psi_{2,\iota_2}, P_{k_3} \psi_{3,\iota_3}] \, \ud t, \quad \int_{\mathbb R} \calG_m[P_{k_1} R_j \psi_{1,\iota}, P_{k_2} \psi_{2,\iota_2}, P_{k_3} \psi_{3,\iota_3}]\, \ud t
	\end{split}
	\end{equation}
	for $\iota_1,\iota_2,\iota_3 \in \{+,-\}$ and with $m$ being one of the symbols in \eqref{eq:symbol.IP}.
		
	Let us control the first type of terms in \eqref{eq:G.terms.to.estimate}; the second type of terms are similar. We split the term into two parts:
	$$\calG_{m}[\psi_{1,\iota_1}, \psi_{2,\iota_2}, \psi_{3,\iota_3}] = \calG_{\chi_1 m}[\psi_{1,\iota_1}, \psi_{2,\iota_2}, \psi_{3,\iota_3}] + \calG_{(1-\chi_1) m}[\psi_{1,\iota_1}, \psi_{2,\iota_2}, \psi_{3,\iota_3}], $$ where $\chi_1 = \chi_1(2^{-b} \Xi_{\iota_1,\iota_2}(\uxi-\ueta,\ueta))$ (recall \eqref{eq:def.Xi}).
	
	For $\calG_{\chi_1 m}[\psi_{1,\iota_1}, \psi_{2,\iota_2}, \psi_{3,\iota_3}]$, we use \eqref{eq:IP.L.est} and compact support in $t$ to obtain
	\begin{equation}\label{eq:normal.form.1}
	\int_{\mathbb R} |\calG_{\chi_1 m}[\psi_{1,\iota_1}, \psi_{2,\iota_2}, \psi_{3,\iota_3}]| \ud t \ls 2^b (2^{k_1-k}+1) \| \rd P_{k_1} \phi_1\|_{L^{\infty}} \| \rd P_{k_2} \phi_2\|_{L^{2}}  \|P_{k} \rd \phi_3\|_{L^{2}}.
	\end{equation}
	
	For $\calG_{(1-\chi_1) m}[\psi_{1,\iota_1}, \psi_{2,\iota_2}, \psi_{3,\iota_3}]$, we note that
	\begin{equation}
	\begin{split}
	&\: i \calG_{(1-\chi_1) m}[\psi_{1,\iota_1}, \psi_{2,\iota_2}, \psi_{3,\iota_3}]  \\
	=&\: -\calH_{(1-\chi_1) m}^{\iota_1,\iota_2,\iota_3}[\psi_{1,\iota_1}, \psi_{2,\iota_2}, \iota_3 i |\nabla| \psi_{3,\iota_3}] - \calH_{(1-\chi_1) m}^{\iota_1,\iota_2,\iota_3}[\iota_1 i|\nabla|\psi_{1,\iota_1}, \psi_{2,\iota_2}, \psi_{3,\iota_3}] \\
	&\: - \calH_{(1-\chi_1) m}^{\iota_1,\iota_2,\iota_3}[\psi_{1,\iota_1}, \iota_2 i |\nabla| \psi_{2,\iota_2}, \psi_{3,\iota_3}] \\
	=&\: -\calH_{(1-\chi_1) m}^{\iota_1,\iota_2,\iota_3}[\psi_{1,\iota_1}, \psi_{2,\iota_2},  \Box_M \phi_{3}] - \calH_{(1-\chi_1) m}^{\iota_1,\iota_2,\iota_3}[\Box_M \phi_{1}, \psi_{2,\iota_2}, \psi_{3,\iota_3}] - \calH_{(1-\chi_1) m}^{\iota_1,\iota_2,\iota_3}[\psi_{1,\iota_1}, \Box_M \phi_{2}, \psi_{3,\iota_3}] \\
	&\: -\int_{\mathbb R^3\times \mathbb R^3} \rd_t \Big[ \f{(1-\chi_1)m(\uxi-\ueta,\ueta)}{\Phi_{\iota_3 \iota_1 \iota_2}(\uxi,\ueta)} \widehat{\psi_{1,\iota_1}}(\uxi - \ueta) \widehat{\psi_{2,\iota_2}}(\ueta) \overline{\widehat{\psi_{3,\iota_3}}(\uxi)}  \Big]\, \ud \uxi \ud \ueta,
	\end{split}
	\end{equation}
	where we used \eqref{eq:psi.and.Box} in the last equality. We then integrate in $t$ over $\mathbb R$ in $t$, and note that the last term, which is a total $\rd_t$ derivative, drops. After using compact support, the time-integral of the remaining three terms can be bounded using \eqref{eq:IP.M.est}. After summing over all dyadic scales $2^{b'} \geq 2^b$, we obtain
	\begin{equation}\label{eq:normal.form.2}
	\begin{split}
	&\: \Big| \int_{\mathbb R} \calG_{(1-\chi_1) m}[\psi_{1,\iota_1}, \psi_{2,\iota_2}, \psi_{3,\iota_3}] \, \ud t \Big| \\
	\ls &\: 2^{-b} 2^{-\min\{k_1,k_2\}} (2^{k_1-k_3}+1)  \max_{\{\sigma(1),\sigma(2),\sigma(3)\} = \{1,2,3\}} \| \Box_M \phi_{\sigma(1)} \|_{L^2} \| \rd \phi_{\sigma(2)} \|_{L^2} \|  \rd \phi_{\sigma(3)}\|_{L^\i}.
	\end{split}
	\end{equation}
	
	Combining \eqref{eq:normal.form.1} and \eqref{eq:normal.form.2}, and recalling that $\de \leq 2^b \leq 10\de$, we obtain the desired estimate. \qedhere
	\end{proof}
	
	\begin{corollary}\label{cor:normal.form}
	Suppose $g$ is a \textbf{constant-coefficient} metric
	$$g = - \mathfrak N (\ud t -\mathfrak a_i \ud x^i)^2 + \mathfrak c_i (\ud x^i + \mathfrak b^i \ud t + \mathfrak B_{j}^{i} \ud x^j)^2,$$
	where $\mathfrak N$, $\mathfrak a_i$, $\mathfrak b^i$, $\mathfrak c_i$, $\mathfrak B_{j}^{i}$ are constants such that $\mathfrak B_{j}^{i}$ has vanishing diagonal entries and 
	\begin{equation}\label{eq:metric.smallness.2}
	|\mathfrak N-1|,\, |\mathfrak c_i-1|,\, |\mathfrak a_i|,\, |\mathfrak b^i|,\, |\mathfrak B_{j}^{i}| \leq 10^{-10}.
	\end{equation}
	
	Then the following modification of the estimate in Proposition~\ref{prop:normal.form} holds:
	\begin{equation}\label{eq:normal.form}
	\begin{split}
	&\: \Big| \la \rd_\bt P_{k_3} \phi_3, \mathfrak Q_{\mu\nu}(P_{k_1} \phi_1, P_{k_2} \phi_2) \ra \Big| \\
	\ls &\: \de (2^{k_1-k_3}+1) \| \rd \phi_{1} \|_{L^\i} \| \rd \phi_{2} \|_{L^2} \|  \rd \phi_{3}\|_{L^2} \\
	&\: + \de^{-1} 2^{-\min\{k_1,k_2\}} (2^{k_1-k_3}+1) \Big(\| \Box_g \phi_{1} \|_{L^\i} \| \rd \phi_{2} \|_{L^2} \| \rd \phi_{3}\|_{L^2} \\
	&\: \qquad\qquad+ \| \rd \phi_{1} \|_{L^\i} \| \Box_g \phi_{2} \|_{L^2}  \| \rd \phi_{3}\|_{L^2}+ \| \rd \phi_{1} \|_{L^\i} \| \rd \phi_{2} \|_{L^2} \| \Box_g \phi_{3}\|_{L^2}\Big),
	\end{split}
	\end{equation}
	with constants \underline{independent} of $g$, where the Littlewood--Paley projection $P_k$ is to be understood in the spatial coordinates $(\tilde{x}^1, \tilde{x}^2, \tilde{x}^3)$ for the coordinate system $(\tilde{t}, \tilde{x}^1, \tilde{x}^2, \tilde{x}^3)$ defined by 
	\begin{equation}\label{eq:new.coord}
	\tilde{t} = t - \mathfrak a_i x^i,\quad \tilde{x}^i = x^i + \mathfrak b^i t +\mathfrak B_{ij} x^j.
	\end{equation}
	\end{corollary}
	\begin{proof}
	The argument is exactly the same except for changing to the $(\tilde{t}, \tilde{x}^1, \tilde{x}^2, \tilde{x}^3)$ coordinates. The smallness condition \eqref{eq:metric.smallness.2} ensures that the implicit constants can be made independent of $g$. \qedhere
	\end{proof}
		
	We can now prove Proposition~\ref{prop:main.trilinear}.

	\begin{proof}[Proof of Proposition~\ref{prop:main.trilinear}]
	By scaling, we can assume that 
	$$\| \phi^{(1)}  \|_{X^\infty_\lambda(g)} = \| \phi^{(2)}  \|_{X^2_\lambda(g)} = \| \phi^{(3)}  \|_{X^2_\lambda(g)} = 1.$$
	
	\pfstep{Step~1: Spatial cutoffs} Let $\ep_0 = \f 12$. Cover the compact set $K$ with $O(\lambda^{-4\ep_0})$ cubes of side-lengths $\lambda^{\ep_0}$, labelled by $C_\alp$. Let $\{ \zeta_\alp^3\}$ be a smooth partition of unity corresponding to these cubes. The cut-off$\zeta_\alp$ can be chosen so that $\| \rd^k \zeta_\alp \|_{L^1} \ls \lambda^{(4-k)\ep_0}$. 
	

	Define $\phi^{(i)}_\alp \doteq \zeta_\alp  \phi^{(i)} $. Since the cutoff functions live at a larger scale than $\lambda$, it is easy to check using the $X^{\infty}_\lambda(g_0)$ norm bound that
	\begin{equation}\label{eq:null.phi.Li}
	\sum_{k=0}^2 \lambda^{1-k} \| \rd^k \phi^{(1)}_\alp \|_{L^\infty} \ls 1.
	\end{equation}
	
	Let $g_\alp$ be the constant coefficient Lorentzian metric given by $g$ at the center of the ball. Then 
	$$\|g - g_\alp \|_{L^\i(2B_\alp)} \ls \lambda^{\ep_0}.$$ 
	$\Box_{g_\alp}$ is now a constant coefficient wave operator. We will estimate the size of $\Box_{g_\alp} \phi^{(i)}_\alp$. First, we estimate
	\begin{equation*}
	\begin{split}
	\| \Box_g \phi^{(1)}_\alp\|_{L^{\infty}} \ls &\: \| \Box_g \phi^{(1)} \|_{L^{\infty}} + \| \rd \zeta_\alp\|_{L^\i} \|\rd \phi^{(1)}\|_{L^{\infty}} + \| \rd^2 \zeta_\alp\|_{L^\i} \| \phi^{(1)}\|_{L^\infty} \\
	\ls &\: 1+ \lambda^{-\ep_0} + \lambda^{1-2\ep_0} \ls \lambda^{-\ep_0} + \lambda^{1-2\ep_0}.
	\end{split}
	\end{equation*}
	Using this, we deduce
	\begin{equation}\label{eq:null.Boxphi.cutoff.Linfty}
	\begin{split}
	\| \Box_{g_\alp} \phi^{(1)}_\alp\|_{L^{\infty}} \ls &\: \| \Box_g \phi^{(1)}_\alp\|_{L^{\infty}} + \| g^{-1} - g^{-1}_\alp\|_{L^\i(2B_\alp)} \| \rd^2 \phi^{(1)}_{\alp} \|_{L^{\infty}} + \|\rd \phi^{(1)}_\alp \|_{L^{\infty}} \\
	\ls &\: \lambda^{-\ep_0} + \lambda^{1-2\ep_0} + \lambda^{-1+\ep_0} + 1\ls \lambda^{-\f 12}.
	\end{split}
	\end{equation}
	
	For $\phi^{(2)}$ and $\phi^{(3)}$, we only have $L^2$-based bounds. For these, we argue similarly as above, but we use also \underline{orthogonality} to obtain
	\begin{equation}\label{eq:null.phi.L2}
	\sum_{j=0}^2 \lambda^{1-j} \Big(\sum_{\alp} \| \rd^j  \phi^{(i)}_\alp \|_{L^2}^2\Big)^{\f 12} \ls 1,\quad  \Big(\sum_{\alp} \| \Box_{g_\alp}  \phi^{(i)}_\alp \|_{L^2}^2\Big)^{\f 12} \ls \lambda^{-\f 12}\quad i = 2,3.
	\end{equation}
	
	\pfstep{Step~2: A simple reduction} Let $f_\alp$ be the value of $f$ at the center of $C_\alp$. Since $f$ is smooth, by the mean value theorem, $|f-f_\alp| \ls \lambda^{\ep_0}$ on the support of $\zeta_\alp$. Hence, 
	\begin{equation}\label{eq:null.reduction.1}
	\begin{split}
	&\: \Big| \la f \rd_\bt \phi^{(1)}, \mathfrak Q_{\mu\nu}(\phi^{(2)}, \phi^{(3)}) \ra -\sum_{\alp} \la f_\alp \rd_\bt \phi^{(1)}_\alp,  \mathfrak Q_{\mu\nu}(\phi^{(2)}_\alp, \phi^{(3)}_\alp) \ra \Big| \\
	\ls &\: \max\{\lambda^{1-\ep_0}, \lambda^{\ep_0} \} \sum_{\alp} \lambda^{4\ep_0} \ls \max\{\lambda^{1-\ep_0}, \lambda^{\ep_0} \} = \lambda^{\f 12},
	\end{split}
	\end{equation}
	where we used $\ep_0 = \f 12$ in the last step.
	
	\pfstep{Step~3: Fourier cut-off} To proceed, we further introduce another spatial cut-off in Fourier space. Here, ``spatial'' is to be understood with respect to the $(\tilde{t}, \tilde{x}^1, \tilde{x}^2, \tilde{x}^3)$ defined in \eqref{eq:new.coord} corresponding to the constant coefficient metric $g_\alp$. Let $a= \f 1{14}$. For each $i$ and $\alp$, define
	$$\widetilde{\phi}^{(i)}_\alp \doteq \ucalF^{-1}(\chi_{\mathcal F}(\uxi) \ucalF \phi^{(i)}_\alp),$$
	where $\chi_{\ucalF}(\uxi)$ is a smooth cutoff function such that $\mathrm{supp}(\chi_{\mathcal F}) \subset \{ \f 12 \lambda^{-1+a} \leq |\uxi| \leq 2 \lambda^{-1-a}\}$ and $\chi_{\mathcal F}(\uxi)\equiv 1$ when $\lambda^{-1+a} \leq |\uxi| \leq \lambda^{-1-a}$. Using \eqref{eq:null.phi.Li}, \eqref{eq:null.Boxphi.cutoff.Linfty} and \eqref{eq:null.phi.L2}, it is easy to see that 
	\begin{equation}\label{eq:null.tphi.cutoff}
	\begin{split}
	\sum_{j=0}^2 \lambda^{1-j} \| \rd^j  \widetilde{\phi}^{(1)}_\alp \|_{L^\i} \ls 1,&\: \quad   \| \Box_{g_\alp}  \widetilde{\phi}^{(1)}_\alp \|_{L^\i}  \ls \lambda^{-\f 12} \\
	\sum_{j=0}^2 \lambda^{1-j} \Big(\sum_{\alp} \| \rd^j  \widetilde{\phi}^{(i)}_\alp \|_{L^2}^2\Big)^{\f 12} \ls 1,&\:\quad  \Big(\sum_{\alp} \| \Box_{g_\alp}  \widetilde{\phi}^{(i)}_\alp \|_{L^2}^2\Big)^{\f 12} \ls \lambda^{-\f 12}\quad i = 2,3.
	\end{split}
	\end{equation}
	
	Moreover, $\phi^{(i)}_\alp - \widetilde{\phi}^{(i)}_\alp$ either has spatial frequency lower than $\f 12\lambda^{-1+a}$ or higher than $2\lambda^{-1-a}$. Denote $\sigma^{(i)}_\alp = \phi^{(i)}_\alp - \widetilde{\phi}^{(i)}_\alp$ and decompose $\sigma^{(i)}_\alp = \sigma^{(i),L}_\alp + \sigma^{(i),H}_\alp$, where $\sigma^{(i),L}_\alp$ and $\sigma^{(i),H}_\alp$ have low- and high-frequency, respectively. Notice that $2^{-k} \widetilde{\underline{\rd}} P_k$ and $P_k$ both have a kernel in $L^1$ and thus the operators are bounded both on $L^2$ and on $L^\i$ on a fixed dyadic frequency. Thus, using \eqref{eq:null.tphi.cutoff}, it follows that the spatial derivatives of $\phi^{(i)}_\alp - \widetilde{\phi}^{(i)}_\alp$ (denoted by $\widetilde{\underline{\rd}} \in \{ \widetilde{\rd}_{\tilde{x}^1}, \widetilde{\rd}_{\tilde{x}^2}, \widetilde{\rd}_{\tilde{x}^3}\}$) obey improved estimates:

	$$\| \widetilde{\underline{\rd}} \sigma^{(1),L}_\alp \|_{L^\i} \ls \sum_{2^k \ls \lambda^{-1+a}} \| \widetilde{\underline{\rd}} P_k \phi^{(1)}_\alp \|_{L^\i} \ls \sum_{2^k \ls \lambda^{-1+a}} 2^k \| \phi^{(1)}_\alp\|_{L^\i} \ls \lambda^a,$$
	and
	$$\| \widetilde{\underline{\rd}} \sigma^{(1),H}_\alp \|_{L^\i} \ls \sum_{2^k \gtrsim \lambda^{-1-a}} 2^{-k} \| \widetilde{\underline{\rd}}{}^2 P_k \phi^{(1)}_\alp \|_{L^\i} \ls \sum_{2^k \gtrsim \lambda^{-1-a}} 2^{-k} \lambda^{-1}  \ls \lambda^a.$$
	Combining and arguing similarly for $\phi^{(i)}_\alp - \widetilde{\phi}^{(i)}_\alp$ when $i = 2,3$, we obtain
	\begin{equation}\label{eq:null.tphi.cutoff.spatial.improved}
	\begin{split}
	 \| \widetilde{\underline{\rd}} (\phi^{(1)}_\alp - \widetilde{\phi}^{(1)}_\alp) \|_{L^\i} \ls \lambda^a, \quad 
	\Big(\sum_{\alp} \| \widetilde{\underline{\rd}} (\phi^{(i)}_\alp - \widetilde{\phi}^{(i)}_\alp) \|_{L^2}^2 \Big)^{\f 12} \ls \lambda^{a},\quad i = 2,3.
	\end{split}
	\end{equation}
	
	We now derive estimates similar  to \eqref{eq:null.tphi.cutoff.spatial.improved}, but for the $\widetilde{\rd}_{\tilde{t}}$ derivative, using the estimates for $\Box_{g_\alp} \widetilde{\phi}^{(i)}_\alp$ and suitable elliptic estimates. We begin with $\widetilde{\rd}_{\tilde{t}}(\phi^{(i)}_\alp - \widetilde{\phi}^{(i)}_\alp)$. For $\sigma^{(i),L}_\alp$, $\sigma^{(i),H}_\alp$ defined above, note that
	$$ \| \sigma^{(1),H}_\alp \|_{L^\i}  \| \underline{\widetilde{\rd}}{}^2 \sigma^{(1),H}_\alp \|_{L^\i} \ls  \lambda^{2+2a} \| \rd^2 \phi^{(1)}_\alp \|_{L^\i} \| \underline{\widetilde{\rd}}{}^2 \sigma^{(1),H}_\alp \|_{L^\i}  \ls \lambda^{2a}$$ and also $\| \sigma^{(1),L}_\alp \|_{L^\i}  \| \underline{\widetilde{\rd}}{}^2 \sigma^{(1),L}_\alp \|_{L^\i}  \ls \lambda^{2a}$. Thus using the one-dimensional bound $\| f' \|_{L^\i}^2 \leq 4 \|f \|_{L^\i} \|f'' \|_{L^\i}$ (see \cite[Chapter~5, Problem~15]{Rudin}), we obtain
	\begin{equation}\label{eq:null.tphi.cutoff.time.improved.1}
	\begin{split}
	 \| \widetilde{\rd}_{\tilde{t}} \sigma^{(1),H}_\alp \|_{L^\i}^2 
	\leq &\: 2 \| \sigma^{(1),H}_\alp \|_{L^\i}\| \widetilde{\rd}_{\tilde{t}}^2 \sigma^{(1),H}_\alp \|_{L^\i}\\
	\ls &\: \| \sigma^{(1),H}_\alp \|_{L^\i} (\|\Box_{g_\alp} \sigma^{(1),H}_\alp\|_{L^\i} + \| \underline{\widetilde{\rd}}{}^2 \sigma^{(1),H}_\alp \|_{L^\i}) 
	\ls \lambda \lambda^{-\f 12} + \lambda^{2a} \ls  \lambda^{2a},
	\end{split}
	\end{equation}
	where in the penultimate inequality we used \eqref{eq:null.tphi.cutoff} and \eqref{eq:null.tphi.cutoff.spatial.improved} and in the final inequality we used $2a = \f 17 < \f 12$.
	
	For $\widetilde{\rd}_{\tilde{t}}(\phi^{(i)}_\alp - \widetilde{\phi}^{(i)}_\alp)$ when $i=2,3$, we integrate by parts and use the wave operator bound as follows:
	\begin{equation}\label{eq:null.tphi.cutoff.time.improved}
	\begin{split}
	 \sum_\alp \| \widetilde{\rd}_{\tilde{t}} (\phi^{(i)}_\alp - \widetilde{\phi}^{(i)}_\alp) \|_{L^2}^2 
	= &\: - \sum_\alp \la \phi^{(i)}_\alp - \widetilde{\phi}^{(i)}_\alp, \widetilde{\rd}^2_{\tilde{t}} (\phi^{(i)}_\alp - \widetilde{\phi}^{(i)}_\alp) \ra \\
	\ls &\: \sum_\alp \|\phi^{(i)}_\alp - \widetilde{\phi}^{(i)}_\alp \|_{L^2} \|\Box_{g_\alp} (\phi^{(i)}_\alp - \widetilde{\phi}^{(i)}_\alp)\|_{L^2} + \sum_\alp \| \widetilde{\underline{\rd}} (\phi^{(i)}_\alp - \widetilde{\phi}^{(i)}_\alp) \|_{L^2}^2 \\
	\ls &\: \lambda \lambda^{-\f 12} + \lambda^{2a} \ls  \lambda^{2a},
	\end{split}
	\end{equation}
	where we used \eqref{eq:null.tphi.cutoff} and \eqref{eq:null.tphi.cutoff.spatial.improved} similarly as in \eqref{eq:null.tphi.cutoff.time.improved.1}.

	Hence, using H\"older's inequality and \eqref{eq:null.tphi.cutoff}, \eqref{eq:null.tphi.cutoff.spatial.improved}, \eqref{eq:null.tphi.cutoff.time.improved.1}, \eqref{eq:null.tphi.cutoff.time.improved}, we obtain
	\begin{equation}\label{eq:null.reduction.2}
	\Big| \sum_{\alp} \la f_\alp \rd \phi^{(3)}_\alp,  Q_{\mu\nu}(\phi^{(1)}_\alp, \phi^{(2)}_\alp) \ra -\sum_{\alp} \la f_\alp \rd \widetilde{\phi}^{(3)}_\alp,  Q_{\mu\nu}(\widetilde{\phi}^{(1)}_\alp, \widetilde{\phi}^{(2)}_\alp) \ra \Big|  \ls \lambda^a,
	\end{equation}
	which is acceptable since $a= \f 1{14} > \f 1{15}$.
	
	By \eqref{eq:null.reduction.1} and \eqref{eq:null.reduction.2}, it therefore suffices to bound the term
	\begin{equation}\label{eq:null.const.coeff.error}
	\sum_{\alp} \la f_\alp \rd \widetilde{\phi}^{(1)}_\alp,  Q_{\mu\nu}(\widetilde{\phi}^{(2)}_\alp, \widetilde{\phi}^{(3)}_\alp) \ra .
	\end{equation}
	
	\pfstep{Step 3: Applying the Ionescu--Pasauder normal form bounds} Now, each term in the summand in \eqref{eq:null.const.coeff.error} can be treated with the Ionescu--Pasauder estimate \eqref{eq:normal.form}. We decompose each $\widetilde{\phi}^{(i)}$ into Littlewood--Paley pieces. Notice that because of the Fourier cut-offs in the previous step, all factors of $2^{k_1-k_3}$ in \eqref{eq:normal.form} are at worst $O(\lambda^{-2a})$, and the $2^{-\min\{k_1,k_2\}}$ factor can be bounded above by $O(\lambda^{1-a})$. Therefore, we have
	\begin{equation}
	\begin{split}
	&\: \Big| \la \rd P_{k_3} \widetilde{\phi}^{(3)}_\alp,  Q_{\mu\nu}(P_{k_1}  \widetilde{\phi}^{(1)}_\alp, P_{k_2} \widetilde{\phi}^{(2)}_\alp) \ra \Big|\\
	\ls &\: \de \lambda^{-2a}  \| \rd \widetilde{\phi}^{(1)}_\alp\|_{L^\i} \| \rd \widetilde{\phi}^{(2)}_\alp\|_{L^2} \| \rd \widetilde{\phi}^{(3)}_\alp\|_{L^2} \\
	&\:  + \de^{-1} \lambda^{1-a} \lambda^{-2a} \Big(\| \Box_g \widetilde{\phi}^{(1)}_\alp \|_{L^\i} \| \rd \widetilde{\phi}^{(2)}_\alp \|_{L^2} \| \rd \widetilde{\phi}^{(3)}_\alp\|_{L^2} \\
	&\: \qquad\qquad+ \| \rd \widetilde{\phi}^{(1)}_\alp \|_{L^\i} \| \Box_g \widetilde{\phi}^{(2)}_\alp \|_{L^2}  \| \rd \widetilde{\phi}^{(3)}_\alp\|_{L^2}+ \| \rd \widetilde{\phi}^{(1)}_\alp \|_{L^\i} \| \rd \widetilde{\phi}^{(2)}_\alp \|_{L^2} \| \Box_g \widetilde{\phi}^{(3)}_\alp\|_{L^2}\Big)\\
	\ls &\: \Big( \de \lambda^{-2a} + \de^{-1} \lambda^{\f 12-3a} \Big)\| \rd \widetilde{\phi}^{(2)}_\alp \|_{L^2} \| \rd \widetilde{\phi}^{(3)}_\alp \|_{L^2} \\
	&\: + \de^{-1} \lambda^{1-3a} \Big( \| \Box_g \widetilde{\phi}^{(2)}_\alp \|_{L^2}  \| \rd \widetilde{\phi}^{(3)}_\alp\|_{L^2} + \| \rd \widetilde{\phi}^{(2)}_\alp \|_{L^2} \| \Box_g \widetilde{\phi}^{(3)}_\alp\|_{L^2} \Big),
	\end{split}
	\end{equation}
	where in the last line we used \eqref{eq:null.tphi.cutoff}. We then choose $\de = \lambda^{\f 14 - \f a2}$ (which is $\leq 1$ as required for $\lambda$ sufficiently small) to optimize the above bound so that 
	\begin{equation}\label{eq:use.of.IP}
	\begin{split}
	&\: \Big| \la \rd P_{k_3} \widetilde{\phi}^{(3)}_\alp,  Q_{\mu\nu}(P_{k_1}  \widetilde{\phi}^{(1)}_\alp, P_{k_2} \widetilde{\phi}^{(2)}_\alp) \ra \Big| \\
	\ls &\: \lambda^{\f 14- \f{5a}2} \Big( \| \rd \widetilde{\phi}^{(2)}_\alp \|_{L^2} \| \rd \widetilde{\phi}^{(3)}_\alp \|_{L^2} + \| \Box_g \widetilde{\phi}^{(2)}_\alp \|_{L^2}  \| \rd \widetilde{\phi}^{(3)}_\alp\|_{L^2} + \| \rd \widetilde{\phi}^{(2)}_\alp \|_{L^2} \| \Box_g \widetilde{\phi}^{(3)}_\alp\|_{L^2} \Big).
	\end{split}
	\end{equation}

	\pfstep{Step 4: Concluding the argument} Finally, we sum over all $\alp$ and $(k_1,k_2,k_3)$ in order to bound the term \eqref{eq:null.const.coeff.error}. 
	First, note that the $\ell^2$ sums in \eqref{eq:null.tphi.cutoff} allows us to sum over all $\alp$ to obtain
	\begin{equation}
	\begin{split}
	 \sum_{\alp} \Big| \la \rd P_{k_3} \widetilde{\phi}^{(3)}_\alp,  Q_{\mu\nu}(P_{k_1}  \widetilde{\phi}^{(1)}_\alp, P_{k_2} \widetilde{\phi}^{(2)}_\alp) \ra \Big| 
	\ls &\: \lambda^{\f 14- \f{5a}2}.
	\end{split}
	\end{equation}
	Due to the Fourier cut-offs, there are at most $O((\log \lambda)^3)$ terms in the sum in $(k_1,k_2,k_3)$ (one $\log \lambda$ from each $k_i$). We therefore bound the term in \eqref{eq:null.const.coeff.error} as follows:
	\begin{equation}
	\begin{split}
	\Big| \sum_{\alp} \la f_\alp \rd \widetilde{\phi}^{(3)}_\alp,  Q_{\mu\nu}(\widetilde{\phi}^{(1)}_\alp, \widetilde{\phi}^{(2)}_\alp) \ra \Big| 
	\ls &\:  \sum_{k_1,k_2,k_3: \lambda^{-1+a} \ls 2^{k_i} \ls \lambda^{-1-a}} \lambda^{\f 14- \f{5a}2} 
	\ls  \lambda^{\f 14- \f{5a}2}  (\log \lambda)^3,
	\end{split}
	\end{equation}
	which is acceptable after recalling $a = \f 1{14}$. This concludes the proof. \qedhere
	\end{proof}

	\section{Einstein equations in wave coordinates}\label{sec:EE}
	
	The goal of this section is to derive the equation for $h$; see Proposition~\ref{prop:h.eqn}. In order to achieve this, we carry out some computations for the Ricci curvature. 
		
	\subsection{Conventions}\label{sec:conventions}
	For the remainder of the paper, we use the following conventions.
	
\begin{enumerate}
\item (Conventions on $g_\lambda$, $g_0$ and $h$)
\begin{enumerate}
\item Introduce the notation $g_\lambda$ so that $g_{\lambda_n} = g_n$, where $g_n$ and $\lambda_n$ are as in Assumption~\ref{ass:main}.
\item Whenever defined, we denote $h = g_\lambda - g_0$ (suppressing the explicit $\lambda$-dependence of $h$ for notational convenience).
\end{enumerate}
\item (Coordinates) Coordinates will be denoted by $(x^0,x^1,\cdots,x^d)$ on $\mathbb R^{d+1}$ and will be denoted by $(x^0,x^1,\cdots,x_d,\xi_0,\xi_1,\cdots,\xi_d)$ on its cotangent bundle. We will often write $x^0 = t$. So that there is no confusion with the $_{0}$ in $g_0$, we will often write the corresponding index as $_{t}$, such as $(g_0)_{tt}$, $\xi_t$, etc.
\item (Indices) We will use lower case Greek indices $\alp,\bt,\nu = 0,1,\cdots,d$ and lower case Latin indices $i,j = 1,\cdots, d$. Repeated indices are always summed over the range indicated above.
\item (Inverse) We will use the notation that $g_0^{\alp\bt}$ and $g_\lambda^{\alp\bt}$ denote the inverses of $g_0$ and $g_\lambda$, respectively.
\item (Limits) From now on, we will write $\lim_{\lambda}$ instead of $\lim_{n\to \infty}$, consistent with the convention above. It will always be understood that this limit is taken along as a sequence. In fact, in the argument below, we will need to take subsequences of the given sequence.
\item (Implicit constants) All the implicit constants in $\ls$ are independent of $\lambda$, but could depend on everything else in the problem, $g_0$ (limiting metric), $A$ (pseudo-differential operator), $U$ (domain), $d$ (dimension), etc. We will also use the big-O and small-o notations, to be understood with respect to the $\lim_{\lambda\to 0}$ limit, i.e., for $y \in \mathbb R$, we say $B = O(\lambda^y)$ if $B \ls \lambda^y$ and $B = o(\lambda^y)$ if $\lim_{\lambda \to 0} \lambda^{-y} B = 0$.
\item (Measures and $L^p$ spaces) Unless otherwise specified, all integration are with respect on the volume form induced by the metric $g_0$. In particular, we use $\la \cdot, \cdot, \ra$ to denote the $L^2$ inner product with the volume form induced by $g_0$. In the context below, notice that these norms are compared to those taken with respect to the Lebesgue measure. 
\item (Fourier transforms) We will fix our notation for the (spacetime) Fourier transform as follows:
	$$\mathcal F(f)(\xi) = \int_{\mathbb R^{d+1}} f(x) e^{-ix\cdot \xi} \, \ud x,\quad \mathcal F^{-1}(h)(x) = \f 1{(2\pi)^{d+1}} \int_{\mathbb R^{d+1}} e^{ix\cdot \xi} h(\xi) \, \ud \xi.$$
\end{enumerate}

\subsection{Localization of the problem}
The following lemma is easy to check, and is an immediate consequence of the continuity of $g_0$.
\begin{lemma}
Given any $x \in U$,  there exists a smaller open set $U'$ satisfying $x\in U' \subset U$ such that after a linear change of variables, 
\begin{equation}\label{eq:g.smallness}
|(g_0)_{\alp\bt} - m_{\alp\bt}|,\, |g_0^{\alp\bt} - m^{\alp\bt}| \leq 10^{-100},
\end{equation}
where $m$ is the Minkowski metric. Moreover, assumptions (1)--(4) in Assumption~\ref{ass:main} continue to hold, except that the bounds in \eqref{eq:ass.1}, \eqref{eq:ass.2} and \eqref{eq:wave.coord.assumption} may worsen by a fixed (i.e., $\lambda$-independent) multiplicative constant.
\end{lemma}

\textbf{From now on, we assume that \eqref{eq:g.smallness} hold.} We also expand out Assumption~\ref{ass:main} in the following lemma, now in the conventions defined in Section~\ref{sec:conventions} above.

\begin{lemma}
Under the above conventions and Assumption~\ref{ass:main}, the following pointwise bounds for the metrics hold
\begin{equation}\label{eq:basic.bds.for.Ric}
\begin{split}
|g_0^{-1}|,\,|g_0|,\, |\rd g_0|,\, |\rd^2 g_0|,\, |g_\lambda|,\, |\rd g_{\lambda}|\ls 1, \quad |\rd^2 g_{\lambda}|\ls \lambda^{-1},\quad |h|\ls \lambda, \, |\rd h|\ls 1,\, |\rd^2 h|\ls \lambda^{-1},
\end{split}
\end{equation}
and the following pointwise bounds for $H^\alp = H^\alp(g_\lambda)(\rd g_\lambda)$ and $H_0^\alp = H^\alp(g_0)(\rd g_0)$ hold
\begin{equation}\label{eq:basic.bds.for.H}
 |H_0|,\, |\rd H_0|,\, |H|,\,|\rd H| \ls 1,\quad |H - H_0| \ls \lambda^\eta,\quad |\rd(H - H_0)| \ls 1.
\end{equation}
Using \eqref{eq:basic.bds.for.Ric}, the bound $ |H - H_0| \ls \lambda^\eta$ also implies the following:
\begin{equation}\label{eq:basic.bds.for.WC}
\Big| g_0^{\alp\alp'} \rd_\alp h_{\bt\alp'} - \f 12 g_0^{\alp\alp'} \rd_\bt h_{\alp\alp'} \Big| = O(\lambda^\eta).
\end{equation}
\end{lemma}

\subsection{Computations for the Ricci tensor}

Since these computations are straightforward but tedious, the proofs will be relegated to the appendix. Define the following notation for the Christoffel symbol:
\begin{equation}
\Gamma_{\alp\bt\mu} = \rd_{(\alp} g_{\mu)\bt}  - \f 12 \rd_\bt g_{\mu\nu},\quad \Gamma_{\alp}{}^\sigma{}_{\mu} = g^{\sigma\bt}\Big( \rd_\alp g_{\bt\mu} + \rd_\mu g_{\alp\bt} - \rd_\bt g_{\mu\nu} \Big).
\end{equation}

The Ricci curvature takes the following form. The proof can be found in Appendix~\ref{sec:appendix}.
\begin{lemma}[Ricci curvature in generalized wave coordinates]\label{lem:Ric}
The following identity holds for any $C^2$ Lorentzian metric $g$:
\begin{equation}
2 \mathrm{Ric}(g)_{\mu\nu} = - \widetilde{\Box}_g g_{\mu\nu} + 2 g_{\rho(\mu} \rd_{\nu)} H^\rho(g)(\rd g) + H^\rho(g)(\rd g) \rd_\rho g_{\mu\nu} + B_{\mu\nu}(g)(\rd g,\rd g),
\end{equation}
where\footnote{Notice that $H^\rho (g)(\rd g) = g^{\alp\bt}\Gamma^{\rho}_{\alp\bt}$ so that the notation is consistent with \eqref{eq:generalized.wave.def}.}
\begin{align}
H^\rho (g)(\rd p) \doteq &\: g^{\rho\sigma}g^{\alp\bt}( \rd_\bt p_{\sigma\alp} - \f 12 \rd_{\sigma} p_{\alp\bt}),\\
\widetilde{\Box}_g p_{\mu\nu} \doteq &\: g^{\alp\bt} \rd^2_{\alp\bt} p_{\mu\nu}, \label{eq:tBox.def} \\
B_{\mu\nu}(g)(\rd p,\rd q) \doteq &\: 2 g^{\alp\bt} g^{\sigma\rho} \rd_{(\mu|} p_{\bt\rho} \rd_\sigma q_{|\nu)\alp} - \f 12 g^{\alp\bt} g^{\sigma\rho} \rd_{(\mu|} p_{\bt\rho} \rd_{|\nu)} q_{\alp\sigma} \notag \\
&\: + g^{\alp\bt} g^{\sigma\rho} \rd_\sigma p_{\nu\alp} \rd_\rho q_{\mu\bt} - g^{\alp\bt} g^{\sigma\rho} \rd_{\alp} p_{\nu\sigma} \rd_\rho q_{\mu\bt}.\label{eq:Bmunu.def}
\end{align}
\end{lemma}

We now take the expression in Lemma~\ref{lem:Ric} and derive an equation for $\Box_{g_0} (g_\lambda - g_0)$ by considering $\mathrm{Ric}(g_\lambda) - \mathrm{Ric}(g_0)$. 
In the next few lemmas, we consider the contributions from the different terms in Lemma~\ref{lem:Ric}. The proofs of Lemma~\ref{lem:Ric}--Lemma~\ref{lem:Ric.quasilinear} are relegated to Appendix~\ref{sec:appendix}, while the proof of Lemma~\ref{lem:Ric.H} is straightforward and omitted.

\begin{lemma}[Linear structure of $B_{\mu\nu}$]\label{lem:Ric.linear}
Under the assumption \eqref{eq:basic.bds.for.Ric}, the following holds (see \eqref{eq:Bmunu.def} for definition of $B_{\mu\nu}$):
\begin{equation}
\begin{split}
B_{\mu\nu}(g_\lambda)(\rd g_\lambda,\rd g_\lambda) - B_{\mu\nu}(g_0)(\rd g_0,\rd g_0) 
= &\:  L_{\mu\nu}(g_0)(\rd h)  + B_{\mu\nu}(g_0)(\rd h,\rd h) + O(\lambda),
\end{split}
\end{equation}
where
\begin{equation}\label{eq:L.def}
L_{\mu\nu}(g_0)(\rd h) \doteq 4 g_0^{\sigma\rho}\Gamma_{\rho}{}^{\alp}{}_{(\mu|}(g_0)\rd_{\sigma} h_{|\nu)\alp} + D_{(\mu|}^{\alp\sigma} (g_0)\rd_{|\nu)} h_{\alp\sigma}
\end{equation}
and 
\begin{equation}
D_{\mu}^{\alp\sigma}(g_0) \doteq g_0^{\alp\bt} g_0^{\sigma\rho} (2 \rd_\rho (g_0)_{\bt \mu} - \rd_{\mu} (g_0)_{\bt\rho}).
\end{equation}
\end{lemma}

\begin{lemma}[Nonlinear structure of $B_{\mu\nu}$]\label{lem:Ric.B}
Under the assumption \eqref{eq:basic.bds.for.Ric}, $B_{\mu\nu}$ (see \eqref{eq:Bmunu.def}) can be decomposed as
$$B_{\mu\nu}(g_0)(\rd h,\rd h) = Q_{\mu\nu}(g_0)(\rd h,\rd h) + P_{\mu\nu}(g_0)(\rd h,\rd h) + O(\lambda^\eta),$$
where $Q_{\mu\nu}$ and $P_{\mu\nu}$ are given by
\begin{equation}\label{eq:Q.def}
\begin{split}
Q_{\mu\nu}(g_0)(\rd h,\rd h) \doteq &\: g_0^{\alp\bt} \mathfrak Q^{(g_0)}(h_{\nu\alp}, h_{\mu\bt}) + g_0^{\alp\bt} g_0^{\sigma\rho} \mathfrak Q_{\bt(\mu|}(h_{\sigma\rho}, h_{|\nu)\alp})  \\
&\:+ 2 g_0^{\alp\bt} g_0^{\sigma\rho} \mathfrak Q_{(\mu|\sigma} (h_{\bt\rho}, h_{|\nu)\alp}) - g_0^{\alp\bt} g_0^{\sigma\rho} \mathfrak Q_{\alp\rho}(h_{\nu\sigma}, h_{\mu\bt})
\end{split}
\end{equation}
and
\begin{equation}\label{eq:P.def}
P_{\mu\nu}(g_0)(\rd h, \rd h) \doteq \frac 14 g_0^{\alp\alp'}\rd_\mu h_{\alp \alp'} g_0^{\bt \bt'}\rd_\nu h_{\bt \bt'} - \frac 12 g_0^{\alp\alp'}\rd_\mu h_{\alp \bt} g_0^{\bt \bt'}\rd_\nu h_{\alp' \bt'}.
\end{equation}
In particular, $Q_{\mu\nu}$ is a linear combination of terms satisfying the classical null condition, i.e., it can be written as a linear combination of the null forms in Definition~\ref{def:Q0} and Definition~\ref{def:Qmunu}.
\end{lemma}

\begin{lemma}[Quasilinear terms]\label{lem:Ric.quasilinear}
Under the assumption \eqref{eq:basic.bds.for.Ric}, the following holds:
$$\widetilde{\Box}_{g_\lambda} (g_{\lambda})_{\mu\nu} - \widetilde{\Box}_{g_0} (g_0)_{\mu\nu} = \widetilde{\Box}_{g_0} h_{\mu\nu} - g_0^{\alp\alp'}g_0^{\bt\bt'} h_{\alp'\bt'} \rd^2_{\alp\bt} h_{\mu\nu} + O(\lambda).$$
\end{lemma}

\begin{lemma}[Terms related to the wave coordinate condition]\label{lem:Ric.H}
Under the assumptions \eqref{eq:basic.bds.for.Ric} and \eqref{eq:basic.bds.for.H}, the following holds:
\begin{equation*}
\begin{split}
&\: 2 (g_\lambda)_{\rho(\mu} \rd_{\nu)} H^\rho + H^\rho \rd_\rho (g_\lambda)_{\mu\nu} -  (2 (g_0)_{\rho(\mu} \rd_{\nu)} H_0^\rho + H_0^\rho \rd_\rho (g_0)_{\mu\nu}) \\
 =&\: 2 (g_0)_{\rho(\mu} \rd_{\nu)} (H^\rho - H_0^\rho) + H_0^\rho \rd_\rho h_{\mu\nu} + O(\lambda^{\eta}).
 \end{split}
 \end{equation*}
\end{lemma}

Putting together the above lemmas, we obtain the following main wave equation for $h$:
\begin{proposition}[Structure of $\Box_{g_0} h$]\label{prop:h.eqn}
Under the assumptions \eqref{eq:basic.bds.for.Ric} and \eqref{eq:basic.bds.for.H}, the following holds:
\begin{equation}\label{eq:h.eqn}
\begin{split}
\Box_{g_0} h_{\mu\nu} = &\: 2 \mathrm{Ric}(g_0)_{\mu\nu} + g_0^{\alp\alp'}g_0^{\bt\bt'} h_{\alp'\bt'} \rd^2_{\alp\bt} h_{\mu\nu} + L_{\mu\nu}(g_0)(\rd h) \\
&\: + Q_{\mu\nu}(g)(\rd h,\rd h) + P_{\mu\nu}(g)(\rd h,\rd h) + 2 (g_0)_{\alp(\mu} \rd_{\nu)} (H^\alp - H^\alp_0) + O(\lambda^\eta),
\end{split}
\end{equation}
where $\Box_{g_0}$ is understood as the wave operator for scalar functions, i.e., $\Box_{g_0} = g_0^{\alp\bt} \rd^2_{\alp\bt} - H^\alp(g_0)(\rd g_0)\rd_\alp$. In particular, we also have
\begin{equation}\label{eq:Boxh.bound}
|\Box_{g_0} h| \ls 1.
\end{equation}
\end{proposition}
\begin{proof}
This is an immediate consequence of Lemmas~\ref{lem:Ric}, \ref{lem:Ric.linear}, \ref{lem:Ric.B}, \ref{lem:Ric.quasilinear} and \ref{lem:Ric.H}. Notice that we combine $\widetilde{\Box}_{g_0} h_{\mu\nu}$ and the term $H_0^\rho \rd_\rho h_{\mu\nu}$ in Lemma~\ref{lem:Ric.H} to form $\Box_{g_0} h_{\mu\nu}$. \qedhere
\end{proof}
	
	\section{Microlocal defect measure and the limiting Ricci curvature}\label{sec:mdm}
	
	We now begin the proof of Theorem~\ref{thm:main}. We work under Assumption~\ref{ass:main} and \eqref{eq:g.smallness}, and continue to use the conventions introduced in Section~\ref{sec:conventions}.
	
	Our goal in this section will be to prove parts (1)--(3) of Theorem~\ref{thm:main}.
	
\subsection{Definition of the microlocal defect measure}

\begin{definition}\label{def:mu.abrs}
Define $\mu_{\alp\bt\rho\sigma}$ so that the following holds for all $A\in \Psi^0$ with principal symbol $a(x,\xi)$ which is real, positively $0$-homogeneous and has compact support in $x$:
$$\lim_{\lambda\to 0} \la \rd_\gamma h_{\alp\bt}, A \rd_\de h_{\rho\sigma} \ra_{L^2} = \int_{S^* U} a \,\xi_\gamma \xi_\de \ud \mu_{\alp\bt\rho\sigma},$$
where we take $\mu_{\alp\bt\rho\sigma}$ to be a measure on $S^* U$ by acting on functions which are positively $2$-homogeneous in $\xi$ (see Definition~\ref{def:measure}), and the $\lim_{\lambda \to 0}$ limit is to be understood after passing to a (further) subsequence. Given $\mu_{\alp\bt\rho\sigma}$, we then define $\mu$ by \eqref{eq:mu.def}.
\end{definition}

\begin{rk}\label{rmk:existence.of.mu}
\begin{enumerate}
\item (Existence of $\mu_{\alp\bt\rho\sigma}$) The existence of the measures $\mu_{\alp\bt\rho\sigma}$ follows from the standard theory of microlocal defect measures. By our assumptions, $\rd h$ is $L^2$-bounded uniformly in $\lambda$. Thus, results in \cite{Gerard,Tartar} show that after passing to a subsequence, there is a measure $\widetilde{\mu}_{\gamma\alp\bt\de\rho\sigma}$ such that 
$$\lim_{\lambda\to 0} \la \rd_\gamma h_{\alp\bt}, A \rd_\de h_{\rho\sigma} \ra_{L^2} = \int_{S^* U} a  \ud \widetilde{\mu}_{\gamma\alp\bt\de\rho\sigma},$$
with $\widetilde{\mu}_{\gamma\alp\bt\de\rho\sigma}$ acting on positively $0$-homogeneous functions. Finally, using localization lemma for microlocal defect measures \cite[Corollary~2.2]{Gerard} and the commutation of mixed partials, it can be shown that $\widetilde{\mu}_{\gamma\alp\bt\de\rho\sigma} = \xi_\gamma \xi_\de \mu_{\alp\bt\rho\sigma}$. A similar and more detailed argument can be found for instance in \cite{FrancfortMurat} or \cite[Sections~6.1, 6.2]{HL.Burnett}.
\item (Symmetries of $\mu_{\alp\bt\rho\sigma}$) It is easy to check that
\begin{equation}\label{eq:mu.obvious.symmetries}
a \mu_{\alp\bt\rho\sigma} = a \mu_{\rho\sigma\alp\bt} = a \mu_{\bt\alp\rho\sigma}.
\end{equation}
Indeed, the first equality follows from the fact that $a$ is also the principal symbol of $A^*$ and the second equality follows from the fact that $h$ is a symmetric tensor. 
\end{enumerate}
\end{rk}

The following fact will be useful, and is a consequence of the generalized wave coordinate condition.
\begin{lemma}\label{lem:wave.coord.for.MLDM}
For any indices $\gamma,\sigma,\nu$,
\begin{equation}\label{eq:wave.coord.for.MLDM}
g_0^{\alp\bt} \xi_\alp \mu_{\bt \gamma\sigma\nu} - \f 12 g_0^{\alp\bt} \xi_\gamma \mu_{\alp\bt\sigma\nu} = 0.
\end{equation}
\end{lemma}
\begin{proof}
This follows from \eqref{eq:basic.bds.for.WC} and Definition~\ref{def:mu.abrs}. \qedhere
\end{proof}

\subsection{Proof of (1)--(3) in Theorem~\ref{thm:main}}
With the definition of $\mu$, we now prove parts (1)--(3) of Theorem~\ref{thm:main}. We start with showing that $\mu$ is supported on the null cone.

\begin{proposition}\label{prop:supp.null.cone}
$$g_0^{\alp\bt} \xi_\alp\xi_\bt \, \ud \mu \equiv 0.$$
\end{proposition}
\begin{proof}
We prove the stronger statement that $g_0^{\alp\bt} \xi_\alp\xi_\bt \, \ud \mu_{\rho\sigma\gamma\de} = 0$ for any indices $\rho$, $\sigma$, $\gamma$, $\de$. This is easy and well-known; see for instance \cite{FrancfortMurat}. Since $h\to 0$ (strongly in $L^2$) and $\Box_{g_0}h$ is uniformly bounded by \eqref{eq:Boxh.bound}
\begin{equation}\label{eq:hBoxhto0}
\begin{split}
\la h_{\rho\sigma}, A\Box_{g_0}  h_{\gamma\de} \ra \to 0.
\end{split}
\end{equation}
On the other hand, integrating by parts and using $[A,\Box_{g_0}] \in \Psi^{1}$, we obtain
\begin{equation}
\la h_{\rho\sigma}, A\Box_{g_0}  h_{\gamma\de} \ra = \la g_0^{\alp\bt} \rd_\alp h_{\rho\sigma}, A \rd_\bt h_{\gamma\de} \ra + o(1) \to \int_{S^* U} a g_0^{\alp\bt} \xi_\alp \xi_\bt \, \ud \mu_{\rho\sigma\gamma\de}.
\end{equation}
Since $a$ is arbitrary, it follows that $ g_0^{\alp\bt} \xi_\alp \xi_\bt \, \ud \mu_{\rho\sigma\gamma\de} \equiv 0$, as desired. \qedhere
\end{proof}

\begin{rk}
Notice that in the proof of Proposition~\ref{prop:supp.null.cone}, we used the uniform boundedness of $\Box_{g_0}h$, which in particular required the assumed bound on $\rd^2 h$ on $\rd H$. We mentioned in Remark~\ref{rmk:improved.assumption} that the assumptions on $\rd^2 h$ and $\rd H$ are not necessary. To see this, consider the unbounded term in \eqref{eq:h.eqn} involving $\rd^2 h$, i.e., $g_0^{\alp\alp'}g_0^{\bt\bt'} h_{\alp'\bt'} \rd^2_{\alp\bt} h_{\gamma\de}$. Its contribution to \eqref{eq:hBoxhto0} is
$$\la h_{\rho\sigma}, A g_0^{\alp\alp'}g_0^{\bt\bt'} h_{\alp'\bt'} \rd^2_{\alp\bt} h_{\gamma\de} \ra,$$
which also $\to 0$, which can be seen after integration by parts. A similar comment applies to the term involving $\rd H$.
\end{rk}

Next, we turn to part (2) of Theorem~\ref{thm:main}.

\begin{proposition}\label{prop:Ricci.limit}
The limiting Ricci curvature is given by
$$\int_{U} \psi \mathrm{Ric}_{\mu \nu}(g_0) \, \mathrm{dVol}_{g_0} = \int_{S^* U} \psi \xi_\mu \xi_\nu \ud \mu,\quad \forall \psi \in C^\infty_c(U),$$
where $\mu= g_0^{\alpha\rho}g_0^{\beta \sigma}(\frac{1}{4}\mu_{\rho \beta \alpha \sigma} -\frac{1}{8}\mu_{\rho \alpha \beta \sigma}).$
\end{proposition}
\begin{proof}
We will use without comment the standard facts (1) $p_\lambda \to p$ uniformly and $q_\lambda \rightharpoonup q$ weakly in $L^2$ implies $p_\lambda q_\lambda \to pq$ in distribution and (2) $p_\lambda \to p$ in distribution implies all its derivatives converge in distribution.

We now consider Proposition~\ref{prop:h.eqn}, thought of as an expression of $\mathrm{Ric}(g_0)$, pair it with some $\psi \in C^\infty_c(\mathbb R^{d+1})$ and consider the $\lim_{\lambda\to 0}$ limit.

First note that the terms $\Box_{g_0} h$ and $L(g_0)(\rd h)$ are \underline{linear} in (the derivatives of) $h$ thus converge weakly to $0$ as $\lambda \to 0$. Similarly, the term $2 (g_0)_{\alp(\mu} \rd_{\nu)} (H^\alp - H^\alp_0)$ is linear in the derivative of $H^\alp - H^\alp_0$ and thus converges weakly to $0$ as $\lambda \to 0$. 

For the quasilinear term $g_0^{\alp\alp'}g_0^{\bt\bt'} h_{\alp'\bt'} \rd^2_{\alp\bt} h_{\mu\nu}$, we note that 
\begin{equation}\label{eq:quasilinear.hidden.null.structure}
\begin{split}
 g_0^{\alp\alp'}g_0^{\bt\bt'} h_{\alp'\bt'} \rd^2_{\alp\bt} h_{\mu\nu} = &\: \rd_\alp (g_0^{\alp\alp'}g_0^{\bt\bt'} h_{\alp'\bt'} \rd_{\bt} h_{\mu\nu} ) - g_0^{\alp\alp'}g_0^{\bt\bt'} \rd_\alp h_{\alp'\bt'} \rd_{\bt} h_{\mu\nu} + O(\lambda) \\
 = &\: \underbrace{\rd_\alp (g_0^{\alp\alp'}g_0^{\bt\bt'} h_{\alp'\bt'} \rd_{\bt} h_{\mu\nu} )}_{=:I} - \underbrace{ \f 12 g_0^{\alp\alp'}g_0^{\bt\bt'} \rd_{\bt'} h_{\alp'\alp} \rd_{\bt} h_{\mu\nu}}_{=:II} + O(\lambda^\eta) ,
\end{split}
\end{equation}
where in the second equality we used the wave coordinate condition \eqref{eq:basic.bds.for.WC}. Notice now that term $I$ is a total derivative of an $O(\lambda)$ term and thus converges weakly to $0$; term $II$ also converges weakly to $0$ because it can be written as $\f 12 g_0^{\alp\alp'} \mathfrak Q^{(g)}_0(h_{\alp'\alp}, h_{\mu\nu})$.

For the null forms, we simply notice that $\mathfrak Q^{(g)}_0(\phi,\psi) = \f 12 \Box_{g_0} (\phi\psi) - \f 12 \phi \Box_{g_0} \psi - \f 12 \psi \Box_{g_0} \phi$, and that $\mathfrak Q_{\alp\bt}(\phi,\psi) = \rd_\alp (\phi \rd_\bt \psi) - \rd_\bt (\phi \rd_\alp \psi)$ so that in both cases $\mathfrak Q(\phi_\lambda,\psi_\lambda) \rightharpoonup \mathfrak Q(\phi,\psi)$.

It thus follows that the only contribution from Proposition~\ref{prop:h.eqn} that does not converge weakly to $0$ comes from the term $P_{\mu\nu}(g)(\rd h,\rd h)$. However, this contribution is exactly $\int_{S^* \mathbb R^{d+1}} \psi \xi_\mu \xi_\nu \ud \mu$ by \eqref{eq:P.def}, \eqref{eq:mu.def} and Definition~\ref{def:mu.abrs}. \qedhere
\end{proof}

Next, we turn to non-negativity of $\mu$, i.e., statement (3) of the main theorem (Theorem~\ref{thm:main}). Let us remark that once we have obtained Proposition~\ref{prop:Ricci.limit}, the non-negativity of $\mu$ already follows from the fact, established in \cite{GW1}, that the weak energy condition holds for the limiting spacetime (irrespective of gauge conditions). Here, however, we give a direct proof of the non-negativity.

\begin{proposition}\label{prop:positive}
The measure $\mu$ is real-valued and non-negative.
\end{proposition}
\begin{proof}
\pfstep{Step~1: Definition of null frame adapted to $\xi$} 
Given $(x,\xi_\alp)$ on the support of $\mu$, denote $n = g_0^{t\alp} \rd_\alp$, $\xi^\alp = g_0^{\alp\bt} \xi_\bt$ and introduce the following $\xi$-dependent vector fields:
\begin{equation}\label{eq:def.frame}
L^{(\xi)} \doteq g_0^{\mu\nu} \xi_\mu \rd_\nu,\quad \underline{L}^{(\xi)} \doteq \f{g_0^{tt}}{(\xi^t)^2}L^{(\xi)} - \f 2{\xi^t} n,\quad \srd_i^{(\xi)} \doteq  \rd_i + \f{\xi_i}{2} \Lbxi + \f{g_0^{tt} \xi_i}{2 (\xi^t)^2} \Lxi.
\end{equation}
These vector fields are well-defined on the support of $\ud \mu$ since $\xi^t \neq 0$ (by Proposition~\ref{prop:supp.null.cone} and \eqref{eq:g.smallness}). These vector fields are chosen so that the following relations hold:
\begin{equation}\label{eq:xi.null.frame.cond}
g_0(L^{(\xi)},L^{(\xi)}) = g_0(\underline{L}^{(\xi)},\underline{L}^{(\xi)}) = g_0(L^{(\xi)},\srd_i^{(\xi)}) = g_0(\Lbxi,\srd_i^{(\xi)}) = 0,\quad g_0(L^{(\xi)},\underline{L}^{(\xi)}) = -2.
\end{equation}
The relations \eqref{eq:xi.null.frame.cond} can be checked by direct computations, after noting that 
\begin{align*}
g_0(L^{(\xi)}, L^{(\xi)}) = \xi^\alp \xi_\alp, \quad g_0(L^{(\xi)}, \rd_i) =  g_0^{\mu\nu} \xi_\mu (g_0)_{\nu i} = \xi_i,\quad g_0(\Lbxi, \rd_i) = \f{g_0^{tt}\xi_i}{(\xi^t)^2}, \\
g_0(L^{(\xi)},n) = (g_0)_{\alp\nu} g_0^{0\alp} g_0^{\mu\nu} \xi_\mu = \xi^0, \quad g_0(n,n)= (g_0)_{\alp\alp'} g_0^{t\alp}g_0^{t\alp'} = g_0^{00},\quad g_0(n,\rd_i) = 0.
\end{align*}

\pfstep{Step~2: Some computations} Before we proceed, we collect some computations. From now on, denote by $\underline{g}_0$ the spatial part of $g$. In particular, $\underline{g}_0^{ij}$ denotes the $(ij)$-component of the inverse of $\underline{g}_0$. It relates to the $g_0^{ij}$ coming from the inverse of $g_0$ through
\begin{equation}\label{eq:ug.inverse}
\underline{g}_0^{ij} = g_0^{ij} - \f{g_0^{ti} g_0^{tj}}{g_0^{tt}},
\end{equation}
which can be derived by standard formulas on inverses of block diagonal matrices. It is also convenient to note the following rewritings of $g_0^{\mu\nu} \xi_\mu \xi_\nu = 0$ (on $\mathrm{supp}(\mu)$, see Proposition~\ref{prop:supp.null.cone}):
\begin{equation}\label{eq:other.ways.to.say.null}
\underline{g}_0^{ij} \xi_i \xi_j = -g_0^{tt} (\xi_t + \f{g^{ti}}{g^{tt}} \xi_i)^2 = - \f{(\xi^t)^2}{g_0^{tt}}.
\end{equation}

Next, we compute that $\Lxi t = \xi^t$, $\Lbxi t =  - \f{g_0^{tt}}{\xi^t}$, $\Lxi x^\ell = \xi^\ell$ and $\Lbxi x^\ell = \f{g_0^{tt}\xi^\ell}{(\xi^t)^2} - \f{2g_0^{ti}}{\xi^t}$. From this, and using \eqref{eq:ug.inverse}, \eqref{eq:other.ways.to.say.null}, it follows that $\srd_i t= 0$ (so that $\srd_i$ is tangential to constant-$t$ hypersurfaces), and that $\underline{g}_0^{ij} \xi_j \srd_i t = \underline{g}_0^{ij} \xi_j \srd_i x^\ell = 0$. In particular, $\{\srd_1, \cdots, \srd_d\}$ are linearly dependent, as they satisfy the relation
\begin{equation}\label{eq:srd.dependence}
\underline{g}_0^{ij} \xi_j \srd_i = 0.
\end{equation}
Notice also that since $\{\Lxi,\Lbxi, \srd_1^{(\xi)},\cdots, \srd_d^{(\xi)}\}$ obviously span the whole tangent space, locally we can find $d-1$ linearly independent and spacelike elements in $\{\srd_1^{(\xi)},\cdots, \srd_d^{(\xi)}\}$ and perform Gram--Schmidt to obtain an orthonormal frame $\{\slashed{e}_B^{(\xi)}\}_{B=1}^{d-1}$ so that $\mathrm{span}\{\slashed{e}_B^{(\xi)}\}_{B=1}^{d-1} = \mathrm{span}\{\srd_i^{(\xi)}\}_{i=1}^{d}$ and such that $g_0(\slashed{e}_B^{(\xi)}, \slashed{e}_C^{(\xi)}) = \de_{BC}$. From now on, fixed such a local orthonormal frame $\{\slashed{e}_B^{(\xi)}\}_{B=1}^{d-1}$. We will use the convention that capital Latin indices run over $B,C=1,\cdots,d-1$ and that repeated indices will be summed over this range. 

Using such an orthonormal frame and recalling the relations \eqref{eq:xi.null.frame.cond}, we can write the inverse metric as follows:
\begin{equation}\label{eq:g0.in.terms.of.frame}
\begin{split}
g_0^{\alp\bt} \rd_\alp\otimes \rd_\bt 
= &\: -\f 12 L^{(\xi)}\otimes \underline{L}^{(\xi)} - \f 12 \underline{L}^{(\xi)}\otimes L^{(\xi)} + \de^{BC} \slashed{e}_B^{(\xi)}\otimes \slashed{e}_C^{(\xi)}.
\end{split}
\end{equation}

\pfstep{Step~3: Analyzing the microlocal defect measure using the orthonormal frame} We recall \eqref{eq:wave.coord.for.MLDM} and contract it with $(\Lxi)^\gamma$, $(\slashed{e}_A^{(\xi)})^\gamma$ and $(\Lbxi)^\gamma$ to get
\begin{align}
(\Lxi)^\gamma g_0^{\alp\bt} \xi_\alp \mu_{\bt \gamma\sigma\nu} - \f 12 g_0^{\alp\bt} (\Lxi)^\gamma \xi_\gamma \mu_{\alp\bt\sigma\nu} = &\: 0, \label{eq:mu.WC.L}\\
(\slashed{e}_A^{(\xi)})^\gamma g_0^{\alp\bt} \xi_\alp \mu_{\bt \gamma\sigma\nu} - \f 12 g_0^{\alp\bt} (\slashed{e}_A^{(\xi)})^\gamma \xi_\gamma \mu_{\alp\bt\sigma\nu} =&\:  0, \label{eq:mu.WC.e}\\
(\Lbxi)^\gamma g_0^{\alp\bt} \xi_\alp \mu_{\bt \gamma\sigma\nu} - \f 12 g_0^{\alp\bt} (\Lbxi)^\gamma \xi_\gamma \mu_{\alp\bt\sigma\nu} =  &\: 0.\label{eq:mu.WC.Lb}
\end{align}
Since $(\Lxi)^\gamma \xi_\gamma = 0$ and $(\slashed{e}_A^{(\xi)})^\gamma \xi_\gamma = 0$ (by \eqref{eq:def.frame}, \eqref{eq:xi.null.frame.cond}), after recalling $\Lxi$ in \eqref{eq:def.frame}, we use \eqref{eq:mu.WC.L} and \eqref{eq:mu.WC.e} to obtain
\begin{equation}\label{eq:mu.cancel.in.frame}
(L^{(\xi)})^\bt (L^{(\xi)})^\gamma \mu_{\bt \gamma\sigma\nu},\, (L^{(\xi)})^\bt (\slashed{e}_B^{(\xi)})^\gamma \mu_{\bt \gamma\sigma\nu} = 0,\quad \forall \sigma,\nu \in \{0,1,\cdots,d\},\,\forall B \in \{1,\cdots,d-1\}.
\end{equation}
Since $(\Lbxi)^\gamma \xi_\gamma = -2$ (by \eqref{eq:def.frame}, \eqref{eq:xi.null.frame.cond}), we use \eqref{eq:mu.WC.Lb} to obtain
\begin{equation}\label{eq:mu.cancel.in.frame.2}
\begin{split}
0 =  (\Lbxi)^\gamma (\Lxi)^\bt \mu_{\bt \gamma\sigma\nu} + g_0^{\alp\bt}  \mu_{\alp\bt\sigma\nu} = &\: \de^{BC} (\slashed{e}^{(\xi)}_B)^\alp  (\slashed{e}^{(\xi)}_C)^\bt \mu_{\alp\bt\sigma\nu}.
\end{split}
\end{equation}

We now expand the $g_0$ in $\mu = g_0^{\alpha\rho}g_0^{\beta \sigma}(\frac{1}{4}\mu_{\rho \beta \alpha \sigma} -\frac{1}{8}\mu_{\rho \alpha \beta \sigma})$ using \eqref{eq:g0.in.terms.of.frame} and \eqref{eq:mu.cancel.in.frame} to get
\begin{equation}
\begin{split}
\mu = &\: (\f 14 \de^{BB'} \de^{CC'} - \f 18 \de^{BC} \de^{B'C'}) (\slashed{e}^{(\xi)}_B)^\alp  (\slashed{e}^{(\xi)}_C)^\bt (\slashed{e}^{(\xi)}_{B'})^\sigma (\slashed{e}^{(\xi)}_{C'})^{\sigma} \mu_{\alp\bt\rho\sigma} \\
= &\: (\f 14 \de^{BB'} \de^{CC'} - \f 18 \de^{BC} \de^{B'C'})  \mu_{\slashed{B}\slashed{C}\slashed{B'}\slashed{C'}},
\end{split}
\end{equation}
where the second line defines the notation $\mu_{\slashed{B}\slashed{C}\slashed{B'}\slashed{C'}}$. Note that $(\Lxi)^\rho (\Lbxi)^\bt (\Lxi)^\alp (\Lbxi)^\sigma \mu_{\rho \beta \alpha \sigma}$ cancels. Using also \eqref{eq:mu.cancel.in.frame.2}, we obtain
\begin{equation}\label{eq:mu.good.form}
\mu =  \f 14 \de^{BB'} \de^{CC'} \mu_{\slashed{B}\slashed{C}\slashed{B'}\slashed{C'}}.
\end{equation}

\pfstep{Step~4: Completing the proof} Since $\xi_t \neq 0$ on $\mathrm{supp}(\mu)$, in order to prove the proposition, it suffices to show that for any $\chi(x,\xi)$ which is real, compactly supported in $x$ and positively $0$-homogeneous in $\xi$,
$$ \int_{S^* \mathbb R^{d+1}} \chi^4(x,\xi) \xi_t^2 \, \ud \mu \geq 0.$$
By cutting off further if necessary, we assume the local orthonormal frame $\{\slashed{e}^{(\xi)}_B \}_{B=1}^{d-1}$ is well-defined on $\mathrm{supp}(\chi)$.

For $B=1,\cdots,d-1$, let $A_{\chi(\slashed{e}_B^{(\xi)})^\alp}$ be a $0$-th order pseudo-differential operator with principal symbol $\chi(x,\xi)(\slashed{e}_B^{(\xi)})^\alp$ (which is well-defined on $\mathrm{supp}(\chi)$ and is positively $0$-homogeneous in $\xi$). Then, using \eqref{eq:mu.good.form}, 
\begin{equation*}
\begin{split}
 \int_{S^* \mathbb R^{d+1}} \chi^4(x,\xi) \xi_t^2 \, \ud \mu
= &\:  \f 14 \lim_{\lambda\to 0} \int_{\mathbb R^{d+1}} \Big( A_{\chi(\slashed{e}_B^{(\xi)})^\alp} \circ  A_{\chi(\slashed{e}_C^{(\xi)})^\bt} \rd_t h_{\alp\bt} \Big)^2,
\end{split}
\end{equation*}
which is manifestly non-negative. \qedhere
\end{proof}

\section{Propagation of the microlocal defect measure}\label{sec:propagation}

We continue to work under Assumption~\ref{ass:main} and \eqref{eq:g.smallness}, and use the conventions introduced in Section~\ref{sec:conventions}.

The goal of this section is to prove the following theorem. This proves part (4) of Theorem~\ref{thm:main} and thus completes the proof of the main theorem.
\begin{theorem}\label{thm:main.transport}
The following identity holds for any smooth $\widetilde{a}:S^*U \to \mathbb R$ which is compactly supported in $x$ and positively $1$-homogeneous in $\xi$:
\begin{equation}\label{eq:main.transport}
\int_{S^* U} \{g_0^{\mu\nu} \xi_\mu\xi_\nu, \widetilde{a}(x,\xi) \}  \, \ud \mu = 0.
\end{equation}
\end{theorem}

\textbf{For the remainder of the section, fix $\widetilde{a}$ that satisfies the assumption of Theorem~\ref{thm:main.transport} and define $a(x,\xi) = \f{\widetilde{a}(x,\xi)}{\xi_t}$.} (Note that this is well-defined because $\xi_t \neq 0$ on $\mathrm{supp}(\mu)$ by \eqref{eq:g.smallness} and Proposition~\ref{prop:supp.null.cone}.)

\subsection{Cutting off $h$}\label{sec:cutoff.assumptions}

Before proceeding, we first introduce another reduction. For the fixed $a$ (or equivalently $\widetilde{a}$) above, its spatial support is contained in $K \subset U$. We will fix the compact set $K$ and introduce cutoffs with respect to this $K$.

For $K$ as above, fix another open set $K' \subset U$ such that $K \subset \mathring{K}'$. Let $\chi$ be a smooth cutoff function such that $\mathrm{supp}(\chi) \subset K'$ and $\chi \equiv 1$ on $K$. From now on, we replace $h$ by $\chi h$ so that it is $C^\infty_c$, which makes taking Fourier transforms easier. Moreover, \textbf{we can now work globally in the whole space $\mathbb R^{d+1}$}, with $g_0$ extended outside $\mathring{K}'$ so that \eqref{eq:g.smallness} holds globally and the $C^k$ norm is globally controlled for all $k\in \mathbb N$. The choice of the extension of $g_0$ will not change the derivation of \eqref{eq:main.transport}.

Notice that after introducing the cutoff, all the estimates \eqref{eq:basic.bds.for.Ric}, \eqref{eq:basic.bds.for.H}, \eqref{eq:basic.bds.for.WC} and \eqref{eq:Boxh.bound} still hold. However, the equation \eqref{eq:h.eqn} no longer holds globally in $U$, but importantly it holds on $\mathrm{supp}(a)$. This will already be sufficient in the proof of Proposition~\ref{prop:general.wave.transport}.

\textbf{We will work under these cutoff assumptions for the remainder of the paper.}

\subsection{Main identity for the propagation of $\mu$}

In this subsection, we derive the main propagation identity for $\mu$; see Proposition~\ref{prop:main.reduced} below. We begin with a general lemma.

\begin{lemma}\label{lem:waveEE}
Let $g_0$ be as before. Let $\phi_\lambda$, $\psi_\lambda$ be smooth functions supported in a fixed compact set in $\mathbb R^{d+1}$ which (1) are uniformly bounded in $H^1$, (2) satisfy $\phi_\lambda,\,\psi_\lambda\to 0 $ in $L^2$, and (3) are such that $\Box_{g_0} \phi_\lambda$, $\Box_{g_0} \psi_\lambda$ are uniformly bounded in $L^2$.

Define $\ud \mu_{\phi\psi}$ to be the cross microlocal defect measure, i.e., for any $A \in \Psi^{0}$ with positively $0$-homogeneous. principal symbol $a$, (up to a subsequence)
$$\lim_{\lambda\to 0} \la \rd_\alp \phi_\lambda, A \rd_\bt \psi \ra =  \int_{S^*\mathbb R^{d+1}} a \xi_\alp \xi_\bt \, \ud \mu_{\phi\psi}.$$

Then for any $A \in \Psi^{0}$ whose principal symbol is a real, positively $0$-homogeneous Fourier multiplier $m(\xi)$, and any vector field $X$, we have
$$\f 12 \lim_{\lambda \to 0}\Big(\la \Box_{g_0} \phi_\lambda, X A\psi_\lambda\ra + \la X \phi_\lambda, A \Box_{g_0} \psi_\lambda \ra \Big) +\f 12 \int_{S^*\RR^{d+1}} \{g_0^{\mu\nu} \xi_\mu\xi_\nu,  X^\rho (x) \xi_\rho m(\xi) \}  \, \ud \mu_{\phi\psi}= 0.$$
\end{lemma}
\begin{proof}
In this proof, we write $g = g_0$, $\phi = \phi_\lambda$, $\psi = \psi_\lambda$ whenever it does not create confusion.

Define
$$\mathbb T^A_{\mu\nu}[\phi, \psi] \doteq \rd_{(\mu|} \phi \rd_{|\nu)} A \psi +  \f 12 g_{\mu\nu} g^{\alp\bt} \rd_{(\alp|} \phi \rd_{|\bt)} A \psi. $$
Let $\nabla$ be the Levi-Civita connection associated to $g_0$. It is easy to check that 
$$g^{\rho\mu} \nabla_\rho \mathbb T^A_{\mu\nu}[\phi, \psi] = \f 12 \Box_{g} \phi \rd_\nu A\psi + \f 12 \rd_\nu \phi \Box_{g} A\psi.$$
It then easily follows that 
\begin{equation}\label{eq:id.for.div.thm}
g^{\rho\mu} \nabla_\rho (\mathbb T^A_{\mu\nu}[\phi, \psi] X^\nu) = \f 12 \Box_{g} \phi X A\psi + \f 12 X \phi A \Box_g \psi + \f 12 X\phi [\Box_{g}, A] \psi+ ^{(X)}\pi^{\mu\nu} \mathbb T^A_{\mu\nu}[\phi, \psi] ,
\end{equation}
where $^{(X)}\pi^{\mu\nu} \doteq \nabla^{(\mu|} X^{|\nu)}$ is the deformation tensor of $X$.

Integrating \eqref{eq:id.for.div.thm} in the whole space using Stoke's theorem, and taking the $\lambda \to 0$ limit, we obtain
\begin{equation}
\begin{split}
0 = &\: \f 12 \lim_{\lambda \to 0}\Big(\la \Box_{g} \phi_\lambda, X A\psi_\lambda\ra + \la X \phi_\lambda, A \Box_{g} \psi_\lambda \ra \Big) \\
&\: + \f 12 \int_{S^*\RR^{d+1}}  X^\alp \xi_\alp \{g^{\mu\nu} \xi_\mu\xi_\nu, m \} \, \ud \mu_{\phi\psi} + \int_{S^*\RR^{d+1}} m {}^{(X)}\pi^{\mu\nu} \xi_\mu\xi_\nu \, \ud \mu_{\phi\psi},
\end{split}
\end{equation} 
where we have used that $g^{\alp\bt} \xi_\alp \xi_\bt = 0$ on the support of $\mu_{\phi\psi}$.

We then compute
$$^{(X)}\pi^{\mu\nu} = \nabla^{(\mu|} X^{|\nu)} = g^{(\mu|\mu'} \rd_{\mu'} X^{|\nu)} + g^{(\mu|\mu'}\Gamma^{|\nu)}_{\mu'\rho} X^\rho = g^{(\mu|\mu'} \rd_{\mu'} X^{|\nu)} - \f 12 X^\rho \rd_{x^\rho} g^{\mu\nu},$$
where the last equation can be obtained using $\nabla_{\rho} g_0^{\mu\nu} = 0$.

Finally, we compute
\begin{equation*}
\begin{split}
&\: \{g^{\mu\nu} \xi_\mu\xi_\nu,  X^\rho (x) \xi_\rho m(\xi) \} \\
= &\: 2 g^{\mu\nu} \xi_\nu  \xi_\rho m(\xi) \rd_{x^\mu} X^\rho - \rd_{x^\bt} g^{\mu\nu} X^\bt (x) m(\xi) - \rd_{x^\bt} g^{\mu\nu} X^\rho (x) \xi_\rho \rd_\bt m(\xi) \\
= &\: 2 { }^{(X)}\pi^{\rho\nu}\xi_\nu \xi_\rho +  X^\rho \xi_\rho \{g^{\mu\nu} \xi_\mu\xi_\nu, m(\xi) \}.
\end{split}
\end{equation*}

Putting all these together yields the conclusion. \qedhere
\end{proof}

We now return to the setting of part (4) of Theorem~\ref{thm:main}, imposing, in addition, the assumptions in Section~\ref{sec:cutoff.assumptions}.

\begin{proposition}\label{prop:general.wave.transport}
Let $A \in \Psi^0$ with principal symbol $a(x,\xi)$ which is real. Define $\widetilde{a}(x,\xi) = \xi_t a(x,\xi)$. Then the following identity holds:
\begin{equation}\label{eq:general.wave.transport}
\begin{split}
\int_{S^*\RR^{d+1}} \{g_0^{\mu\nu} \xi_\mu\xi_\nu, \widetilde{a}(x,\xi) \}  \, \ud \mu = &\: \f 14\int_{S^*\RR^{d+1}} g_0^{\mu\nu} \xi_\nu \widetilde{a}(x,\xi)  \rd_{x^\mu} (2g_0^{\alp\alp'} g_0^{\bt\bt'} - g_0^{\alp\bt} g_0^{\alp'\bt'}) \, \ud \mu_{\alp\bt\alp'\bt'} \\
&\: + \f 12 \lim_{\lambda\to 0}  \la (2 g_0^{\alp\alp'} g_0^{\bt\bt'} - g_0^{\alp\bt} g_0^{\alp'\bt'}) \rd_{t} h_{\alp\bt}, A\Box_{g_0} h_{\alp'\bt'} \ra.
\end{split}
\end{equation}
\end{proposition}
\begin{proof}
By the Stone--Weierstrass theorem, it suffices to check the identity when $a(x,\xi)= f(x) m(\xi)$, where $f \in C^\infty_c(\mathbb R^{d+1};\mathbb R)$ and $m$ is real and positively $0$-homogeneous.

Define $A^m \in \Psi^0$ with principal symbol $m$ as above. By Lemma~\ref{lem:waveEE}, given any vector field $X$, it holds that
\begin{equation}\label{eq:wave.EE.conq}
\int_{S^*\RR^{d+1}} \{g_0^{\mu\nu} \xi_\mu\xi_\nu,  X^\rho (x) \xi_\rho m(\xi) \}  \, \ud \mu_{\alp\bt\alp'\bt'}= \lim_{\lambda\to 0} \Big( \la X h_{\alp\bt}, A^m \Box_{g_0} h_{\alp'\bt'} \ra + \la X h_{\alp'\bt'}, A^m\Box_{g_0} h_{\alp\bt} \ra \Big),
\end{equation}
where we have used that $[A^m,X] ,\, A^m - (A^m)^* \in \Psi^{-1}$.

Applying \eqref{eq:wave.EE.conq} with $X = f g_0^{\alp\alp'} g_0^{\bt\bt'} \rd_t$, we obtain
\begin{equation}
\begin{split}
&\: \int_{S^*\RR^{d+1}}  g_0^{\alp\alp'} g_0^{\bt\bt'} \{g_0^{\mu\nu} \xi_\mu\xi_\nu, \widetilde{a}(x,\xi) \}  \, \ud \mu_{\alp\bt\alp'\bt'} \\
= &\: \int_{S^*\RR^{d+1}}  g_0^{\alp\alp'} g_0^{\bt\bt'} \{g_0^{\mu\nu} \xi_\mu\xi_\nu, f(x) \xi_{t} m(\xi) \}  \, \ud \mu_{\alp\bt\alp'\bt'} \\
= &\: 2 \int_{S^*\RR^{d+1}} f g_0^{\mu\nu} \xi_\nu \xi_{t} m \rd_{x^\mu} (g_0^{\alp\alp'} g_0^{\bt\bt'}) \, \ud \mu_{\alp\bt\alp'\bt'} + 2\lim_{\lambda\to 0}  \la f g_0^{\alp\alp'} g_0^{\bt\bt'} \rd_{t} h_{\alp\bt}, A^m\Box_{g_0} h_{\alp'\bt'} \ra \\
= &\: 2 \int_{S^*\RR^{d+1}} g_0^{\mu\nu} \xi_\nu \widetilde{a} \rd_{x^\mu} (g_0^{\alp\alp'} g_0^{\bt\bt'}) \, \ud \mu_{\alp\bt\alp'\bt'} + 2\lim_{\lambda\to 0}  \la g_0^{\alp\alp'} g_0^{\bt\bt'} \rd_{t} h_{\alp\bt}, A \Box_{g_0} h_{\alp'\bt'} \ra.
\end{split}
\end{equation}
We can compute $\int_{S^*\RR^{d+1}}  g_0^{\alp\bt} g_0^{\alp'\bt'} \{g_0^{\mu\nu} \xi_\mu\xi_\nu, f \xi_t m(\xi) \}  \, \ud \mu_{\alp\bt\alp'\bt'}$ in a similar manner. Thus,
\begin{equation}
\begin{split}
 &\: \int_{S^*\RR^{d+1}}  (g_0^{\alp\alp'} g_0^{\bt\bt'} - \f 12 g_0^{\alp\bt} g_0^{\alp'\bt'}) \{g_0^{\mu\nu} \xi_\mu\xi_\nu, \widetilde{a}(x,\xi) \}  \, \ud \mu_{\alp\bt\alp'\bt'} \\
= &\:\int_{S^*\RR^{d+1}} g_0^{\mu\nu} \xi_\nu \widetilde{a}(x,\xi) \rd_{x^\mu} (2g_0^{\alp\alp'} g_0^{\bt\bt'} - g_0^{\alp\bt} g_0^{\alp'\bt'}) \, \ud \mu_{\alp\bt\alp'\bt'} \\
&\: + \lim_{\lambda\to 0}  \la (2 g_0^{\alp\alp'} g_0^{\bt\bt'} - g_0^{\alp\bt} g_0^{\alp'\bt'}) \rd_{t} h_{\alp\bt}, A \Box_{g_0} h_{\alp'\bt'} \ra.
\end{split}
\end{equation}
The desired conclusion hence follows from the definition $\mu \doteq \f 14(g_0^{\alp\alp'} g_0^{\bt\bt'} - \f 12 g_0^{\alp\bt} g_0^{\alp'\bt'}) \mu_{\alp\bt\alp'\bt'}$. \qedhere
\end{proof}

Combining the result above with Proposition~\ref{prop:h.eqn}, we obtain our main propagation identity:
\begin{proposition}\label{prop:main.reduced}
The following identity holds:
\begin{equation}\label{eq:main.reduced}
\begin{split}
&\: \int_{S^*\RR^{d+1}} \{g_0^{\mu\nu} \xi_\mu\xi_\nu, \widetilde{a}(x,\xi) \}  \, \ud \mu \\
= &\: \f 14\int_{S^*\RR^{d+1}} g_0^{\mu\nu} \xi_\nu \widetilde{a}(x,\xi)  \rd_{x^\mu} (2g_0^{\alp\alp'} g_0^{\bt\bt'} - g_0^{\alp\bt} g_0^{\alp'\bt'}) \, \ud \mu_{\alp\bt\alp'\bt'} \\
&\: + \f 12 \lim_{\lambda\to 0}  \Big\la (2 g_0^{\alp\alp'} g_0^{\bt\bt'} - g_0^{\alp\bt} g_0^{\alp'\bt'}) \rd_{t} h_{\alp\bt}, A\Big(g_0^{\mu\mu'} g_0^{\nu\nu'} h_{\mu\nu} \rd_{\mu'\nu'}^2 h_{\alp'\bt'} \Big) \Big\ra \\
&\: + \f 12 \lim_{\lambda\to 0}  \Big\la (2 g_0^{\alp\alp'} g_0^{\bt\bt'} - g_0^{\alp\bt} g_0^{\alp'\bt'}) \rd_{t} h_{\alp\bt}, A\Big(L_{\alp'\bt'}(g_0)(\rd h_n)\Big) \Big\ra \\
&\: + \f 12 \lim_{\lambda\to 0}  \Big\la (2 g_0^{\alp\alp'} g_0^{\bt\bt'} - g_0^{\alp\bt} g_0^{\alp'\bt'}) \rd_{t} h_{\alp\bt}, A\Big(Q_{\alp'\bt'}(g_0)(\rd h, \rd h)\Big) \Big\ra \\
&\: + \f 12 \lim_{\lambda\to 0}  \Big\la (2 g_0^{\alp\alp'} g_0^{\bt\bt'} - g_0^{\alp\bt} g_0^{\alp'\bt'}) \rd_{t} h_{\alp\bt}, A\Big(P_{\alp'\bt'}(g_0)(\rd h, \rd h)\Big) \Big\ra \\
&\: +  \lim_{\lambda\to 0}  \Big\la (2 g_0^{\alp\alp'} g_0^{\bt\bt'} - g_0^{\alp\bt} g_0^{\alp'\bt'}) \rd_{t} h_{\alp\bt}, A\Big((g_0)_{\rho(\alp'} \rd_{\bt')} [H^\rho - H^\rho_0]\Big) \Big\ra.
\end{split}
\end{equation}
\end{proposition}
\begin{proof}
We start with \eqref{eq:general.wave.transport} and use the equation for $\Box_{g_0} h$ from Proposition~\ref{prop:h.eqn}.

Since $g_0$, $\mathrm{Ric}(g_0)$ are smooth, it follows from integration by parts and $|h|\to 0$ that
$$ \f 12 \lim_{\lambda\to 0}  \Big\la (2 g_0^{\alp\alp'} g_0^{\bt\bt'} - g_0^{\alp\bt} g_0^{\alp'\bt'}) \rd_{t} h_{\alp\bt}, A\Big(\mathrm{Ric}_{\alp'\bt'}(g_0)\Big) \Big\ra = 0.$$
The next five terms in Proposition~\ref{prop:h.eqn} give the corresponding five terms in \eqref{eq:main.reduced}. Finally, the $O(\lambda^{\eta})$ contribution in Proposition~\ref{prop:h.eqn} vanishes in the $\lambda \to 0$ limit and do not contribute to \eqref{eq:main.reduced}. \qedhere
\end{proof}

The remainder of the paper thus involves handling the four terms on the right-hand side of \eqref{eq:main.reduced}. 

\subsection{The quasilinear term}\label{sec:quasilinear}

The main goal of this subsection is the following proposition:
\begin{proposition}\label{prop:quasilinear}
\begin{equation}\label{eq:quasilinear.main.term}
\begin{split}
&\: \lim_{\lambda\to 0}  \Big\la (2 g_0^{\alp\alp'} g_0^{\bt\bt'} - g_0^{\alp\bt} g_0^{\alp'\bt'}) \rd_{t} h_{\alp\bt}, A\Big(g_0^{\mu\mu'} g_0^{\nu\nu'} h_{\mu\nu} \rd_{\mu'\nu'}^2 h_{\alp'\bt'} \Big) \Big\ra =0.
\end{split}
\end{equation}
\end{proposition}

\subsubsection{A preliminary reduction}

We start with a preliminary observation, namely that we can replace all instances of $h$ in \eqref{eq:quasilinear.main.term} by their frequency cutoff versions. The idea is related to that in Step~3 in the proof of Proposition~\ref{prop:main.trilinear}, except now we use a spacetime Fourier cutoff.
\begin{lemma}\label{lem:cut.off.high.freq.quasilinear}
Define $\chi:[0,\infty) \to [0,1]$ to be a cutoff function supported in $\{|\xi| \leq 2\}$ and such that $\chi \equiv 1$ when on $[0,1]$. Define $h^{\hbox{\Rightscissors}}_{\alp\bt}$ so that $\calF h^{\hbox{\Rightscissors}}_{\alp\bt}(\xi) \doteq \chi(\lambda^{1.01} |\xi|)\calF h_{\alp\bt}(\xi)$. Then
\begin{equation}\label{eq:cut.off.high.freq.quasilinear}
\begin{split}
 \lim_{\lambda\to 0} &\: \Big[ \Big\la (2 g_0^{\alp\alp'} g_0^{\bt\bt'} - g_0^{\alp\bt} g_0^{\alp'\bt'}) \rd_{t} h_{\alp\bt}, A\Big(g_0^{\mu\mu'} g_0^{\nu\nu'} h_{\mu\nu} \rd_{\mu'\nu'}^2 h_{\alp'\bt'} \Big) \Big\ra \\
 &\: - \Big\la (2 g_0^{\alp\alp'} g_0^{\bt\bt'} - g_0^{\alp\bt} g_0^{\alp'\bt'}) \rd_{t} h^{\hbox{\Rightscissors}}_{\alp\bt}, A\Big(g_0^{\mu\mu'} g_0^{\nu\nu'} h_{\mu\nu} \rd_{\mu'\nu'}^2 h^{\hbox{\Rightscissors}}_{\alp'\bt'} \Big) \Big\ra \Big] = 0.
\end{split}
\end{equation}
\end{lemma}
\begin{proof}
For notational convenience, we write $\bar{h} = h - h^{\hbox{\Rightscissors}}$ in the remainder of this proof. The key estimate is that 
\begin{equation}\label{eq:h.bar.improvement}
\|\rd \bar{h} \|_{L^p} \ls \lambda^{1.01} \| \rd^2 h\|_{L^p} \ls \lambda^{1.01}\lambda^{-1} \ls \lambda^{0.01},\quad \forall p \in [1,\infty].
\end{equation}
Importantly, this is better than the estimate for $\rd h$ itself. We also have the following estimates for $\bar{h}$ and $h^{\hbox{\Rightscissors}}$, which follow easily from the definitions of the cutoffs and \eqref{eq:basic.bds.for.Ric}:
\begin{equation}\label{eq:h.bar.other}
\|\rd h^{\hbox{\Rightscissors}} \|_{L^p} \ls 1,\quad \|\rd^2 h^{\hbox{\Rightscissors}} \|_{L^p} \leq \lambda^{-1} \|\rd^2 \bar{h}\|_{L^p} \ls \lambda^{-1}.
\end{equation}

Writing $h = h^{\hbox{\Rightscissors}} + \bar{h}$ in \eqref{eq:cut.off.high.freq.quasilinear}, we need to control the following three terms:
\begin{align}
\Big\la (2 g_0^{\alp\alp'} g_0^{\bt\bt'} - g_0^{\alp\bt} g_0^{\alp'\bt'}) \rd_{t} \bar{h}_{\alp\bt}, A\Big(g_0^{\mu\mu'} g_0^{\nu\nu'} h_{\mu\nu} \rd_{\mu'\nu'}^2 h^{\hbox{\Rightscissors}}_{\alp'\bt'} \Big) \Big\ra, \label{eq:cut.off.high.freq.quasilinear.1}\\
\Big\la (2 g_0^{\alp\alp'} g_0^{\bt\bt'} - g_0^{\alp\bt} g_0^{\alp'\bt'}) \rd_{t} \bar{h}_{\alp\bt}, A\Big(g_0^{\mu\mu'} g_0^{\nu\nu'} h_{\mu\nu} \rd_{\mu'\nu'}^2 \bar{h}_{\alp'\bt'} \Big) \Big\ra.\label{eq:cut.off.high.freq.quasilinear.2}\\
\Big\la (2 g_0^{\alp\alp'} g_0^{\bt\bt'} - g_0^{\alp\bt} g_0^{\alp'\bt'}) \rd_{t} h^{\hbox{\Rightscissors}}_{\alp\bt}, A\Big(g_0^{\mu\mu'} g_0^{\nu\nu'} h_{\mu\nu} \rd_{\mu'\nu'}^2 \bar{h}_{\alp'\bt'} \Big) \Big\ra, \label{eq:cut.off.high.freq.quasilinear.3}
\end{align}

For \eqref{eq:cut.off.high.freq.quasilinear.1}, we use that $A:L^2\to L^2$ is bounded and the H\"older inequality to obtain
\begin{equation}
|\hbox{\eqref{eq:cut.off.high.freq.quasilinear.1}}|\ls \|\rd \bar{h}\|_{L^2} \| h \|_{L^\i}  \|\rd^2 h^{\hbox{\Rightscissors}} \|_{L^2} \ls \lambda^{0.01} \cdot \lambda \cdot \lambda^{-1} = \lambda^{0.01} = o(1),
\end{equation}
where we used \eqref{eq:h.bar.improvement} together with \eqref{eq:basic.bds.for.Ric} and \eqref{eq:h.bar.other}. The term \eqref{eq:cut.off.high.freq.quasilinear.2} can be treated similarly. 

For the term \eqref{eq:cut.off.high.freq.quasilinear.3}, we need to integrate by parts. First note that $A-A^*,\, [A,g_0] \in \Psi^{-1}$ and are bounded as maps $L^2\to H^{-1}$. Thus, using the bounds \eqref{eq:basic.bds.for.Ric} and \eqref{eq:h.bar.other}, we obtain
\begin{equation}\label{eq:cut.off.high.freq.quasilinear.4}
\hbox{\eqref{eq:cut.off.high.freq.quasilinear.3}} = \Big\la (2 g_0^{\alp\alp'} g_0^{\bt\bt'} - g_0^{\alp\bt} g_0^{\alp'\bt'}) A\rd_{t} h^{\hbox{\Rightscissors}}_{\alp\bt}, g_0^{\mu\mu'} g_0^{\nu\nu'} h_{\mu\nu} \rd_{\mu'\nu'}^2 \bar{h}_{\alp'\bt'} \Big\ra + o(1).
\end{equation}

We now integrate by parts the $\rd_{\mu'}$ away so as to utilize \eqref{eq:h.bar.improvement}. Notice that if $\rd_{\mu'}$ hits on $g_0$, this gives a much better term. Moreover, $[\rd,A] \in \Psi^{0}$ also gives better terms. We thus have
\begin{equation}
\begin{split}
\hbox{\eqref{eq:cut.off.high.freq.quasilinear.3}} =&\: - \Big\la (2 g_0^{\alp\alp'} g_0^{\bt\bt'} - g_0^{\alp\bt} g_0^{\alp'\bt'}) A\rd_{t} h^{\hbox{\Rightscissors}}_{\alp\bt}, g_0^{\mu\mu'} g_0^{\nu\nu'} \rd_{\mu'} h_{\mu\nu} \rd_{\nu'} \bar{h}_{\alp'\bt'} \Big\ra \\
&\: - \Big\la (2 g_0^{\alp\alp'} g_0^{\bt\bt'} - g_0^{\alp\bt} g_0^{\alp'\bt'}) A\rd^2_{\mu' t} h^{\hbox{\Rightscissors}}_{\alp\bt}, g_0^{\mu\mu'} g_0^{\nu\nu'} h_{\mu\nu} \rd_{\nu'} \bar{h}_{\alp'\bt'} \Big\ra +o(1)\\
\ls &\: \| \rd h^{\hbox{\Rightscissors}} \|_{L^2} \| \rd h\|_{L^2} \|\rd \bar{h}\|_{L^\i} + \| \rd^2 h^{\hbox{\Rightscissors}} \|_{L^2} \| h\|_{L^2} \|\rd \bar{h}\|_{L^\i} \\
\ls &\: 1\cdot 1 \cdot \lambda^{0.01} + \lambda^{-1} \cdot \lambda \cdot \lambda^{0.01} = \lambda^{0.01} = o(1),
\end{split}
\end{equation}
where as before, we used \eqref{eq:basic.bds.for.Ric}, \eqref{eq:h.bar.improvement} and \eqref{eq:h.bar.other}. \qedhere

\end{proof}

The reason that it is useful to consider $h^{\hbox{\Rightscissors}}$ instead of $h$ itself is the following lemma. Note that the bound for $\| h^{\hbox{\Rightscissors}} \|_{X^p_\lambda(g_0)}$ is no better than the bounds for $h$ given in \eqref{eq:basic.bds.for.Ric}, but the frequency cut-off gives us access to third derivatives of $h^{\hbox{\Rightscissors}}$ and obtain a bound for $\| \rd h^{\hbox{\Rightscissors}} \|_{X^\infty_\lambda(g_0)}$. This improvement will be used in the proof of Proposition~\ref{prop:quasilinear.3} below.
\begin{lemma}\label{lem:hcut.est}
$h^{\hbox{\Rightscissors}}$ satisfies
$$\| h^{\hbox{\Rightscissors}} \|_{X^p_\lambda(g_0)} \ls 1,\quad \forall p \in [1,\infty].$$
Moreover, $\rd h^{\hbox{\Rightscissors}}$ satisfies the estimates
$$\| \rd h^{\hbox{\Rightscissors}} \|_{X^\infty_\lambda(g_0)} \ls \lambda^{-1.01}.$$
\end{lemma}

\subsubsection{Setting up the Fourier decomposition}

In order to estimate the term in Proposition~\ref{prop:quasilinear}, we need to decompose $h_{\mu\nu}$ into three pieces using suitable frequency cutoff functions. The reader should think of $\mfh^{(1)}$ as ``low frequency'', $\mfh^{(2)}$ as ``spatial frequency dominated,'' and $\mfh^{(3)}$ as ``temporal frequency dominated.'' Moreover, the ``spatial frequency dominated'' part is chosen so that the frequency is supported away from the light cone.

We first fix a parameter $\de_{\Theta}$ that we will use to define the decomposition.
\begin{lemma}\label{lem:elliptic}
Given $\xi \in T^*\mathbb R^{d+1}$, denote its spatial part by $\uxi$. Then there exists $\de_{\Theta}>0$ such that 
$$|\xi_t| \leq \de_{\Theta} |\uxi| \implies g_0^{\alp\bt} \xi_\alp \xi_\bt \gtrsim |\xi|^2.$$
\end{lemma}
\begin{proof}
This is a consequence of \eqref{eq:g.smallness}. \qedhere
\end{proof}

From now on, fix $\de_{\Theta}$ so that Lemma~\ref{lem:elliptic} holds. We now introduce the decomposition of $h_{\mu\nu}$.
\begin{definition}\label{def:freq.decomposition}
Decompose
$$h_{\mu\nu} = \sum_{i=1}^3 \mfh^{(i)}_{\mu\nu},$$
where we define
\begin{align}
\mathcal F (\mfh^{(1)})(\xi) \doteq \Theta(\lambda^{-b} |\xi|)\mathcal F(h),\quad \mathcal F (\mfh^{(2)}) \doteq (1-\Theta(\lambda^{-b} |\xi|)) \Theta(\tfrac{10|\xi_t|}{\de_\Theta |\uxi|})\mathcal F(h), \label{eq:def.decomposition}
\end{align}
for $\Theta:[0,\infty] \to [0,1]$ being a smooth cutoff function such that $\Theta \equiv 1$ on $[0,1]$ and $\Theta \equiv 0$ on $[2,\infty]$, and for 
\begin{equation}\label{eq:b.range}
b \in ( \max\{\tfrac{29}{30},1-\eta\},1 )
\end{equation} 
being a fixed constant. It will be convenient to denote the corresponding projection operators by $\mathcal P^{(i)}$ so that $\mathcal P^{(i)} h = \mfh^{(i)}$.
\end{definition}

In the next three subsections, we will consider the term \eqref{eq:quasilinear.main.term} with $h_{\mu\nu}$ replaced by $\mfh^{(1)}$, $\mfh^{(2)}$ and $\mfh^{(3)}$ respectively. We will then combining the results to prove Proposition~\ref{prop:quasilinear} in Section~\ref{sec:quasilinear.everything}.

\subsubsection{The $\mfh^{(1)}$ term}

We start with the $\mfh^{(1)}$ term. The key property that we will use for $\mfh^{(1)}$ is the following
\begin{lemma}\label{lem:mfh1}
$$\|\mfh^{(1)} \|_{L^\i} \ls \lambda, \quad \| \rd \mfh^{(1)} \|_{L^\i} \ls \lambda^{1-b} \log (\tfrac{1}{\lambda}).$$
\end{lemma}
We now control the term with $\mfh^{(1)}$. The key is that since $\rd \mfh^{(1)}$ is better by Lemma~\ref{lem:mfh1}, we perform multiple integration by parts to ensure that a derivative falls on $\mfh^{(1)}$.
\begin{proposition}\label{prop:quasilinear.1}
\begin{equation}
\begin{split}
\lim_{\lambda\to 0}  \Big\la (2 g_0^{\alp\alp'} g_0^{\bt\bt'} - g_0^{\alp\bt} g_0^{\alp'\bt'}) \rd_{t} h^{\hbox{\Rightscissors}}_{\alp\bt}, A\Big(g_0^{\mu\mu'} g_0^{\nu\nu'} \mfh^{(1)}_{\mu\nu} \rd_{\mu'\nu'}^2 h^{\hbox{\Rightscissors}}_{\alp'\bt'} \Big) \Big\ra  = 0.
\end{split}
\end{equation}
\end{proposition}
\begin{proof}
First, note that since $g_0^{\mu\mu'} g_0^{\nu\nu'}$ is smooth, the commutator $[A, g_0^{\mu\mu'} g_0^{\nu\nu'}]: L^2 \to H^{-1}$ so that
\begin{equation}\label{eq:quasilinear.low.1}
\begin{split}
&\:  \Big\la (2 g_0^{\alp\alp'} g_0^{\bt\bt'} - g_0^{\alp\bt} g_0^{\alp'\bt'}) \rd_{t} h^{\hbox{\Rightscissors}}_{\alp\bt}, A\Big(g_0^{\mu\mu'} g_0^{\nu\nu'} \mfh^{(1)}_{\mu\nu} \rd_{\mu'\nu'}^2 h^{\hbox{\Rightscissors}}_{\alp'\bt'} \Big) \Big\ra \\
= &\:  \la (2 g_0^{\alp\alp'} g_0^{\bt\bt'} - g_0^{\alp\bt} g_0^{\alp'\bt'})g_0^{\mu\mu'} g_0^{\nu\nu'} A \rd_{t} h^{\hbox{\Rightscissors}}_{\alp\bt},  \mfh^{(1)}_{\mu\nu} \rd_{\mu'\nu'}^2 h^{\hbox{\Rightscissors}}_{\alp'\bt'} ) \ra  + o(1).
\end{split}
\end{equation}

We then integrate by parts in $\rd_{\mu'}$. Note that if $\rd_{\mu'}$ hits any factor of $g_0$, we then have an integral of $\mfh  \rd h \rd h$, which is $o(1)$ by our assumptions. Hence,
\begin{equation}\label{eq:quasilinear.low.2}
\begin{split}
\mbox{\eqref{eq:quasilinear.low.1}} = &\: - \la (2 g_0^{\alp\alp'} g_0^{\bt\bt'} - g_0^{\alp\bt} g_0^{\alp'\bt'})g_0^{\mu\mu'} g_0^{\nu\nu'} A \rd_{t} h^{\hbox{\Rightscissors}}_{\alp\bt},  \rd_{\mu'} \mfh^{(1)}_{\mu\nu} \rd_{\nu'} h^{\hbox{\Rightscissors}}_{\alp'\bt'}  \ra \\
&\: + \la (2 g_0^{\alp\alp'} g_0^{\bt\bt'} - g_0^{\alp\bt} g_0^{\alp'\bt'})g_0^{\mu\mu'} g_0^{\nu\nu'} [A, \rd_{\mu'}] \rd_{t} h^{\hbox{\Rightscissors}}_{\alp\bt},   \mfh^{(1)}_{\mu\nu} \rd_{\nu'} h^{\hbox{\Rightscissors}}_{\alp'\bt'}  \ra \\
&\: -\la (2 g_0^{\alp\alp'} g_0^{\bt\bt'} - g_0^{\alp\bt} g_0^{\alp'\bt'})g_0^{\mu\mu'} g_0^{\nu\nu'} A \rd^2_{t\mu'} h^{\hbox{\Rightscissors}}_{\alp\bt},   \mfh^{(1)}_{\mu\nu} \rd_{\nu'} h^{\hbox{\Rightscissors}}_{\alp'\bt'}  \ra + o(1) \doteq I + II + III + o(1).
\end{split}
\end{equation}
For term $I$, we use the estimate for $\rd \mfh^{(1)}$ in Lemma~\ref{lem:mfh1}, which is sufficient to show that $I = o(1)$. For term $II$, note that $[A,\rd_{\mu'}]$ is bounded on $L^2$, and so using $h^{\hbox{\Rightscissors}}_{\mu\nu} \to 0$ uniformly, we obtain $II = o(1)$. Hence,
\begin{equation}\label{eq:quasilinear.low.3}
I, II = o(1).
\end{equation}

Thus it remains to consider $III$ in \eqref{eq:quasilinear.low.2}. For this term, we integrate by part in $\rd_t$. As before, if $\rd_t$ hits any factor of $g_0$, the resulting term in $o(1)$. Hence,
\begin{equation}\label{eq:quasilinear.low.4}
\begin{split}
III = &\: \la (2 g_0^{\alp\alp'} g_0^{\bt\bt'} - g_0^{\alp\bt} g_0^{\alp'\bt'})g_0^{\mu\mu'} g_0^{\nu\nu'} A \rd_{\mu'} h^{\hbox{\Rightscissors}}_{\alp\bt}, \rd_t  \mfh^{(1)}_{\mu\nu} \rd_{\nu'} h^{\hbox{\Rightscissors}}_{\alp'\bt'}  \ra \\
&\: + \la (2 g_0^{\alp\alp'} g_0^{\bt\bt'} - g_0^{\alp\bt} g_0^{\alp'\bt'})g_0^{\mu\mu'} g_0^{\nu\nu'} A \rd_{\mu'} h^{\hbox{\Rightscissors}}_{\alp\bt},   \mfh^{(1)}_{\mu\nu} \rd^2_{t\nu'} h^{\hbox{\Rightscissors}}_{\alp'\bt'}  \ra +o(1) \doteq III_a + III_b + o(1).
\end{split}
\end{equation}

For $III_a$ in \eqref{eq:quasilinear.low.4}, we again use the Lemma~\ref{lem:mfh1} (as in the term $I$ in \eqref{eq:quasilinear.low.2}) to obtain smallness and get that 
\begin{equation}\label{eq:quasilinear.low.5}
III_a = o(1).
\end{equation}

As for $III_b$ in \eqref{eq:quasilinear.low.4}, we first relabel indices and then note that after commuting $[A, \mfh^{(1)}_{\mu\nu}]$, we obtain a term which is exactly $-III$.
\begin{equation}\label{eq:quasilinear.low.6}
\begin{split}
III_b = &\: \la (2 g_0^{\alp\alp'} g_0^{\bt\bt'} - g_0^{\alp\bt} g_0^{\alp'\bt'})g_0^{\mu\mu'} g_0^{\nu\nu'} A \rd_{\nu'} h^{\hbox{\Rightscissors}}_{\alp'\bt'},  \mfh^{(1)}_{\mu\nu} \rd^2_{t\mu'} h^{\hbox{\Rightscissors}}_{\alp\bt}  \ra \\
= &\: \la (2 g_0^{\alp\alp'} g_0^{\bt\bt'} - g_0^{\alp\bt} g_0^{\alp'\bt'})g_0^{\mu\mu'} g_0^{\nu\nu'}  \rd_{\nu'} h^{\hbox{\Rightscissors}}_{\alp'\bt'}, A( \mfh^{(1)}_{\mu\nu} \rd^2_{t\mu'} h^{\hbox{\Rightscissors}}_{\alp\bt} ) \ra + o(1) \\
= &\: \la (2 g_0^{\alp\alp'} g_0^{\bt\bt'} - g_0^{\alp\bt} g_0^{\alp'\bt'})g_0^{\mu\mu'} g_0^{\nu\nu'}  \rd_{\nu'} h^{\hbox{\Rightscissors}}_{\alp'\bt'}, [A, \mfh^{(1)}_{\mu\nu}] \rd^2_{t\mu'} h^{\hbox{\Rightscissors}}_{\alp\bt} ) \ra -III + o(1)
\end{split}
\end{equation}

Combining \eqref{eq:quasilinear.low.4}--\eqref{eq:quasilinear.low.6}, we thus obtain 
\begin{equation}\label{eq:quasilinear.low.7}
2III = \la (2 g_0^{\alp\alp'} g_0^{\bt\bt'} - g_0^{\alp\bt} g_0^{\alp'\bt'})g_0^{\mu\mu'} g_0^{\nu\nu'}  \rd_{\nu'} h^{\hbox{\Rightscissors}}_{\alp'\bt'}, [A, \mfh^{(1)}_{\mu\nu}] \rd^2_{t\mu'} h^{\hbox{\Rightscissors}}_{\alp\bt} ) \ra + o(1).
\end{equation}

Combining \eqref{eq:quasilinear.low.2}--\eqref{eq:quasilinear.low.3}, \eqref{eq:quasilinear.low.7} then gives
\begin{equation}\label{eq:quasilinear.low.8}
\begin{split}
&\:  \Big\la (2 g_0^{\alp\alp'} g_0^{\bt\bt'} - g_0^{\alp\bt} g_0^{\alp'\bt'}) \rd_{t} h^{\hbox{\Rightscissors}}_{\alp\bt}, A\Big(g_0^{\mu\mu'} g_0^{\nu\nu'} \mfh^{(1)}_{\mu\nu} \rd_{\mu'\nu'}^2 h^{\hbox{\Rightscissors}}_{\alp'\bt'} \Big) \Big\ra \\
= &\:  \f 12 \la (2 g_0^{\alp\alp'} g_0^{\bt\bt'} - g_0^{\alp\bt} g_0^{\alp'\bt'})g_0^{\mu\mu'} g_0^{\nu\nu'}  \rd_{\nu'} h^{\hbox{\Rightscissors}}_{\alp'\bt'}, [A, \mfh^{(1)}_{\mu\nu}] \rd^2_{t\mu'} h^{\hbox{\Rightscissors}}_{\alp\bt}  \ra  + o(1).
\end{split}
\end{equation}

Finally, observe that $[A, \mfh^{(1)}_{\mu\nu}] \rd^2_{t\mu'} h^{\hbox{\Rightscissors}}_{\alp\bt} = [A \rd_t, \mfh^{(1)}_{\mu\nu}] \rd_{\mu'} h^{\hbox{\Rightscissors}}_{\alp\bt} + o_{L^2}(1)$ using that $\| \rd \mfh^{(1)}_{\mu\nu}\|_{L^\i} = o(1)$ by Lemma~\ref{lem:mfh1}. On the other hand, the Calder\'on commutator estimate (for the commutator of a pseudo-differential operator in $\Psi^1$ and a Lipschitz function, see, for instance, \cite[Corollary, p.309]{Stein}) imply
\begin{equation}
\|[A \rd_t, \mfh^{(1)}_{\mu\nu}] \rd_{\mu'} h^{\hbox{\Rightscissors}}_{\alp\bt} \|_{L^2} \ls \| \rd h^{\hbox{\Rightscissors}}\|_{L^2} \| \mfh^{(1)} \|_{W^{1,\infty}}.
\end{equation}
Thus, using also Cauchy--Schwarz and Lemma~\ref{lem:mfh1}, we obtain
\begin{equation}
|\mbox{RHS of \eqref{eq:quasilinear.low.8}}| \ls \| \rd h^{\hbox{\Rightscissors}}\|_{L^2}^2 \| \mfh^{(1)} \|_{W^{1,\infty}} = o(1),
\end{equation}
which gives the desired conclusion. \qedhere
\end{proof}

\subsubsection{The $\mfh^{(2)}$ term}

The key property that we need for $\mfh^{(2)}$ is the following improved estimate. This can be viewed as an elliptic estimate.
\begin{lemma}\label{lem:mfh2}
$$\|  \mfh^{(2)} \|_{L^2} \ls \lambda^{2b},\quad \|\mfh^{(2)} \|_{H^1} \ls \lambda^b.$$
\end{lemma}
\begin{proof}
Using Lemma~\ref{lem:elliptic}, we can find $S_2\in \Psi^{-2}$ such that $\sigma(S_2) = \Big(g_0^{\alp\bt}\xi_\alp\xi_\bt\Big)^{-1}$ when $|\xi_t|\leq \de_\Theta |\uxi|$ and $|\xi|\geq 1$. In particular, there is an $R_1 \in \Psi^{-1}$ such that $\mfh^{(2)} = S_2 \Box_{g_0} \mfh^{(2)} + R_1 \mfh^{(2)}$ (using the Fourier support of $\mfh^{(2)}$). Therefore,
$$\|  \mfh^{(2)} \|_{H^2} \ls \| S_2 \Box_{g_0} \mfh^{(2)} \|_{H^2} + \| R_1 \mfh^{(2)} \|_{H^2} \ls \|\Box_{g_0} \mfh^{(2)} \|_{L^2} + \|\mfh^{(2)} \|_{H^1} \ls 1,$$
where in the bound, we have used that $(1-\Theta(\lambda^{-\f 12} |\xi|)) \Theta(\tfrac{2|\xi_t|}{\de_\Theta |\uxi|})$ satisfies $S^0$ symbol bounds uniformly in $\lambda$ and so $\| \Box_{g_0} \mfh^{(2)} \|_{L^2} \ls \|\Box_{g_0} h \|_{H^1} +  \| h \|_{H^1} \ls 1$. Recalling that $\mfh^{(2)}$ is defined to have frequency $\gtrsim \lambda^{-b}$, we thus have $\|  \mfh^{(2)} \|_{L^2} \ls \lambda^{2b},\quad \|\mfh^{(2)} \|_{H^1} \ls \lambda^b$. \qedhere
\end{proof}
With the bounds in Lemma~\ref{lem:mfh2}, the desired estimate in this case is almost immediate.
\begin{proposition}\label{prop:quasilinear.2}
\begin{equation}
\begin{split}
\lim_{\lambda\to 0}  \Big\la (2 g_0^{\alp\alp'} g_0^{\bt\bt'} - g_0^{\alp\bt} g_0^{\alp'\bt'}) \rd_{t} h^{\hbox{\Rightscissors}}_{\alp\bt}, A\Big(g_0^{\mu\mu'} g_0^{\nu\nu'} \mfh^{(2)}_{\mu\nu} \rd_{\mu'\nu'}^2 h^{\hbox{\Rightscissors}}_{\alp'\bt'} \Big) \Big\ra  = 0.
\end{split}
\end{equation}
\end{proposition}
\begin{proof}
By $L^{4/3}$ boundedness of $A$ and H\"older's inequality, we have 
\begin{equation}
\begin{split}
&\:  \Big\la (2 g_0^{\alp\alp'} g_0^{\bt\bt'} - g_0^{\alp\bt} g_0^{\alp'\bt'}) \rd_{t} h^{\hbox{\Rightscissors}}_{\alp\bt}, A\Big(g_0^{\mu\mu'} g_0^{\nu\nu'} \mfh^{(2)}_{\mu\nu} \rd_{\mu'\nu'}^2 h^{\hbox{\Rightscissors}}_{\alp'\bt'} \Big) \Big\ra \\
\ls &\: \| \rd h \|_{L^4} \|\mfh^{(2)} \|_{L^2} \| \rd^2 h \|_{L^4} \ls \lambda^{2b-1} = o(1). \qedhere
\end{split}
\end{equation}
\end{proof}

\subsubsection{The $\mfh^{(3)}$ term}

The property that we will use for $\mfh^{(3)}$ is captured in the following lemma. The lemma roughly says that $\mfh^{(3)}$ has a nicely behaved anti-$\rd_t$ derivative.
\begin{lemma}\label{lem:mfh3}
For every $\lambda \in (0,\lambda_0)$, there exists a ($\lambda$-dependent) $\mfk_{\mu\nu}$ such that the following holds:
\begin{enumerate}
\item \label{item:lem.mfh3.1} $\rd_t \mfk_{\mu\nu} = \mfh^{(3)}_{\mu\nu}$.
\item \label{item:lem.mfh3.2} The following estimates hold:
\begin{equation}\label{eq:mfk.basic.est}
\| \mfk\|_{L^2} \ls \lambda^{1+b},\quad \| \mfk\|_{H^1} \ls  \lambda^{b},\quad \| \Box_{g_0}\mfk\|_{L^2} \ls  \lambda^b.
\end{equation}
\item \label{item:lem.mfh3.3} The following wave coordinate condition cancellation holds:
\begin{equation}\label{eq:wave.coord.mfk}
\| H^\nu(g_0)(\rd \mfk) \|_{L^2} \ls  \max\{\lambda^{1+\f b2}, \lambda^{b+\eta}\}. 
\end{equation}
\end{enumerate}
\end{lemma}
\begin{proof}
Define $S_1 \in \Psi^{-1}$ as a Fourier multiplier operator with Fourier multiplier $\mathfrak m(\xi) \doteq \f 1{2\pi i} \xi_t^{-1} \Big(1-\Theta(2\lambda^{-b} |\xi|) \Big)\Big(1- \Theta(\tfrac{20|\xi_t|}{\de_\Theta |\uxi|})\Big)$, and for each component, let $\mfk_{\mu\nu}  \doteq S_1 \mfh^{(3)}_{\mu\nu} = S_1 \circ \mathcal P^{(3)} h_{\mu\nu}$.

Recalling the Fourier support of $\mfh^{(3)}$, it follows that $(1-\Theta(2\lambda^{-b} |\xi|))(1- \Theta(\tfrac{20|\xi_t|}{\de_\Theta |\uxi|})) = 1$ on the $\mathrm{supp}(\mathcal F \mfh^{(3)})$. From this, we obtain $\rd_t \mfk = \mfh^{(3)}$, which proves part (\ref{item:lem.mfh3.1}).

Before turning to parts (\ref{item:lem.mfh3.2}) and (\ref{item:lem.mfh3.3}), we derive some properties of $S_1$. First, by Plancherel's theorem, we gain from the high $\xi_t$ frequency to obtain
\begin{equation}\label{eq:S1.1}
\| S_1\circ \mathcal P^{(3)} \|_{H^k \to H^k} \ls_k \lambda^{b}.
\end{equation}
Second, notice that the multiplier $\mathfrak m$ satisfies the symbol bounds $|\rd_\xi^\alp \mathfrak m(\xi)| \ls_{|\alp|} (1+|\xi|)^{-|\alp|-1}$ \emph{independently} of $\lambda$. In particular, this implies by standard results on pseudo-differential operators that
\begin{equation}\label{eq:S1.2}
\| [S_1\circ \mathcal P^{(3)}, \Box_{g_0} ] \|_{L^2 \to L^2} \ls 1,\quad \| [S_1\circ \mathcal P^{(3)}, g_0^{\alp\bt} g_0^{\de\rho} \rd_\sigma] \|_{H^{-1} \to L^2} \ls 1.
\end{equation}

We now turn to the proof of (\ref{item:lem.mfh3.2}). The first two estimates in \eqref{eq:mfk.basic.est} are simple consequence of \eqref{eq:S1.1} and \eqref{eq:basic.bds.for.Ric}. To prove the third estimate in \eqref{eq:mfk.basic.est}, we note
\begin{equation}
\|\Box_{g_0} \mfk \|_{L^2} \ls \|[S_1\circ \mathcal P^{(3)}, \Box_{g_0}]\|_{L^2\to L^2} \| h \|_{L^2} +  \|S_1\circ \mathcal P^{(3)}\|_{L^2 \to L^2} \|\Box_{g_0} h \|_{L^2} \ls \min\{\lambda, \lambda^b\} = \lambda^b,
\end{equation}
where we have used \eqref{eq:S1.1}, \eqref{eq:S1.2}, \eqref{eq:basic.bds.for.Ric} and \eqref{eq:Boxh.bound}.

Finally, turning to (\ref{item:lem.mfh3.3}), we write
\begin{equation}\label{eq:mfk.almost.wave.1}
\begin{split}
&\: H^{\nu}(g_0)(\partial \mathfrak k)\\
= &\: g_0^{\nu\nu'} g_0^{\mu\mu'} (\rd_{\mu'} (\mathfrak k)_{\mu\nu'} - \f 12 \rd_{\nu'} (\mathfrak k)_{\mu\mu'}) \\
= &\: -[S_1 \circ\mathcal P^{(3)}, g_0^{\nu\nu'} g_0^{\mu\mu'} \rd_{\mu'}] h_{\mu\nu'} + \f 12  [S_1\circ \mathcal P^{(3)}, g_0^{\nu\nu'} g_0^{\mu\mu'} \rd_{\nu'} ] h_{\mu\mu'} + S_1\circ \mathcal P^{(3)} H^\nu(g_0)(\rd h) \\
\doteq &\:  I + II + III.
\end{split}
\end{equation}

$I$, $II$ are similar; we consider only $I$. Notice that the second estimate in \eqref{eq:S1.2} is by itself not sufficient to treat term $I$ in \eqref{eq:mfk.almost.wave.1} because we do not have improved bounds for $\|h\|_{H^{-1}}$ (compared to $\|h\|_{L^2}$). Instead, we decompose $h \doteq h_{\mathrm{high}} + h_{\mathrm{low}}$, where $h_{\mathrm{high}}$ and $h_{\mathrm{low}}$ have frequency support $\gtrsim \lambda^{-\f b2}$ and $\ls \lambda^{-\f b2}$, respectively. For $h_{\mathrm{high}}$, we use \eqref{eq:S1.2}, $\| h \|_{L^2} \ls \lambda$, the frequency support and Plancherel's theorem to obtain
\begin{equation}\label{eq:mfk.almost.wave.1.1.1}
\| [S_1 \circ\mathcal P^{(3)}, g_0^{\nu\nu'} g_0^{\mu\mu'} \rd_{\mu'}] (h_{\mathrm{high}})_{\mu\nu'} \|_{L^2} \ls \|h_{\mathrm{high}} \|_{H^{-1}} \ls \lambda^{\f b2} \| h\|_{L^2} \ls \lambda^{1+\f b2}.
\end{equation}
As for $h_{\mathrm{low}}$, we note that for $\lambda$ small, the frequency support implies $[S_1 \circ \mathcal P^{(3)}, g_0^{\nu\nu'} g_0^{\mu\mu'} \rd_{\mu'}] (h_{\mathrm{low}})_{\mu\nu'} =  S_1 \circ \mathcal P^{(3)}(g_0^{\nu\nu'} g_0^{\mu\mu'} \rd_{\mu'} (h_{\mathrm{low}})_{\mu\nu'})$ and thus by \eqref{eq:S1.1} and then Plancherel's theorem, we obtain
\begin{equation}\label{eq:mfk.almost.wave.1.1.2}
\begin{split}
\| [S_1 \circ\mathcal P^{(3)}, g_0^{\nu\nu'} g_0^{\mu\mu'} \rd_{\mu'}] (h_{\mathrm{low}})_{\mu\nu'} \|_{L^2} \ls &\: \| S_1 \circ \mathcal P^{(3)}(g_0^{\nu\nu'} g_0^{\mu\mu'} \rd_{\mu'} (h_{\mathrm{low}})_{\mu\nu'})\|_{L^2} \\
\ls &\: \lambda^b \| \rd h_{\mathrm{low}}\|_{L^2} \ls \lambda^{\f b2} \| h \|_{L^2} \ls \lambda^{1+\f b2}.
\end{split}
\end{equation}
Combining \eqref{eq:mfk.almost.wave.1.1.1}, \eqref{eq:mfk.almost.wave.1.1.2} for $I$ in \eqref{eq:mfk.almost.wave.1}, and handling $II$ similarly, we obtain
\begin{equation}\label{eq:mfk.almost.wave.1.1}
\| \mbox{$I$ in \eqref{eq:mfk.almost.wave.1}} \|_{L^2} + \| \mbox{$II$ in \eqref{eq:mfk.almost.wave.1}} \|_{L^2} \ls \lambda^{1+\f b2}.
\end{equation}

Finally, for $III$ in \eqref{eq:mfk.almost.wave.1}, we simply use \eqref{eq:S1.1} to obtain 
\begin{equation}\label{eq:mfk.almost.wave.1.3}
\| \mbox{$III$ in \eqref{eq:mfk.almost.wave.1}} \|_{L^2}  \ls \|S_1 \circ \mathcal P^{(3)} H^\nu(g_0)(\rd h)\|_{L^2} \ls \lambda^{b+\eta}.
\end{equation}
Combining \eqref{eq:mfk.almost.wave.1}, \eqref{eq:mfk.almost.wave.1.1} and \eqref{eq:mfk.almost.wave.1.3} yields \eqref{eq:wave.coord.mfk}. \qedhere
\end{proof}

We now use Lemma~\ref{lem:mfh3} to handle the term \eqref{eq:quasilinear.3.main.term} below. The key point is that after writing $\rd_t \mfk = \mfh^{(3)}$ (using Lemma~\ref{lem:mfh3}), we can integrate by parts and use the wave coordinate type condition in \eqref{eq:wave.coord.mfk} to reveal a null structure. The fact that the quasilinear terms have some hidden null structure is also used in \cite{LinRod}, see also \cite{IP}.
\begin{proposition}\label{prop:quasilinear.3}
\begin{equation}\label{eq:quasilinear.3.main.term}
\begin{split}
\lim_{\lambda\to 0}  \Big\la (2 g_0^{\alp\alp'} g_0^{\bt\bt'} - g_0^{\alp\bt} g_0^{\alp'\bt'}) \rd_{t} h_{\alp\bt}, A\Big(g_0^{\mu\mu'} g_0^{\nu\nu'} \mfh^{(3)}_{\mu\nu} \rd_{\mu'\nu'}^2 h^{\hbox{\Rightscissors}}_{\alp'\bt'} \Big) \Big\ra  = 0.
\end{split}
\end{equation}
\end{proposition}

\begin{proof}
Let $\mathfrak k$ be as in Lemma~\ref{lem:mfh3}. We now compute as follows:
\begin{equation}
\begin{split}
&\: \la (2 g_0^{\alp\alp'} g_0^{\bt\bt'} - g_0^{\alp\bt} g_0^{\alp'\bt'})  \rd_\gamma h^{\hbox{\Rightscissors}}_{\alp\bt}, A (g_0^{\mu\mu'} g_0^{\nu\nu'}\rd_t (\mathfrak k)_{\mu\nu} \rd^2_{\mu'\nu'} h^{\hbox{\Rightscissors}}_{\alp'\bt'} )\ra \\
= &\: \la (2 g_0^{\alp\alp'} g_0^{\bt\bt'} - g_0^{\alp\bt} g_0^{\alp'\bt'})  \rd_\gamma h^{\hbox{\Rightscissors}}_{\alp\bt}, A (g_0^{\mu\mu'} g_0^{\nu\nu'} \rd_{\mu'} (\mathfrak k)_{\mu\nu} \rd^2_{t\nu'} h^{\hbox{\Rightscissors}}_{\alp'\bt'} )\ra \\
&\:+  \la (2 g_0^{\alp\alp'} g_0^{\bt\bt'} - g_0^{\alp\bt} g_0^{\alp'\bt'})  \rd_\gamma h^{\hbox{\Rightscissors}}_{\alp\bt}, A (g_0^{\mu\mu'} g_0^{\nu\nu'} Q_{t\mu'}((\mathfrak k)_{\mu\nu}, \rd_{\nu'} h^{\hbox{\Rightscissors}}_{\alp'\bt'}) )\ra \\
= &\: \f 12 \la (2 g_0^{\alp\alp'} g_0^{\bt\bt'} - g_0^{\alp\bt} g_0^{\alp'\bt'}) \rd_\gamma h^{\hbox{\Rightscissors}}_{\alp\bt}, A (g_0^{\mu\mu'} g_0^{\nu\nu'} \rd_{\nu} (\mathfrak k)_{\mu\mu'} \rd^2_{t\nu'} h^{\hbox{\Rightscissors}}_{\alp'\bt'} )\ra \\
&\:+  \la (2 g_0^{\alp\alp'} g_0^{\bt\bt'} - g_0^{\alp\bt} g_0^{\alp'\bt'}) \rd_\gamma h^{\hbox{\Rightscissors}}_{\alp\bt}, A (g_0^{\mu\mu'} g_0^{\nu\nu'}  Q_{t\mu'}((\mathfrak k)_{\mu\nu}, \rd_{\nu'} h^{\hbox{\Rightscissors}}_{\alp'\bt'}) )\ra + o(1)\\
= &\: \f 12 \la g_0^{\mu\mu'} (2 g_0^{\alp\alp'} g_0^{\bt\bt'} - g_0^{\alp\bt} g_0^{\alp'\bt'})  \rd_\gamma A^*h^{\hbox{\Rightscissors}}_{\alp\bt},  Q_0( (\mathfrak k)_{\mu\mu'}, \rd_{t} h^{\hbox{\Rightscissors}}_{\alp'\bt'}) \ra \\
&\:+  \la g_0^{\mu\mu'} g_0^{\nu\nu'} (2 g_0^{\alp\alp'} g_0^{\bt\bt'} - g_0^{\alp\bt} g_0^{\alp'\bt'}) \rd_\gamma A^*h^{\hbox{\Rightscissors}}_{\alp\bt},  Q_{t\mu'}((\mathfrak k)_{\mu\nu}, \rd_{\nu'} h^{\hbox{\Rightscissors}}_{\alp'\bt'}) \ra + o(1),
\end{split}
\end{equation}
where in the first step we exchanged $\rd_t$ and $\rd_{\mu'}$ at the expense of a null form, in the second step we used the bound for $H(g_0)(\rd \mathfrak k)$ in \eqref{eq:wave.coord.mfk}, and in the third step we noted that the commutation of $A^*$ with $g_0^{\mu\mu'} g_0^{\alp\alp'} g_0^{\bt\bt'}  \rd_\gamma$ and $g_0^{\mu\mu'} g_0^{\nu\nu'} g_0^{\alp\alp'} g_0^{\bt\bt'} \rd_\gamma$ are in $\Psi^{-1}$.

Finally, we claim that the remaining terms are $o(1)$ after using the null form bounds in Proposition~\ref{prop:stupid.trilinear} and Proposition~\ref{prop:main.trilinear}. For the first term, note that Proposition~\ref{prop:stupid.trilinear} allows us to put $A^*h_{\alp\bt},\,\mathfrak k_{\mu\mu'} \in X_\lambda^2(g_0)$ and $\rd_{\nu'} h_{\alp'\bt'} \in X_\lambda^\infty(g_0)$ so that 
\begin{equation*}
\begin{split}
&\: \Big| \la g_0^{\mu\mu'} g_0^{\alp\alp'} g_0^{\bt\bt'}  \rd_\gamma A^*h^{\hbox{\Rightscissors}}_{\alp\bt},  Q_0( (\mathfrak k)_{\mu\mu'}, \rd_{t} h^{\hbox{\Rightscissors}}_{\alp'\bt'}) \ra  \Big| \\
\ls &\: \lambda \|\rd_{t} h^{\hbox{\Rightscissors}} \|_{X^\i_\lambda(g_0)} \| A^*h^{\hbox{\Rightscissors}} \|_{X^2_\lambda(g_0)} \| \mathfrak k \|_{X^2_\lambda(g_0)} \ls \lambda \cdot \lambda^{-1.01} \cdot 1 \cdot \lambda^b = \lambda^{b-0.01} = o(1)
\end{split}
\end{equation*}
since $b\geq \f {29}{30}$ by \eqref{eq:b.range}. In the penultimate step above, we used that $A^*:L^2\to L^2$ is a bounded operator and applied the estimates in Lemma~\ref{lem:hcut.est} and \eqref{eq:mfk.basic.est}. The second term can be treated similarly except for using Proposition~\ref{prop:main.trilinear} instead so that we have
\begin{equation*}
\begin{split}
&\: \Big| \la g_0^{\mu\mu'} g_0^{\nu\nu'} g_0^{\alp\alp'} g_0^{\bt\bt'} \rd_\gamma A^*h^{\hbox{\Rightscissors}}_{\alp\bt},  Q_{t\mu'}((\mathfrak k)_{\mu\nu}, \rd_{\nu'} h^{\hbox{\Rightscissors}}_{\alp'\bt'}) \ra  \Big| \\
\ls &\: \lambda^{\f 1{15}} \|\rd h^{\hbox{\Rightscissors}} \|_{X^\i_\lambda(g_0)} \| A^*h^{\hbox{\Rightscissors}} \|_{X^2_\lambda(g_0)} \| \mathfrak k \|_{X^2_\lambda(g_0)} \ls \lambda^{\f 1{15}} \cdot \lambda^{-1.01} \cdot 1 \cdot \lambda^b = \lambda^{b-1.01+\f 1{15}} = o(1)
\end{split}
\end{equation*}
since (by \eqref{eq:b.range}) $b - 1.01+\f 1{15} >0$. \qedhere
\end{proof}

\subsubsection{Putting everything together}\label{sec:quasilinear.everything}

\begin{proof}[Proof of Proposition~\ref{prop:quasilinear}]
This is an immediate consequence of the combination of Lemma~\ref{lem:cut.off.high.freq.quasilinear}, Proposition~\ref{prop:quasilinear.1}, Proposition~\ref{prop:quasilinear.2} and Proposition~\ref{prop:quasilinear.3}. \qedhere
\end{proof}

\subsection{The linear terms $L_{\mu\nu}(g_0)(\rd h)$}

The goal of this subsection is the following bilinear estimate:
\begin{proposition}\label{prop:linear}
\begin{equation}\label{eq:linear.main}
\begin{split}
&\: \lim_{\lambda\to 0} \f 12 \Big\la (2 g_0^{\alp\alp'} g_0^{\bt\bt'} - g_0^{\alp\bt} g_0^{\alp'\bt'}) \rd_{t} h_{\alp\bt} ,A L_{\alp'\bt'}(g_0)(\rd h) \Big\ra \\
= &\: -\int_{S^* \RR^{d+1}} \widetilde{a}\partial_{\sigma'} \Big(g_0^{\alpha \alp'} g_0^{\beta\bt'} - \f 12 g_0^{\alp\bt} g_0^{\alp'\bt'} \Big) g_0^{\sigma\sigma'} \xi_\sigma \, \ud
\mu_{\alp\bt\alp'\bt'}.
\end{split}
\end{equation}
\end{proposition}

In order to prove Proposition~\ref{prop:linear}, it is useful to first establish a lemma which relies on the wave coordinate condition.
\begin{lemma}\label{lem:linear.wave.coord.cancel}
\begin{equation}\label{eq:linear.wave.coord.cancel}
\lim_{\lambda \to 0}  \Big( \la g_0^{\alp\alp'} g_0^{\bt\bt'} \rd_t h_{\alp'\bt'}, A D^{\sigma\sigma'}_{(\alp|} \rd_{|\bt)} h_{\sigma\sigma'} \ra - \f 12 \la g_0^{\alp\alp'} g_0^{\bt\bt'} \rd_t h_{\alp\alp'}, A D^{\sigma\sigma'}_{(\bt|} \rd_{|\bt')} h_{\sigma\sigma'} \ra  \Big) = 0 .
\end{equation}
\end{lemma}
\begin{proof}
Using the definition of the microlocal defect measure, we have
$$\mbox{LHS of \eqref{eq:linear.wave.coord.cancel}} = \int_{S^*\RR^{d+1}} g_0^{\alp\alp'} g_0^{\bt\bt'} a D^{\sigma\sigma'}_{\bt} \xi_t \,\Big( \xi_\alp\ud \mu_{\alp'\bt'\sigma\sigma'} - \f 12 \xi_{\bt'}  \ud \mu_{\alp\alp'\sigma\sigma'} \Big), $$
which vanishes by Lemma~\ref{lem:wave.coord.for.MLDM}. (We remark that the exact form of $D$ does not play any role here.) \qedhere
\end{proof}

We now return to the proof of Proposition~\ref{prop:linear}:
\begin{proof}[Proof of Proposition~\ref{prop:linear}]
To simplify the notations, we write $\Gamma_{\sigma}{}^{\bt}{}_{\alp} = \Gamma_{\sigma}{}^{\bt}{}_{\alp}(g_0)$.

Recalling the definition of $L$ in \eqref{eq:L.def} and using Lemma~\ref{lem:linear.wave.coord.cancel} to handle the term involving $D$, we have
\begin{equation}\label{eq:linear.term.prelim}
\begin{split}
&\: \mbox{LHS of \eqref{eq:linear.main}} \\
= &\: \lim_{\lambda\to 0} \Big( 4 \la g_0^{\alp\alp'} g_0^{\bt\bt'}  \rd_t h_{\alp'\bt'},A g_0^{\sigma\sigma'} \Gamma_{\sigma'}{}^{\rho}{}_{(\alp|}\rd_{\sigma} h_{|\bt)\rho} \ra - 2  \la g_0^{\alp\alp'} g_0^{\bt\bt'} \rd_t h_{\alp\alp'},A g_0^{\sigma\sigma'} \Gamma_{\sigma'}{}^{\rho}{}_{(\bt|}\rd_{\sigma} h_{|\bt')\rho} \ra \Big) \doteq I + II. 
\end{split}
\end{equation}

We compute
\begin{equation}\label{eq:linear.term.I}
\begin{split}
I = \quad &\:\lim_{\lambda \to 0} 4 \la g_0^{\alp\alp'} g_0^{\bt\bt'}  \rd_t h_{\alp'\bt'},A g_0^{\sigma\sigma'} \Gamma_{\sigma'}{}^{\rho}{}_{(\alp|}\rd_{\sigma} h_{|\bt)\rho} \ra \\
\stackrel{\mathclap{\tiny\mbox{Def~\ref{def:mu.abrs}}}}{=}\quad&\:4\int_{S^* \RR^{d+1}} a \xi_t \xi_{\sigma} g_0^{\alpha \alp' }g_0^{\bt\bt'}g^{\sigma\sigma'}_0 \Gamma_{\sigma'}{}^{\rho}{}_{(\alp|} \,\ud\mu_{\alp'\beta'|\bt)\rho}\\
=\quad &\:2 \int_{S^* \RR^{d+1}}  a \xi_t \xi_{\sigma} g_0^{\alpha \alp' }g_0^{\bt\bt'}g^{\sigma\sigma'}_0 g_0^{\rho\rho'} (\partial_{\sigma'}(g_0)_{\rho'(\alpha|}
+\partial_{(\alpha|}(g_0)_{\sigma' \rho'}-\partial_{\rho'} (g_0)_{\sigma'(\alpha|}) \,\ud\mu_{\alp'\beta'|\bt)\rho} \\
\stackrel{\mathclap{\tiny\mbox{\eqref{eq:mu.obvious.symmetries}}}}{=}\quad &\:2 \int_{S^* \RR^{d+1}}  a \xi_t \xi_{\sigma} g_0^{\alpha \alp' }g_0^{\bt\bt'}g^{\sigma\sigma'}_0 g_0^{\rho\rho'} (\partial_{\sigma'}(g_0)_{\rho'(\alpha|})\,\ud\mu_{\alp'\beta'|\bt)\rho} \\
=\quad &\:-\int_{S^* \RR^{d+1}}  a \xi_t \xi_{\sigma} g_0^{\bt\bt'}g^{\sigma\sigma'}_0  (\partial_{\sigma'}g_0^{\alp\alp'})\,\ud\mu_{\alp'\bt'\alp\bt}
-\int_{S^* \RR^{d+1}}  a \xi_t \xi_{\sigma} g_0^{\alpha \alp' }g^{\sigma\sigma'}_0  (\partial_{\sigma'}g_0^{\bt\bt'})\,\ud\mu_{\alp'\beta'\alp\bt} \\
=\quad &\:-\int_{S^* \RR^{d+1}} a\partial_{\sigma'} \left(g_0^{\alpha \alp'} g_0^{\beta\bt'}\right) g_0^{\sigma\sigma'} \xi_t\xi_\sigma \, \ud
\mu_{\alp'\bt'\alp\bt}.
\end{split}
\end{equation}

We turn to term $II$. A completely analogous computation shows that
\begin{equation}\label{eq:linear.term.II}
\begin{split}
II = &\:- 2 \lim_{\lambda \to 0}  \la g_0^{\alp\alp'} g_0^{\bt\bt'} \rd_t h_{\alp\alp'},A g_0^{\sigma\sigma'} \Gamma_{\sigma'}{}^{\rho}{}_{(\bt|}\rd_{\sigma} h_{|\bt')\rho} \ra\\
=&\:- 2\int_{S^* \RR^{d+1}} a \xi_t \xi_{\sigma} g_0^{\alp\alp'} g_0^{\bt\bt'} g_0^{\sigma\sigma'}\Gamma_{\sigma'}{}^{\rho}{}_{(\bt|}\, \ud\mu_{\alp\alp' |\beta')\rho}\\
=&\:- \int_{S^* \RR^{d+1}}  a \xi_t \xi_{\sigma} g_0^{\alp\alp'} g_0^{\bt\bt'} g_0^{\sigma\sigma'} g_0^{\rho\rho'} (\partial_{\sigma'}(g_0)_{\rho'(\bt|}
+\partial_{(\bt|} (g_0)_{\sigma' \rho'}-\partial_{\rho'} (g_0)_{\sigma'(\bt|})
\, \ud\mu_{\alp\alp' |\beta')\rho}\\
=&\:-\int_{S^* \RR^{d+1}}  a \xi_t \xi_{\sigma} g_0^{\alp\alp'} g_0^{\bt\bt'} g_0^{\sigma\sigma'} g_0^{\rho\rho'} (\partial_{\sigma'}(g_0)_{\rho'(\bt|}
)
\, \ud\mu_{\alp\alp' |\beta')\rho} \\
=&\: \int_{S^* \RR^{d+1}}  a \xi_t \xi_{\sigma} g_0^{\alp\alp'}  g_0^{\sigma\sigma'}  (\partial_{\sigma'}g_0^{\bt\bt'}
)
\, \ud\mu_{\alp\alp' \bt\beta'}
=  \f 12\int_{S^* \RR^{d+1}}  a \xi_t \xi_{\sigma}   g_0^{\sigma\sigma'}  \partial_{\sigma'}(g_0^{\alp\alp'}g_0^{\bt\bt'})
\, \ud\mu_{\alp\alp' \bt\beta'}.
\end{split}
\end{equation}

Plugging \eqref{eq:linear.term.I} and \eqref{eq:linear.term.II} into \eqref{eq:linear.term.prelim} (and relabelling the indices) yields the desired conclusion. \qedhere
\end{proof}

\subsection{The quadratic terms $Q(g_0)(\rd h,\rd h)$ and $P(g_0)(\rd h,\rd h)$}

\begin{proposition}\label{prop:null.form}
\begin{equation}
\lim_{\lambda\to 0}  \Big\la (2 g_0^{\alp\alp'} g_0^{\bt\bt'} - g_0^{\alp\bt} g_0^{\alp'\bt'}) \rd_{t} h_{\alp\bt}, A\Big(Q_{\alp'\bt'}(g_0)(\rd h, \rd h)\Big) \Big\ra = 0.
\end{equation}
\end{proposition}
\begin{proof}
First using that $A-A^*,\,[A,g_0] \in \Psi^{-1}$ and $[A,\rd]\in \Psi^0$ and the usual bounds for $h$, we have
\begin{equation*}
\begin{split}
&\: \Big\la (2 g_0^{\alp\alp'} g_0^{\bt\bt'} - g_0^{\alp\bt} g_0^{\alp'\bt'}) \rd_{t} h_{\alp\bt}, A\Big(Q_{\alp'\bt'}(g_0)(\rd h, \rd h)\Big) \Big\ra \\
= &\: \Big\la (2 g_0^{\alp\alp'} g_0^{\bt\bt'} - g_0^{\alp\bt} g_0^{\alp'\bt'}) \rd_{t} A h_{\alp\bt}, \Big(Q_{\alp'\bt'}(g_0)(\rd h, \rd h)\Big) \Big\ra + o(1).
\end{split}
\end{equation*}
Recall from \eqref{eq:Q.def} that $Q(g_0)(\rd h,\rd h)$ consists of a linear combination of null forms, thus after using Proposition~\ref{prop:stupid.trilinear} and Proposition~\ref{prop:main.trilinear}, we obtain
\begin{equation*}
\begin{split}
&\:\Big| \Big\la (2 g_0^{\alp\alp'} g_0^{\bt\bt'} - g_0^{\alp\bt} g_0^{\alp'\bt'}) \rd_{t} A h_{\alp\bt}, \Big(Q_{\alp'\bt'}(g_0)(\rd h, \rd h)\Big) \Big\ra \Big|\\
\ls &\: \lambda^{\f 1{15}} \| Ah\|_{X^2_\lambda(g_0)} \| h\|_{X^2_\lambda(g_0)} \| h\|_{X^\infty_\lambda(g_0)} \ls \lambda^{\f 1{15}} \| h\|_{X^2_\lambda(g_0)}^2  \| h\|_{X^\infty_\lambda(g_0)} \ls \lambda^{\f 1{15}} = o(1).
\end{split}
\end{equation*}
\qedhere
\end{proof}

\begin{proposition}\label{prop:quadratic}
\begin{equation}
\begin{split}
\lim_{\lambda\to 0}  \Big\la (2 g_0^{\alp\alp'} g_0^{\bt\bt'} - g_0^{\alp\bt} g_0^{\alp'\bt'}) \rd_{t} h_{\alp\bt}, A\Big(P_{\alp'\bt'}(g_0)(\rd h, \rd h)\Big) \Big\ra = 0.
\end{split}
\end{equation}
\end{proposition}
\begin{proof}
We treat the first term in \eqref{eq:P.def} in detail; the other is very similar. The key point here is that after integrating by parts and using wave coordinate condition, we effectively have null condition terms:
\begin{equation*}
\begin{split}
&\: \la g_0^{\alp\alp'} g_0^{\bt\bt'} \rd_t h_{\alp'\bt'}, A ( g_0^{\rho\rho'} g_0^{\sigma\sigma'} \rd_\alp h_{\rho\sigma} \rd_\bt h_{\rho'\sigma'} )\ra \\
= &\: \la g_0^{\alp\alp'} g_0^{\bt\bt'} (\rd_t A^* h_{\alp'\bt'}),  g_0^{\rho\rho'} g_0^{\sigma\sigma'} \rd_\alp h_{\rho\sigma} \rd_\bt h_{\rho'\sigma'} \ra + o(1) \\
= &\: \la g_0^{\alp\alp'} g_0^{\bt\bt'} (\rd_\alp A^* h_{\alp'\bt'}),  g_0^{\rho\rho'} g_0^{\sigma\sigma'} \rd_t h_{\rho\sigma} \rd_\bt h_{\rho'\sigma'} \ra - \la g_0^{\alp\alp'} g_0^{\bt\bt'} \mathfrak Q_{t\alp}(A^* h_{\alp'\bt'},h_{\rho\sigma}) ,  g_0^{\rho\rho'} g_0^{\sigma\sigma'} \rd_\bt h_{\rho'\sigma'} \ra + o(1) \\
= &\: \f 12\la g_0^{\alp\alp'} \mathfrak Q^{(g_0)}_0 (A^* h_{\alp\alp'}, h_{\rho'\sigma'}),  g_0^{\rho\rho'} g_0^{\sigma\sigma'} \rd_t h_{\rho\sigma} \ra - \la g_0^{\alp\alp'} g_0^{\bt\bt'} \mathfrak Q_{t\alp}(A^* h_{\alp'\bt'},h_{\rho\sigma}) ,  g_0^{\rho\rho'} g_0^{\sigma\sigma'} \rd_\bt h_{\rho'\sigma'} \ra + o(1) \\
= &\: o(1),
\end{split}
\end{equation*}
where in the first step, we used that the commutator $[A^*, g_0^{\alp\alp'} g_0^{\bt\bt'} \rd_t]$ gives rise to a lower order term; in the second step, we swapped the $\rd_t$ and $\rd_\alp$ derivative at the expense of a null form; in the third step, we used the wave coordinate condition in \eqref{eq:basic.bds.for.H}; finally, we use the trilinear estimate in Proposition~\ref{prop:stupid.trilinear} and Proposition~\ref{prop:main.trilinear} to handle the terms involving the null condition (in a similar manner as Proposition~\ref{prop:null.form}). \qedhere
\end{proof}

\subsection{The wave gauge term}

\begin{proposition}\label{prop:wave.gauge}
	The following holds:
	$$\lim_{\lambda\to 0}  \Big\la (2 g_0^{\alp\alp'} g_0^{\bt\bt'} - g_0^{\alp\bt} g_0^{\alp'\bt'}) \rd_{t} h_{\alp\bt}, A\Big((g_0)_{\rho(\alp'} \rd_{\bt')} [H^\rho_\lambda - H^\rho_0]\Big) \Big\ra = 0.$$
\end{proposition}
\begin{proof}
We integrate by parts in $\rd_t$ and $\rd_{\bt'}$. The error terms from $[A,\rd_t]$, $[A,\rd_{\bt'}]$ and the derivatives hitting on the smooth functions all give lower order contributions. Hence,
\begin{equation}\label{eq:wave.gauge.est.before.using.wave.gauge.again}
\begin{split}
&\:  \Big\la (2 g_0^{\alp\alp'} g_0^{\bt\bt'} - g_0^{\alp\bt} g_0^{\alp'\bt'}) \rd_{t} h_{\alp\bt}, A\Big((g_0)_{\rho(\alp'} \rd_{\bt')} [H^\rho_\lambda - H^\rho_0]\Big) \Big\ra \\
= &\: \Big\la (2 g_0^{\alp\alp'} g_0^{\bt\bt'} - g_0^{\alp\bt} g_0^{\alp'\bt'}) \rd_{(\bt'|} h_{\alp\bt}, A\Big((g_0)_{|\alp')\rho} \rd_{t} [H^\rho_\lambda - H^\rho_0]\Big) \Big\ra + o(1) \\
= &\: \Big\la (2 g_0^{\bt\bt'}\rd_{\bt'} h_{\rho\bt} - g_0^{\alp\bt} \rd_{\rho} h_{\alp\bt}) , A\Big( \rd_{t} [H^\rho_\lambda - H^\rho_0]\Big) \Big\ra + o(1),
\end{split}
\end{equation}
where in the second step we used that the commutator $[A,(g_0)_{\alp'\rho}]$ or $[A,(g_0)_{\bt'\rho}]$ gives lower order contributions.

Finally, by the wave coordinate condition \eqref{eq:basic.bds.for.WC} and the bound for $\rd_t H$ in \eqref{eq:basic.bds.for.H}, we obtain
$$\Big| \hbox{\eqref{eq:wave.gauge.est.before.using.wave.gauge.again}} \Big| \ls \| g_0^{\bt\bt'}\rd_{\bt'} h_{\rho\bt} - g_0^{\alp\bt} \rd_{\rho} h_{\alp\bt} \|_{L^2} \| \rd_{t} [H^\rho_\lambda - H^\rho_0] \|_{L^2} \ls \lambda^{\eta}  = o(1),$$
which concludes the proof. \qedhere
\end{proof}

\subsection{Putting everything together}\label{sec:end.of.proof}

\begin{proof}[Proof of Theorem~\ref{thm:main.transport}]
We start with \eqref{eq:main.reduced} in Proposition~\ref{prop:main.reduced}. By Proposition~\ref{prop:linear}, the $L_{\mu\nu}(g_0)(\rd h)$ term cancels the $\f 14\int_{S^*\RR^{d+1}} g_0^{\mu\nu} \xi_\nu \widetilde{a}(x,\xi)  \rd_{x^\mu} (2g_0^{\alp\alp'} g_0^{\bt\bt'} - g_0^{\alp\bt} g_0^{\alp'\bt'}) \, \ud \mu_{\alp\bt\alp'\bt'}$ term in \eqref{eq:main.reduced}. By Propositions~\ref{prop:quasilinear}, \ref{prop:null.form}, \ref{prop:quadratic} and \ref{prop:wave.gauge}, all the remaining terms in \eqref{eq:main.reduced} vanish. This concludes the proof of Theorem~\ref{thm:main.transport}. \qedhere
\end{proof}

\appendix

\section{Proofs of algebraic properties of the Einstein equations}\label{sec:appendix}

\begin{proof}[Proof of Lemma~\ref{lem:Ric}]
To simplify notation, we will write $H^\rho = H^\rho(g)(\rd g)$.

We begin with the standard expression
\begin{equation}\label{eq:Ric.exp.1}
\mathrm{Ric}_{\mu\nu}(g) = g^{\alp\bt} \Big(\rd_\bt \Gamma_{\mu\alp\nu} - \rd_\nu \Gamma_{\mu\alp\bt} + \Gamma_{\nu\sigma\alp}\Gamma_{\mu}{}^\sigma{}_\bt - \Gamma_{\alp\sigma\bt}\Gamma_{\mu}{}^\sigma{}_\nu\Big).
\end{equation}

Note that 
\begin{equation}\label{eq:Ric.exp.2}
\begin{split}
\rd_{(\mu} (g_{\nu)\rho} H^\rho) = &\: \rd_{(\mu} \Big(g_{\nu)\rho} g^{\rho\sigma}(g^{\alp\bt} \rd_\bt g_{\sigma\alp} - \f 12 g^{\alp\bt} \rd_{\sigma} g_{\alp\bt} ) )\Big) = \rd_{(\mu|} (g^{\alp\bt} \rd_\bt g_{|\nu)\alp} - \f 12 g^{\alp\bt} \rd_{|\nu)} g_{\alp\bt} ).
\end{split}
\end{equation}
Thus,
\begin{equation}\label{eq:Ric.exp.3}
\begin{split}
&\: g^{\alp\bt} \Big(\rd_\bt \Gamma_{\mu\alp\nu} - \rd_\nu \Gamma_{\mu\alp\bt} \Big) \\
= &\: g^{\alp\bt} \Big(\f 12 \rd^2_{\bt\mu} g_{\nu\alp}  - \f 12 \rd^2_{\alp\bt} g_{\mu\nu} - \f 12 \rd^2_{\nu\mu} g_{\bt\alp} + \f 12 \rd^2_{\nu\alp} g_{\mu\bt} \Big) \\
= &\: - \f 12 \widetilde{\Box}_g g_{\mu\nu} + \rd_{(\mu}  (g_{\nu)\rho} H^\rho)- \rd_{(\mu|} g^{\alp\bt} \rd_\bt g_{|\nu)\alp} + \f 12 \rd_{(\mu|} g^{\alp\bt} \rd_{|\nu)} g_{\alp\bt} \\
= &\: - \f 12 \widetilde{\Box}_g g_{\mu\nu} + g_{\rho(\mu} \rd_{\nu)} H^\rho +  H^\rho \rd_{(\mu}  g_{\nu)\rho} + g^{\alp\bt} g^{\sigma\rho} \rd_{(\mu|} g_{\bt\rho} \rd_\sigma g_{|\nu)\alp} - \f 12 g^{\alp\bt} g^{\sigma\rho} \rd_{\mu} g_{\bt\rho} \rd_{\nu} g_{\alp\sigma}.
\end{split}
\end{equation}

As for the terms quadratic in $\Gamma$, we have
\begin{equation}\label{eq:Ric.exp.4}
\begin{split}
g^{\alp\bt} \Gamma_{\alp\sigma\bt}\Gamma_{\mu}{}^\sigma{}_\nu = g_{\sigma \rho} H^\rho \Gamma_{\mu}{}^\sigma{}_\nu = H^\rho(\rd_{(\mu}  g_{\nu)\rho} - \f 12 \rd_\rho g_{\mu\nu})
\end{split}
\end{equation}
and
\begin{equation}\label{eq:Ric.exp.5}
\begin{split}
g^{\alp\bt} \Gamma_{\nu\sigma\alp}\Gamma_{\mu}{}^\sigma{}_\bt 
= &\: g^{\alp\bt} g^{\sigma\rho} (\rd_{(\nu} g_{\alp)\sigma}  - \f 12 \rd_\sigma g_{\nu\alp})(\rd_{(\mu} g_{\bt)\rho}  - \f 12 \rd_\rho g_{\mu\bt}) \\
= &\: \f 14 g^{\alp\bt} g^{\sigma\rho} \rd_\mu g_{\alp\sigma} \rd_\nu g_{\bt\rho} + \f 12 g^{\alp\bt} g^{\sigma\rho} \rd_\sigma g_{\nu\alp} \rd_\rho g_{\mu\bt} - \f 12  g^{\alp\bt} g^{\sigma\rho} \rd_{\alp} g_{\nu\sigma} \rd_\rho g_{\mu\bt}.
\end{split}
\end{equation}
Plugging \eqref{eq:Ric.exp.3}, \eqref{eq:Ric.exp.4}, \eqref{eq:Ric.exp.5} into \eqref{eq:Ric.exp.1} yields the desired conclusion.\qedhere
\end{proof}

\begin{proof}[Proof of Lemma~\ref{lem:Ric.linear}]
Using \eqref{eq:basic.bds.for.Ric} and H\"older's inequality, we deduce
\begin{equation}
\begin{split}
&\: B_{\mu\nu}(g_\lambda)(\rd g_\lambda,\rd g_\lambda) - B_{\mu\nu}(g_0)(\rd g_0,\rd g_0) \\
= &\: B_{\mu\nu}(g_0) (\rd h, \rd g_0) + B_{\mu\nu}(g_0) (\rd g_0, \rd h) + B_{\mu\nu}(g_0) (\rd h, \rd h) + O(\lambda).
\end{split}
\end{equation}

To simplify notations, we write $g = g_0$ and $\Gamma_{\rho}{}^{\alp}{}_{\mu} = \Gamma_{\rho}{}^{\alp}{}_{\mu}(g_0)$. We then compute using \eqref{eq:Bmunu.def}:
\begin{equation}
\begin{split}
&\: B_{\mu\nu}(g_0) (\rd h, \rd g_0) + B_{\mu\nu}(g_0) (\rd g_0, \rd h)  \\
= &\: 2 g^{\alp\bt} g^{\sigma\rho} \rd_{(\mu|} g_{\bt\rho} \rd_\sigma h_{|\nu)\alp} + 2 g^{\alp\bt} g^{\sigma\rho} \rd_{(\mu|} h_{\bt\rho} \rd_\sigma g_{|\nu)\alp} \\
&\: - g^{\alp\bt} g^{\sigma\rho} \rd_{(\mu|} g_{\bt\rho} \rd_{|\nu)} h_{\alp\sigma} + 2 g^{\alp\bt} g^{\sigma\rho} \rd_\sigma g_{(\mu|\alp} \rd_\rho h_{|\nu)\bt} - 2 g^{\alp\bt} g^{\sigma\rho} \rd_{\alp} g_{(\mu|\sigma} \rd_\rho h_{|\nu)\bt} \\
= &\: 4 g^{\sigma\rho}\Gamma_{\rho}{}^{\alp}{}_{(\mu|}\rd_{\sigma} h_{|\nu)\alp} + g^{\alp\bt} g^{\sigma\rho} (2 \rd_\rho g_{\bt (\mu|} - \rd_{(\mu|} g_{\bt\rho}) \rd_{|\nu)} h_{\alp\sigma},
\end{split}
\end{equation}
which gives the conclusion. \qedhere 
\end{proof}

\begin{proof}[Proof of Lemma~\ref{lem:Ric.B}]
This computation is similar to the computation in \cite{LinRod.WN}. We need to manipulate the first term and the last term in \eqref{eq:Bmunu.def}. For the first term, using \eqref{eq:basic.bds.for.H} twice, we obtain
\begin{equation}\label{eq:p.q.derivation.1}
\begin{split}
&\: 2 g_0^{\alp\bt} g_0^{\sigma\rho} \rd_{(\mu|} h_{\bt\rho} \rd_\sigma h_{|\nu)\alp} \\
= &\: 2 g_0^{\alp\bt} g_0^{\sigma\rho} \rd_{\sigma} h_{\bt\rho} \rd_{(\mu|} h_{|\nu)\alp} + (2 g_0^{\alp\bt} g_0^{\sigma\rho} \rd_{(\mu|} h_{\bt\rho} \rd_\sigma h_{|\nu)\alp} - 2 g_0^{\alp\bt} g_0^{\sigma\rho} \rd_{\sigma} h_{\bt\rho} \rd_(\mu| h_{|\nu)\alp}) \\
\stackrel{\mathclap{\tiny\mbox{\eqref{eq:basic.bds.for.H}}}}{=} &\: g_0^{\alp\bt} g_0^{\sigma\rho} \rd_{\bt} h_{\sigma\rho} \rd_{(\mu|} h_{|\nu)\alp} + (2 g_0^{\alp\bt} g_0^{\sigma\rho} \rd_{(\mu|} h_{\bt\rho} \rd_\sigma h_{|\nu)\alp} - 2 g_0^{\alp\bt} g_0^{\sigma\rho} \rd_{\sigma} h_{\bt\rho} \rd_(\mu| h_{|\nu)\alp}) + O(\lambda^{\eta})\\
= &\: g_0^{\alp\bt} g_0^{\sigma\rho} \rd_{(\mu|} h_{\sigma\rho} \rd_{\bt} h_{|\nu)\alp} + (g_0^{\alp\bt} g_0^{\sigma\rho} \rd_{\bt} h_{\sigma\rho} \rd_{(\mu|} h_{|\nu)\alp} - g_0^{\alp\bt} g_0^{\sigma\rho} \rd_{(\mu|} h_{\sigma\rho} \rd_{\bt} h_{|\nu)\alp})\\
&\:  + (2 g_0^{\alp\bt} g_0^{\sigma\rho} \rd_{(\mu|} h_{\bt\rho} \rd_\sigma h_{|\nu)\alp} - 2 g_0^{\alp\bt} g_0^{\sigma\rho} \rd_{\sigma} h_{\bt\rho} \rd_(\mu| h_{|\nu)\alp}) + O(\lambda^{\eta})\\
\stackrel{\mathclap{\tiny\mbox{\eqref{eq:basic.bds.for.H}}}}{=} &\: \f 12 g_0^{\alp\bt} g_0^{\sigma\rho} \rd_{\mu} h_{\sigma\rho} \rd_{\nu} h_{\alp\bt} + (g_0^{\alp\bt} g_0^{\sigma\rho} \rd_{\bt} h_{\sigma\rho} \rd_{(\mu|} h_{|\nu)\alp} - g_0^{\alp\bt} g_0^{\sigma\rho} \rd_{(\mu|} h_{\sigma\rho} \rd_{\bt} h_{|\nu)\alp})\\
&\:  + (2 g_0^{\alp\bt} g_0^{\sigma\rho} \rd_{(\mu|} h_{\bt\rho} \rd_\sigma h_{|\nu)\alp} - 2 g_0^{\alp\bt} g_0^{\sigma\rho} \rd_{\sigma} h_{\bt\rho} \rd_(\mu| h_{|\nu)\alp}) + O(\lambda^{\eta}).
\end{split}
\end{equation}

In a similar manner, we obtain
\begin{equation}\label{eq:p.q.derivation.2}
\begin{split}
&\: -g_0^{\alp\bt} g_0^{\sigma\rho} \rd_{\alp} h_{\nu\sigma} \rd_\rho h_{\mu\bt} \\
= &\: -g_0^{\alp\bt} g_0^{\sigma\rho} \rd_{\rho} h_{\nu\sigma} \rd_\alp h_{\mu\bt} - (g_0^{\alp\bt} g_0^{\sigma\rho} \rd_{\alp} h_{\nu\sigma} \rd_\rho h_{\mu\bt} - g_0^{\alp\bt} g_0^{\sigma\rho} \rd_{\rho} h_{\nu\sigma} \rd_\alp h_{\mu\bt}) \\
\stackrel{\mathclap{\tiny\mbox{\eqref{eq:basic.bds.for.H}}}}{=} &\: -\f 14 g_0^{\alp\bt} g_0^{\sigma\rho} \rd_{\nu} h_{\rho\sigma} \rd_\mu h_{\alp\bt} - (g_0^{\alp\bt} g_0^{\sigma\rho} \rd_{\alp} h_{\nu\sigma} \rd_\rho h_{\mu\bt} - g_0^{\alp\bt} g_0^{\sigma\rho} \rd_{\rho} h_{\nu\sigma} \rd_\alp h_{\mu\bt}) + O(\lambda^{\eta}).
\end{split}
\end{equation}
Plugging \eqref{eq:p.q.derivation.1} and \eqref{eq:p.q.derivation.2} into \eqref{eq:Bmunu.def} and using notations in Definition~\ref{def:Q0} and Definition~\ref{def:Qmunu}, we obtain the desired decomposition. \qedhere
\end{proof}

\begin{proof}[Proof of Lemma~\ref{lem:Ric.quasilinear}]
Using \eqref{eq:tBox.def}, we compute
\begin{equation}\label{eq:quasi.derivation.1}
\widetilde{\Box}_{g_\lambda} (g_{\lambda})_{\mu\nu} - \widetilde{\Box}_{g_0} (g_0)_{\mu\nu} = (g_{\lambda}^{\alp\bt} -g_0^{\alp\bt}) \rd^2_{\alp\bt} (g_{\lambda})_{\mu\nu} + \widetilde{\Box}_{g_0} h_{\mu\nu} = (g_{\lambda}^{\alp\bt} -g_0^{\alp\bt}) \rd^2_{\alp\bt} h_{\mu\nu} + \widetilde{\Box}_{g_0} h_{\mu\nu}+O(\lambda),
\end{equation}
where we have used \eqref{eq:basic.bds.for.Ric} to control $ (g_{\lambda}^{\alp\bt} -g_0^{\alp\bt}) \rd^2_{\alp\bt} (g_0)_{\mu\nu}$. Using \eqref{eq:basic.bds.for.Ric} to compute the difference of the inverses, we obtain
\begin{equation}\label{eq:quasi.derivation.2}
g_{\lambda}^{\alp\bt} - g_0^{\alp\bt} = - g_0^{\alp\alp'}g_0^{\bt\bt'} h_{\alp'\bt'} + O(\lambda^2).
\end{equation}
Plugging \eqref{eq:quasi.derivation.2} into \eqref{eq:quasi.derivation.1} and using the second derivative bounds in \eqref{eq:basic.bds.for.Ric}, we obtain the conclusion. \qedhere
\end{proof}

\bibliographystyle{DLplain}
\bibliography{HFlimit2}

\end{document}